\documentclass[]{lmcs}

\usepackage[utf8]{inputenc}
\usepackage{amssymb}
\mathtoolsset{centercolon}

\usepackage{xcolor}
\usepackage{url}
\usepackage{subfig}

\usepackage{tikz}
\usetikzlibrary{decorations.pathmorphing}
\usetikzlibrary{decorations.pathreplacing}
\usetikzlibrary{calc}

\tikzstyle{vertex} = [circle, fill=black, draw=black, inner sep=0.8mm]

\renewcommand{\phi}{\varphi}
\renewcommand{\epsilon}{\varepsilon}
\newcommand{\nat}{\mathbb{N}}
\newcommand{\ZZ}{\mathbb{Z}}
\newcommand{\FF}{\mathbb{F}}

\newcommand{\defining}[1]{\emph{#1}}

\newcommand{\iso}{\cong}

\newcommand{\disunion}{\mathbin{\uplus}}

\newcommand{\qspace}{\mathchoice{\;}{\,}{}{}}
\newcommand{\set}[1]{\lbrace #1 \rbrace}
\DeclarePairedDelimiterX\setcond[2]{\{}{\}}{\mathchoice{\,}{}{}{}#1 \;\delimsize\vert\; #2\mathchoice{\,}{}{}{}}

\newcommand{\numVarA}{\nu}
\newcommand{\numVarB}{\mu}
\newcommand{\uniVarA}{x}
\newcommand{\uniVarB}{y}
\newcommand{\uniVarC}{z}
\newcommand{\varVec}[1]{\bar{#1}}
\newcommand{\numVarVA}{\varVec{\numVarA}}
\newcommand{\numVarVB}{\varVec{\numVarB}}
\newcommand{\uniVarVA}{\varVec{\uniVarA}}
\newcommand{\uniVarVB}{\varVec{\uniVarB}}

\newcommand{\termA}{s}
\newcommand{\termB}{t}

\newcommand{\formA}{\Phi}
\newcommand{\formB}{\Psi}

\newcommand{\interpret}{\Theta}

\newcommand{\rel}{R}

\newcommand{\auto}{\phi}
\newcommand{\autoA}{\phi}
\newcommand{\autoB}{\psi}
\newcommand{\autGroup}[1]{\mathsf{Aut}(#1)}

\newcommand{\tup}[1]{\bar{#1}}

\newcommand{\Struct}{\mathfrak{A}}
\newcommand{\StructV}{A}

\newcommand{\vertA}{u}
\newcommand{\vertB}{v}
\newcommand{\vertC}{w}

\newcommand{\spleq}{\preceq}
\newcommand{\spless}{\prec}
\newcommand{\countingExt}[1]{#1^{\#}}

\newcommand{\StructA}{\mathfrak{A}}
\newcommand{\StructB}{\mathfrak{B}}
\newcommand{\StructC}{\mathfrak{C}}
\newcommand{\StructVA}{A}
\newcommand{\StructVB}{B}

\newcommand{\reduct}[2]{#1 \upharpoonright #2}
\newcommand{\arity}[1]{\mathrm{ar}(#1)}

\newcommand{\GraphClass}{\mathcal{K}}

\newcommand{\HFsym}{\mathsf{HF}}
\newcommand{\HF}[1]{\HFsym(#1)}

\newcommand{\sig}{\tau}
\newcommand{\sigA}{\tau}
\newcommand{\sigB}{\sigma}
\newcommand{\sigOf}[1]{\mathrm{sig}(#1)}

\newcommand{\neighbors}[2]{N_{#1}(#2)}

\newcommand{\PTime}{\textsc{Ptime}}
\newcommand{\NPTime}{\textsc{NPtime}}

\newcommand{\CPT}{\ensuremath{\mathrm{CPT}}}
\newcommand{\CPTWSC}{\ensuremath{\mathrm{CPT{+}WSC}}}
\newcommand{\FO}{\ensuremath{\mathrm{FO}}}
\newcommand{\IFP}{\ensuremath{\mathrm{IFP}}}
\newcommand{\IFPSC}{\ensuremath{\mathrm{IFP{+}SC}}}
\newcommand{\IFPWSC}{\ensuremath{\mathrm{IFP{+}WSC}}}
\newcommand{\IFPWSCI}{\ensuremath{\mathrm{IFP{+}WSC{+}I}}}
\newcommand{\IFPC}{\ensuremath{\mathrm{IFPC}}}
\newcommand{\IFPCWSC}{\ensuremath{\mathrm{IFPC{+}WSC}}}
\newcommand{\IFPCWSCI}{\ensuremath{\mathrm{IFPC{+}WSC{+}I}}}
\newcommand{\ifpcwscsig}[1]{\ensuremath{\mathrm{IFPC{+}WSC}[#1]}}
\newcommand{\ifpcwscisig}[1]{\ensuremath{\mathrm{IFPC{+}WSC{+}I}[#1]}}
\newcommand{\Ck}[1]{\mathcal{C}_{#1}}

\newcommand{\CFImsym}{\mathsf{CFI}^\omega}
\newcommand{\CFIm}[1]{\CFImsym(#1)}
\newcommand{\CFImk}[2]{\mathsf{CFI}^{#1}(#2)}
\newcommand{\WSCof}[1]{\mathrm{WSC}(#1)}
\newcommand{\Iof}[1]{\mathrm{I}(#1)}
\newcommand{\WSCIof}[1]{\mathrm{WSCI}(#1)}
\newcommand{\WSCIkof}[2]{\mathrm{WSCI}^{#1}(#2)}

\newcommand{\restrictVect}[2]{#1|_{#2}}

\DeclareMathSymbol{\shortminus}{\mathbin}{AMSa}{"39}
\newcommand{\inv}[1]{#1^{\shortminus 1}}

\newcommand{\hiddenEqq}{\phantom{{}:={}}}

\newcommand{\gamebij}{\lambda}

\newcommand{\formOut}{\formA_{\text{out}}}
\newcommand{\formStep}{\formA_{\text{step}}}
\newcommand{\formChoice}{\formA_{\text{choice}}}
\newcommand{\formWit}{\formA_{\text{wit}}}

\newcommand{\ifpwscSym}{\mathsf{ifp}\text{-}\mathsf{wsc}}
\newcommand{\ifpwsc}[8]{\ifpwscSym_{#1; #2; #3; #4}.\qspace(#5,#6,#7, #8)}
\newcommand{\ifpwscbig}[8]{\ifpwscSym_{#1; #2; #3; #4}.\qspace\big(#5,#6,#7, #8\big)}
\newcommand{\ifp}{\mathsf{ifp}}

\newcommand{\stepF}[0]{f_{\text{step}}}
\newcommand{\choiceF}[0]{f_{\text{choice}}}
\newcommand{\autF}[0]{f_{\text{wit}}}
\newcommand{\wscStarSym}{\mathsf{WSC}^*}
\newcommand{\wsc}[3]{\wscStarSym(#1,#2,#3)}

\newcommand{\iop}[2]{\mathsf{I}(#1;#2)}

\newcommand{\CFIsym}{\mathsf{CFI}}
\newcommand{\CFI}[1]{\CFIsym(#1)}
\newcommand{\orig}[1]{\operatorname{orig}(#1)}
\newcommand{\bVertA}{\mathfrak{u}}
\newcommand{\bVertB}{\mathfrak{v}}
\newcommand{\bVertC}{\mathfrak{w}}
\newcommand{\bEdge}{\mathfrak{e}}

\newcommand{\tlength}[1]{| #1 |}

\newcommand{\multipede}[1]{\mathsf{MP}(#1)}

\newcommand{\gadgetsb}{V}
\newcommand{\feetb}{W}
\newcommand{\footbset}{X}
\newcommand{\footbsetA}{X}
\newcommand{\footbsetB}{Y}
\newcommand{\footbsetC}{Z}
\newcommand{\attractor}[1]{\operatorname{attr}(#1)}
\newcommand{\closure}[1]{\operatorname{cl}(#1)}
\newcommand{\gluing}[3]{#1 \cup_#2 #3}
\newcommand{\kequiv}[1]{\simeq^#1_{\mathcal{C}}}

\newcommand{\indRel}{\trianglelefteq}
\newcommand{\indClosure}[1]{#1^\indRel}

\newcommand{\Pk}[1]{\mathcal{P}_{#1}}
\newcommand{\existsP}[1]{\exists^P{#1}}
\newcommand{\partind}[2]{P_{#2}(#1)}
\newcommand{\partvert}[2]{V_{#2}(#1)}

\newcommand{\footinduced}[2]{#1[[#2]]}
\newcommand{\feetOf}[1]{F(#1)}
\newcommand{\segmentOf}[1]{S(#1)}

\newcommand{\feetrestrict}[2]{#1_{#2}}
\newcommand{\gadgetrestrict}[1]{#1_{\text{G}}}
\newcommand{\nongadgetrestrict}[1]{#1_{\text{F}}}
\newcommand{\gadgetfixed}[1]{#1_{\text{gf}}}
\newcommand{\closurefixed}[1]{#1_{\text{cf}}}

\newenvironment{claimproof}[1][\proofname]{\begin{claimprooftemp}[#1]}{\end{claimprooftemp}}


\begin{document}
	\title[Witnessed Symmetric Choice and Interpretations in IFPC]{Witnessed Symmetric Choice
		and Interpretations
		in Fixed-Point Logic with Counting}
	\titlecomment{An extended abstract of this article appeared in the Proceedings of the 50th International Colloquium on Automata, Languages, and Programming (ICALP 2023)~\cite{Lichter2023b}.}
	
	\author[Moritz Lichter]{Moritz Lichter
		\lmcsorcid{0000-0001-5437-8074}}
	\address{RWTH Aachen University, Chair of Logic and Theory of Discrete Systems, Ahornstraße 55, 52074 Aachen, Germany}
	\email{lichter@lics.rwth-aachen.de}
	\thanks{This research received funding from the European Research Council (ERC) under the European Union’s Horizon 2020 research and innovation programme (EngageS: grant agreement No.\ 820148).}
	
\begin{abstract}
	At the core of the quest for a logic for \PTime{}
	is a mismatch between algorithms making arbitrary choices
	and isomorphism-invariant logics.
	One approach to overcome this problem is
	witnessed symmetric choice.
	It allows for choices from definable orbits
	which are certified by definable witnessing automorphisms.
	
	We consider the extension of fixed-point logic with counting (\IFPC{})
	with witnessed symmetric choice (\IFPCWSC{})
	and a further extension with an interpretation operator (\IFPCWSCI{}).
	The latter operator evaluates a subformula in the structure defined by an interpretation.
	This structure  possibly has other automorphisms
	exploitable by the WSC-operator.
	For similar extensions of pure fixed-point logic (IFP),
	it is known that
	\IFPWSCI{} simulates counting which \IFPWSC{} fails to do.
	For \IFPCWSC{} it is unknown whether the interpretation operator increases expressiveness and thus allows studying the relation between WSC and interpretations beyond counting.
	
	We separate \IFPCWSC{} from \IFPCWSCI{}
	by showing that \IFPCWSC{} is not closed under \FO-interpretations.
	Additionally, we prove that nesting WSC-operators increases the expressiveness
	using the so-called CFI graphs.
	We show that if \IFPCWSCI{} canonizes a particular class of base graphs,
	then it also canonizes the corresponding CFI graphs.
	This differs from various other logics,
	where CFI graphs provide difficult instances.
\end{abstract}

\maketitle

\section{Introduction}

The quest for a logic for \PTime{} is one of the prominent open questions in finite model theory~\cite{ChandraHarel82, Grohe2008}.
It asks whether there is a logic defining exactly all polynomial-time decidable properties of finite structures.
While Fagin's theorem~\cite{Fagin74} initiated descriptive complexity theory
by showing that there is a logic capturing \NPTime{},
the question for \PTime{}  is still open. 
One problem at the core of the question is a mismatch between logics and algorithms.
For algorithms, it is common to make arbitrary choices as long as they do not affect the output.
Well-known examples of such algorithms are depth-first search on graphs,
which chooses the next neighbor of a vertex to continue, or Gaussian elimination to check an equation system for satisfiability, where the next variable (or column) to process is chosen.
While the intermediate steps of these algorithms depend on the choices made,
their answer finally does not depend on them (so whether the vertex was found or the system is satisfiable).
Designing logics to express such behavior turned out to be difficult.
On the one hand, it is undecidable whether an algorithm has this property.
In the setting of decision problems the desired property is isomorphism-invariance: for isomorphic inputs, the output is the same.
Isomorphism-invariance is usually part of the correctness proof of an algorithm.
On the other hand, every reasonable logic is required to be isomorphism-invariant by design~\cite{ebbinghaus1985, Gurevich1988}.
This means that a logic syntactically enforces isomorphism-invariance.
In contrast to algorithms, it is not possible to define something non-isomorphism-invariant in a logic.
Hence, it is in general unclear how choice-making algorithms can be expressed in a logic.

For totally ordered structures, inflationary fixed-point logic (\IFP{}) captures \PTime{} due to the Immerman-Vardi Theorem~\cite{Immerman87}. 
On ordered structures, no arbitrary choices are needed and the total order is used to ``choose'' the unique minimal element.
Thus, the ability to make choices is crucial on unordered structures.
We therefore would like to support choices in a logic while still guaranteeing isomorphism-invariance.
There are logics in which arbitrary choices can be made (e.g.~\cite{ArvindBiswas87, DawarRicherby03Nondet}),
but for these it is undecidable whether a formula is isomorphism-invariant~\cite{ArvindBiswas87} (or, in some sense, it is not decidable whether a string is a syntactically well-formed formula).
In particular, such a logic fails to be a logic capturing~\PTime{} in the sense of Gurevich~\cite{Gurevich1988}.
Similarly, when extending structures by an arbitrary order,
it is undecidable whether a formula is order-invariant, i.e., it evaluates equally for all such orders (see~\cite{GradelKLMSVVW07}).

One approach to overcome the lack of choices in logics
is to support a restricted form of choice.
Whenever we want to choose from a set of  elements related by symmetries (automorphisms) of the input, for example all vertices of a clique, then it does not matter which element is chosen.
In the logical setting,
if only choices from definable orbits are allowed,
that is from sets of definable objects  related by an automorphism of the input structure,
the output is still guaranteed to be isomorphism-invariant.
This form of choice is called \emph{symmetric choice} (SC).
However, it is unknown whether orbits can be computed in \PTime{}.
Hence, it is unknown whether a logic with symmetric choice can be evaluated in \PTime{}
because during the evaluation
it has to be verified that the choice-sets are indeed orbits.
This problem is solved by handing over the obligation to check whether the choice-sets are orbits from the evaluation to the formulas themselves.
To make a choice, not only a choice-set but also
a set of witnessing automorphisms has to be defined.
These automorphisms certify that the choice-set is indeed an orbit
in the following way: For every pair of elements $a$ and $b$ in the choice-set,
an automorphism mapping $a$ to $b$  has to be provided.
This condition can easily be checked in \PTime{}.
We call this restricted form of choice \emph{witnessed symmetric choice} (WSC).
Symmetric and witnessed symmetric choice was first studied by Gire and Hoang~\cite{GireHoang98}. The latter variant is called specified symmetric choice in this work and differs in formal details (see the paragraph on related work for more details), but their results carry over to the notion presented in this article.
Gire and Hoang extend \IFP{} with (witnessed) symmetric choice (\IFPSC{} and \IFPWSC{}) and show that \IFPWSC{} is strictly more expressive than \IFP{}.
Afterwards \IFPSC{} was studied by Dawar and Richerby~\cite{DawarRicherby03}.
They allowed for nested symmetric choice operators,
proved that parameters of choice operators increase the expressiveness,
showed that nested symmetric choice operators are more expressive
than a single one,
and conjectured that with additional nesting depth
the expressiveness increases.

Besides witnessed symmetric choice,
other operators were proposed to extend 
the expressiveness of logics not capturing \PTime{}
including a counting operator (see~\cite{Otto1997}) and an operator based on logical interpretations~\cite{GireHoang98}.
The counting operator increases the expressiveness of
\IFP{} and of the logic of Choiceless Polynomial Time (\CPT{})~\cite{BlassGS02}.
In \IFP{}, witnessed symmetric choice and counting are incomparable.
On the on hand, \IFPWSC{{}} defines the CFI query for ordered base graphs~\cite{GireHoang98},
which \IFP{} with counting (\IFPC) fails to~\cite{CaiFI1992}.
On the other hand, \IFPWSC{} fails to define the parity of the universe (indicated in~\cite{GireHoang98}),
which is possible in \IFPC{}.
However, for the combination of counting and symmetric choice not much is known.
In this article, we investigate the relation of counting, witnessed symmetric choice,
and interpretations to better understand their expressive power.

Recently, the extension of \CPT{} with witnessed symmetric choice (\CPTWSC{}) was studied by Lichter and Schweitzer~\cite{LichterSchweitzer24}.
\CPTWSC{} has the interesting property
that a \CPTWSC{}-definable isomorphism test on a class of structures implies a \CPTWSC{}-definable canonization for this class.
Canonization is the task of defining an isomorphic but totally ordered copy of the input structure.
The only requirement is that the class of structures is closed under individualization, so under assigning unique colors to vertices,
which in most cases is unproblematic~\cite{DBLP:conf/mfcs/KieferSS15, DBLP:journals/ipl/Mathon79}.
Individualization is natural in the context of choices because
a choice is, in some sense, an individualization.
The concept of canonization is essential in the quest for a logic for \PTime{}.
It provides the routinely employed approach to capture \PTime{} on a class of structures:
define canonization, obtain isomorphic and ordered structures, and apply the Immerman-Vardi Theorem (e.g.~\cite{AbuZaidGraedelGrohePakusa2014, Grohe2017, GroheN19, LichterS21}).
While \CPTWSC{} has the nice property that defining isomorphism implies canonization,
we do~not know whether the same holds for \CPT{}
or whether \CPTWSC{} is more expressive than \CPT{}.
Proving this requires separating \CPT{} from \PTime{},
which has been open for a long time.

(Witnessed) symmetric choice has the drawback that it can only choose from orbits
of the input structure.
This structure might have complicated orbits
that we cannot define or witness in the logic.
However, there could be a reduction to a different structure with easier orbits exploitable by witnessed symmetric choice.
For logics, the natural concept of a reduction is an interpretation,
i.e., defining a structure in terms of another one.
Interpretations are in some sense incompatible
with (witnessed) symmetric choice
because we always have to choose from orbits of the input structure.
To exploit a combination of choices and interpretations,
Gire and Hoang proposed an interpretation operator~\cite{GireHoang98}.
It evaluates a formula in the image of an interpretation.
For logics closed under interpretations (e.g.~\IFP{}, \IFPC{}, and \CPT{})
such an interpretation operator does not increase expressiveness.
However, for the extension with witnessed symmetric choice this is different:
\IFPWSC{} is less expressive than the extension of \IFPWSC{} with the interpretation operator.
The interpretation operator together with witnessed symmetric choice 
simulates counting.
However, it is indicated in~\cite{GireHoang98} that witnessed symmetric choice alone fails to simulate counting.

We are interested in the relation between witnessed symmetric choice and the interpretation operator not specifically for \IFP{} but more generally.
Most of the existing results in~\cite{GireHoang98,DawarRicherby03} showing that (witnessed) symmetric choice or the interpretation operator increases in some way the expressiveness of \IFP{} are based on counting.
However, counting is not the reason for using witnessed symmetric choice.
Counting can be achieved naturally in \IFPC{}.
Thus, it is unknown whether the interpretation operator increases expressiveness of \IFPCWSC{}.
In \CPT{},
it is not possible to show that witnessed symmetric choice or the interpretation operator
increase expressiveness without separating \CPT{} from \PTime{}~\cite{LichterSchweitzer24}.

Overall, a natural base logic for studying the interplay of witnessed symmetric choice
and the interpretation operator is \IFPC{}.
On the one hand, separation results based on counting are not applicable in \IFPC{}.
On the other hand, there are known \IFPC{}-undefinable \PTime{} properties, namely the already mentioned CFI query,
which can be used to separate extensions of \IFPC{}.
The CFI construction assigns to a connected graph, the so-called base graph,
two non-isomorphic CFI graphs: one is called even and the other is called odd.
The CFI query is to define whether a given CFI graph is even.

\paragraph{Results}
We define the logics \IFPCWSC{} and \IFPCWSCI{}.
They extend \IFPC{} with a fixed-point operator featuring witnessed symmetric choice
and the latter additionally with an interpretation operator.
We show that the interpretation operator increases expressiveness of the logic:
\begin{thm}
	\label{thm:wsc-le-wsci}
	$\IFPC{} < \IFPCWSC{} < \IFPCWSCI{} \leq \PTime$.
\end{thm}
In particular, this separates \IFPCWSC{} from \PTime{}.
Such a result  does not follow from existing techniques
because separating \IFPWSC{} from \PTime{} is based on counting~\cite{GireHoang98}.
Moreover, we show that \IFPCWSC{} is not even closed under one-dimensional \FO{}-interpretations.
Proving Theorem~\ref{thm:wsc-le-wsci} relies on the CFI construction.
Similarly to~\cite{LichterSchweitzer24},
we show that if \IFPCWSCI{} distinguishes orbits,
then \IFPCWSCI{} defines a canonization.
We apply this to CFI graphs:
\begin{thm}
	\label{thm:canonize-CFI if-base}
	If \IFPCWSCI{} canonizes a class of colored base graphs~$\GraphClass$ (closed under individualization),
	then \IFPCWSCI{} canonizes the class of CFI graphs $\CFI{\GraphClass}$ over~$\GraphClass$.
\end{thm}
The conclusion is that for $\IFPCWSCI$ a class of CFI graphs is not more difficult than the corresponding class of base graphs,
which is different from the case for many other logics~\cite{CaiFI1992,GradelPakusa19, Lichter2023, DawarGraedelLichter22}.
However, to canonize the CFI graphs in our proof,
the nesting depth of WSC-fixed-point operators and interpretation operators increases.
We show that this increase is unavoidable.
\begin{thm}
	\label{thm:cfi-wsci-wsci}
	There is a class of base graphs $\GraphClass$ such that
	\begin{enumerate}
		\item $\WSCIof{\IFPC{}}$ defines a canonization for~$\GraphClass$,
		\item $\WSCIof{\IFPC{}}$ does not define the CFI query for~$\CFI{\GraphClass}$, and
		\item $\WSCIof{\WSCIof{\IFPC{}}}$ defines a canonization for $\CFI{\GraphClass}$.
	\end{enumerate}
\end{thm}
Here $\WSCIof{L}$ is the fragment of $\IFPCWSCI$
using \IFPC{}-formula-formation-rules to compose $L$-formulas and 
an additional interpretation operator nested inside a WSC-fixed-point operator.
Theorem~\ref{thm:cfi-wsci-wsci} can be seen as a first step
towards an operator nesting hierarchy for $\IFPCWSCI$.

\paragraph{Our Techniques}
We adapt the techniques of~\cite{LichterSchweitzer24} from \CPT{}
to \IFPC{} to define a WSC-fixed-point operator.
It has some small but essential differences to~\cite{GireHoang98,DawarRicherby03}.
Similar to~\cite{LichterSchweitzer24} for \CPT{},
Gurevich's canonization algorithm~\cite{Gurevich97} is expressible in \IFPCWSC{}.
It suffices to distinguish orbits of a class of individualization-closed structures to define a canonization.

To prove Theorem~\ref{thm:canonize-CFI if-base},
we use the interpretation operator to show
that if \IFPCWSCI{} distinguishes orbits of the base graphs,
then \IFPCWSCI{} distinguishes also orbits of the CFI graphs
and thus canonizes the CFI graphs.
The CFI-graph-canonizing formula nests one WSC-fixed-point operator (for Gurevich's algorithm) and one interpretation operator (to distinguish orbits) more than the orbit-distinguishing formula of the base graphs.
To show that this increase in nesting depth is necessary, we construct double CFI graphs.
We start with a class of CFI graphs $\CFI{\GraphClass'}$ canonized in $\WSCIof{\IFPC}$.
We create a class of new base graphs $\GraphClass$ from the $\CFI{\GraphClass'}$-graphs.
Applying the CFI construction once more, $\CFI{\GraphClass}$ is canonized in $\WSCIof{\WSCIof{\IFPC}}$ but not in $\WSCIof{\IFPC{}}$:
To define orbits of $\CFI{\GraphClass}$,
we have to define orbits of the base graph,
for which we need to distinguish the CFI graphs $\CFI{\GraphClass'}$.

To prove $\IFPCWSC < \IFPCWSCI$,
we construct a class of asymmetric structures, i.e., structures without non-trivial automorphisms,
for which isomorphism is not \IFPC{}-definable.
Because asymmetric structures have only singleton orbits, witnessed symmetric choice is not
beneficial, thus $\IFPCWSC{} = \IFPC{}$,
and isomorphism is not \IFPCWSC{}-definable.
These structures combine CFI graphs and the so-called multipedes~\cite{GurevichShelah96},
which are asymmetric and for which \IFPC{} fails to distinguish orbits.
An interpretation removes the multipedes
and reduces the isomorphism problem to the ones of CFI graphs.
Thus,  isomorphism of this class of structures is $\IFPCWSCI{}$-definable.

\paragraph{Related Work}
The logic \IFPC{} was separated from \PTime{}
by the CFI query~\cite{CaiFI1992}.
CFI graphs not only turned out to be difficult for \IFPC{}
but variants of them were also used to separate rank logic~\cite{Lichter2023}
and the more general linear-algebraic logic~\cite{DawarGraedelLichter22} from \PTime{}.
\CPT{} was shown to define the CFI query for 
ordered base graphs~\cite{DawarRicherbyRossman2008}
and base graphs of maximal logarithmic color class size~\cite{PakusaSchalthoeferSelman2016}.
Defining the CFI query for these graphs in \CPT{}
turned out to be comparatively more complicated
than in \IFPWSC{} for ordered base graphs~\cite{GireHoang98}.
It is still open whether \CPT{} defines the CFI query for all base graphs.

The definitions of the symmetric choice operator
in~\cite{GireHoang98, DawarRicherby03}
differ at crucial points from the one in~\cite{LichterSchweitzer24} and in this article:
Witnessed symmetric choice is integrated into a fixed-point operator
where in each step (or iteration) a choice is made.
In this article, we require that the choice-set is an orbit with respect to all previous steps in the fixed-point computations (formally, we only consider automorphisms that fix all intermediate steps so far).
In the definition of (specified) symmetric choice in~\cite{GireHoang98, DawarRicherby03}
only the current step of the fixed-point iteration has to be fixed by automorphisms.
Hence, defining orbits is possibly harder with the definition of this article compared to~\cite{GireHoang98, DawarRicherby03}.
On the opposite, this stronger orbit condition
allows to give the formula defining the witnessing automorphisms
access to the obtained fixed-point (and not only to the intermediate steps, for which it has to witness orbits).
In this sense, it is easier to witness orbits with the definition of witnessed symmetric choice in this article.
So the precise formal notion of witnessed symmetric choice presented in this article and the one of specified symmetric choice of Gire and Hoang may be incomparable with regards to expressivity.
Currently, no separating example or inclusion in any direction is known.
The precise definition of witnessed symmetric choice in this paper
is essential to implement Gurevich's canonization algorithm.
The logics also differ in another subtle point dealing with reducts,
which is discussed in more detail in the end of Section~\ref{sec:witnessed-symmetric-choice}.

\CPTWSC{} in~\cite{LichterSchweitzer24} is actually a three-valued logic
using, besides true and false, an error marker for non-witnessed choices.
This is needed for \CPT{} because fixed-point computations in \CPT{} do not necessarily terminate in a polynomial number of steps. Instead, computation is aborted
(and orbits cannot be witnessed).
For \IFPC{} this problem does not occur
because fixed-points are always reached within polynomially many steps.

There are more approaches integrating choices in first-order logic:
Choice operators independent of a fixed-point operator
were studied in~\cite{BlassGurevich2000, Otto2000}.
They are no candidates to capture \PTime{}
because of nondeterminism, undecidable syntax, or too high complexity.
A similar statement holds for the nondeterministic version of the fixed-point operator with choice, where choices can be made from arbitrary choice-sets and not only of orbits~\cite{DawarRicherby03Nondet}.
For a more detailed overview, we refer to~\cite{RicherbyThesis2004}.

Multipedes~\cite{GurevichShelah96} are asymmetric structures,
which are not characterized up to isomorphism in $k$-variable counting logic for every fixed number of variables $k$.
Asymmetry turns multipedes to hard instances for graph isomorphism algorithms
in the individualization-refinement framework~\cite{NeuenSchweitzer18, AndersSchweitzer21}.
The size of a multipede not identifiable in $k$-variable counting logic is large compared to $k$.
There are asymmetric graphs
with similar properties, but whose order is linear in~$k$~\cite{DawarKhan19}.
Both constructions are based on the~CFI~construction.

There is another remarkable but not directly connected coincidence to
lengths of resolution proofs.
Resolution proofs for non-isomorphism of CFI-graphs
have exponential size~\cite{Toran13}.
When adding a global symmetry rule (SRC-I),
which exploits automorphisms of the formula
(so akin to symmetric choice),
the length becomes polynomial~\cite{SchweitzerSeebach21}.
For asymmetric multipedes  the length in the SRC-I system
is still exponential~\cite{ToranWoerz23}.
But when considering the local symmetry rule (SRC-II)
exploiting local automorphisms
(so somewhat akin to symmetric choice after restricting to a substructure with an interpretation)
the length becomes polynomial again~\cite{SchweitzerSeebach21}.

\paragraph{Structure of this Article}
Section~\ref{sec:prelimiaries} reviews preliminaries including the logic \IFPC{} and logical interpretations.
Next, Section~\ref{sec:witnessed-symmetric-choice} introduces
the logics \IFPCWSC{} and \IFPCWSCI{}.
Section~\ref{sec:cfi-construction} reviews the CFI construction
and Section~\ref{sec:canonize-cfi-graphs} considers their canonization (Theorem~\ref{thm:canonize-CFI if-base}).
Section~\ref{sec:nesting-of-operators} proves that the increase in the operator nesting depth is unavoidable to canonize CFI graphs (Theorem~\ref{thm:cfi-wsci-wsci}).
Finally, Section~\ref{sec:separating-wsc-wsci} separates \IFPCWSC{} from \IFPCWSCI{} (Theorem~\ref{thm:wsc-le-wsci}).
Section~\ref{sec:discussion} concludes with a discussion of the results.

\section{Preliminaries}
\label{sec:prelimiaries}

For a number $k\in \nat$, we set $[k] := \set{1,\dots, k}$.
For some set~$N$, the $i$-th entry of a $k$-tuple $\tup{t} \in  N^k$ is denoted by~$t_i$ and its length by $|\tup{t}|=k$.
The set of all tuples of length at most~$k$ is~$N^{\leq k}$
and the set of all tuples of finite length is~$N^*$.

A \defining{relational signature} is a set of relation symbols $\set{\rel_1,\dots, \rel_\ell}$ with associated \defining{arities} $\arity{\rel_i} \in \nat$. We use letters $\sigA$ and $\sigB$ for signatures.
Let $\sig = \set{\rel_1,\dots, \rel_\ell}$ be a signature.
A~\defining{$\sig$\nobreakdash-structure} is a tuple $\StructA = (\StructVA, \rel_1^\Struct, \dots, \rel_\ell^\Struct)$  where $\rel_i^\Struct \subseteq \StructVA^{\arity{\rel_i}}$ for every $i \in [\ell]$.
The set~$\StructVA$ is called the \defining{universe} of $\StructA$ and its elements are called \defining{vertices}.
We use Fraktur letters~$\StructA$ and~$\StructB$ for relational structures (except graphs)
and denote their universes always by~$\StructVA$ and~$\StructVB$, respectively.
For vertices, we use the letters~$\vertA$,~$\vertB$, and~$\vertC$.
For $\sigB \subseteq\sig$, the \defining{reduct} $\reduct{\StructA}{\sigB}$
is the restriction of~$\StructA$ to the relations contained in~$\sigB$.
For a subset $W \subseteq \StructVA$,
we denote by $\StructA[W]$ the \defining{substructure
of~$\StructA$ induced by~$W$}.
The structure $\StructA[W]$ has universe $W$
and relations $\rel^{\StructA[W]} := \rel^\StructA \cap W^{\arity{\rel}}$ for all $\rel \in \sig$.
We sometimes also view a $k$-tuple $\tup{\vertA} \in \StructV^k$ as a set
and write $\StructA[\tup{\vertA}]$ for $\StructA[\setcond{\vertA_i}{i \in [k]}]$.
In this article we consider finite structures.

A \defining{colored graph} is an $\set{E, \spleq}$-structure $G=(V,E^G,\spleq^G)$.
The binary relation $E$ is the edge relation and the binary relation $\spleq$
is a total preorder.
Its equivalence classes are the \defining{color classes} or just \defining{colors}.
We usually just write $G=(V,E,\spleq)$ for a colored graph.
This article considers undirected graphs without loops (so $E$ is symmetric and does not contain pairs of the form $(\vertA,\vertA)$.
The \defining{neighborhood} of a vertex $\vertA \in V$ in $G$ is $\neighbors{G}{\vertA} := \setcond{\vertB \in V}{(\vertA,\vertB) \in E}$.
For a subset $W \subseteq V$ of vertices of~$G$, the \defining{subgraph of $G$ induced by $W$} is $G[W]$.
The graph $G$ is \defining{$k$-connected}
if $|V| > k$ and, for every $V' \subseteq V$ of size at most $k-1$,
the graph $G \setminus V'$ obtained from~$G$ by deleting all vertices in~$V'$ is connected.
The \defining{treewidth} of a graph measures how close a graph is to being a tree (see e.g.~\cite{DawarRicherby07}). We omit a formal definition here and only use the following fact:
If a graph~$G$ is a \defining{minor} of~$H$,
so~$G$ can be obtained from~$H$ by deleting vertices, deleting edges, and contracting edges,
then the treewidth of~$G$ is at most the treewidth of~$H$.

For two $\sig$-structures $\StructA$ and $\StructB$,
an \defining{isomorphism} $\autoA \colon \StructA \to \StructB$
is a bijection $\StructVA \to \StructVB$ 
such that~$\tup{\vertA} \in \rel^\StructA$
if and only if $\autoA(\tup{\vertA})=\left(\autoA(\vertA_1),\dots,\autoA(\vertA_{\arity{\rel}}\right) \in \rel^\StructB$
for all~$\rel \in \sig$ and all~$\tup{\vertA} \in \StructVA^{\arity{\rel}}$.
For $\tup{\vertA}\in \StructVA^k$ and $\tup{\vertB} \in \StructVB^k$,
the structures $(\StructA, \tup{\vertA})$ and $(\StructB, \tup{\vertB})$
are isomorphic, denoted ${(\StructA, \tup{\vertA})\iso (\StructB, \tup{\vertB})}$,
if there is an isomorphism $\autoA\colon\StructA\to\StructB$ satisfying $\autoA(\tup{\vertA}) = \tup{\vertB}$.
An \defining{automorphism}~$\autoA$ of $(\StructA, \tup{\vertA})$ is an isomorphism $(\StructA, \tup{\vertA})\to (\StructA, \tup{\vertA})$.
We say that~$\autoA$ \defining{fixes}~$\tup{\vertA}$
and write $\autGroup{(\StructA, \tup{\vertA})}$
for the set of all automorphisms fixing $\tup{\vertA}$.
We will use the same notation also for other objects than tuples,
e.g.,~for automorphisms fixing relations.
A \defining{$k$-orbit} of $(\StructA, \tup{\vertA})$ is a maximal set of $k$-tuples $O \subseteq \StructVA^k$
such that for every $\tup{\vertB},\tup{\vertC} \in O$
there is an automorphism~$\autoA\in \autGroup{(\StructA, \tup{\vertA})}$
satisfying $\autoA(\tup{\vertB})  = (\autoA(\vertB_1),\dots,\autoA(\vertB_k))= \tup{\vertC}$.

\paragraph{Fixed-Point Logic with Counting}

\newcommand{\eval}[2]{#2^{#1}}
We recall fixed-point logic with counting \IFPC{} (proposed in~\cite{Immerman87b}, also see~\cite{Otto1997}).
Let $\sig$ be a signature
and $\Struct = (\StructV, \rel_1^\Struct, \dots, \rel_\ell^\Struct)$ be a $\sig$\nobreakdash-structure.
We extend $\sig$ and $\Struct$ by a numeric sort for natural numbers.
Define $\countingExt{\sig} := \sig \disunion \{\cdot, +, 0,1\}$
and
$\countingExt{\Struct} := (\StructV, \rel_1^\Struct, \dots, \rel_\ell^\Struct, \nat, \cdot, +, 0, 1)$
to be the two-sorted $\countingExt{\sig}$\nobreakdash-structure
that is the disjoint union of $\Struct$ and $\nat$.

$\IFPC{}[\sig]$ is a two-sorted logic using the signature $\countingExt{\sig}$.
\defining{Element variables} range over the vertices and \defining{numeric variables} range over the natural numbers.
For element variables we use the letters~$\uniVarA$, $\uniVarB$, and $\uniVarC$,
for numeric variables the Greek letters~$\numVarA$ and~$\numVarB$,
and for numeric terms the letters~$\termA$ and~$\termB$.
\IFPC{}-formulas are built from first-order formulas and a fixed-point operator.
\IFPC{}-terms (or numeric terms) are build from counting terms, the constants $0$ and $1$, and the numeric function symbols $\cdot$ and $+$.
A term or formula is closed if it has no free variables.
When quantifying over numeric variables,
their range needs to be bounded to ensure \PTime{}-evaluation:
For an \IFPC{}-formula~$\formA$,
a closed numeric \IFPC{}-term~$\termA$,
a numeric variable~$\numVarA$ possibly free in~$\formA$,
and a quantifier $Q \in\set{\forall, \exists}$,  the formula
\[Q \numVarA \leq \termA.\qspace\formA\]
is an \IFPC{}-formula.
An \defining{inflationary fixed-point operator} defines a relation $\rel$.
We allow $\rel$ to mix vertices with numbers.
For an $\IFPC{}[\sig,\rel]$-formula $\formA$
and variables $\uniVarVA\numVarVB$ possibly free in $\formA$,
the fixed-point operator
\[\left[\ifp \rel \uniVarVA\numVarVB \leq \tup{\termA}.\qspace\formA \right](\uniVarVA\numVarVB)\]
is an $\IFPC{}[\sig]$-formula. Here, $\tup{\termA}$ is a  $|\numVarVB|$-tuple of closed numeric terms which bounds the values of $\numVarVB$
similar to the case of a quantifier.
The crucial element of \IFPC{} is \defining{counting terms}.
They count the number of tuples satisfying a formula.
Let $\formA$ be an \IFPC{}-formula
with possibly free variables~$\uniVarVA$ and~$\numVarVA$
and let $\tup{\termA}$ be a $|\numVarVA|$-tuple of closed  numeric \IFPC{}-terms. Then
\[ \#\uniVarVA\numVarVA \leq \tup{\termA}.\qspace\formA \]
is a numeric \IFPC{}-term.

\IFPC{}-formulas (or terms) are evaluated over $\countingExt{\Struct}$.
For a numeric term $\termA(\uniVarVA\numVarVA)$, we denote by
$\eval{\Struct}{\termA} \colon \StructV^{|\uniVarVA|} \times \nat^{|\numVarVA|} \to \nat$
the function mapping the values for the free variables of $\termA$
to the value that $\termA$ takes in $\countingExt{\Struct}$.
Likewise, for a formula $\formA(\uniVarVA\numVarVA)$,
we write $\eval{\Struct}{\formA} \subseteq \StructV^{|\uniVarVA|} \times \nat^{|\numVarVA|}$ for the set of values for the free variables
satisfying $\formA$.
A formula $Qv \leq s.\qspace \formA$ as above is satisfied if there is a number at most the upper bound defined by $s$ satisfying $\formA$.
Formally, let the free universe variables of $\formA$ be $\uniVarVA$
and its free numeric variables be $\numVarVB\numVarA$.
Then
\[\eval{\Struct}{(Qv \leq s.\qspace \formA)} := \setcond*{\tup{\vertA}\tup{n} \in A^{|\uniVarVA|} \times \nat^{|\numVarVB|} } { \tup{\vertA}\tup{n}m \in \eval{\Struct}{\formA}\text{ for some/all } m \in \set{0,\dots, s^\StructA}},\]
where it depends on $Q \in \set{\exists,\forall}$ whether we consider some or all such $m$.
Let $\formA(\uniVarVB\uniVarVA\numVarVB\numVarVA)$ be an $\IFPC{}[\sig,R]$-formula and 
$\tup{\termA}$ be a $|\numVarVA|$-tuple  of closed numeric terms.
Evaluation of the inflationary fixed-point $\left[\ifp \rel \uniVarVA\numVarVA \leq \tup{\termA}.\qspace\formA \right]\uniVarVA\numVarVA$
starts with $\rel$ being empty.
Iteratively all tuples satisfying $\formA$ are added to $\rel$ until this process stabilizes.
The numeric part of these tuples is limited by the bounds defined by $\tup{\termA}$.
Formally, for values of the free variables
$\tup{\vertA} \in \StructV^{|\uniVarVB|}$ and $\tup{m} \in \nat^{|\numVarVB|}$
we inductively define a series of relations~$\rel_i^\Struct$
called stages via
\begin{align*}
	\rel^\Struct_0 &:= \emptyset,\\
	\rel^\Struct_{i+1} &:= \rel_i^\Struct \cup 
	\setcond*{\tup{\vertC}\tup{n} \in \StructV^{|\uniVarVA|} \times \nat^{|\numVarVA|}}{n_i \leq \termA_i^\StructA \text{ for all } i \in [|\numVarVA|] , \tup{\vertA}\tup{\vertC}\tup{m}\tup{n} \in \eval{(\Struct, \rel_i^\Struct)}{\formA}},
\end{align*}
where $(\Struct, \rel_i^\Struct)$ denotes the $(\sig \cup \set{\rel})$-structure
obtained from extending~$\Struct$ with~$\rel_i^\Struct$.
By definition, $\rel_i^\Struct \subseteq \rel_{i+1}^\Struct \subseteq \StructV^{|\uniVarVA|} \times \set{0,...,\termA_1^\Struct} \times \cdots \times \set{0,...,\termA_{|\numVarVA|}^\Struct}$ for every $i \in \nat$.
Because \IFPC{}-terms always evaluate to a number polynomial in the size of the input structure,
the series stabilizes after a polynomial number of steps, i.e, $\rel_\ell^\Struct = \rel_{\ell+1}^\Struct =: \rel^\Struct_{\tup{\vertA}\tup{m}}$ for some $\ell \in \nat$.
The fixed-point operator evaluates as follows:
\[ \eval{\Struct}{\big(\left[\ifp \rel \uniVarVA\numVarVA \leq \tup{\termA}.\qspace\formA \right](\uniVarVA\numVarVA)\big)} := \setcond*{\tup{\vertA}\tup{\vertB}\tup{m}\tup{n}}{ \tup{\vertB}\tup{n} \in  \rel_{\tup{\vertA} \tup{m}}^\Struct}.\]
Counting terms evaluate to the number of satisfying tuples.
For a formula $\formA(\uniVarVB\uniVarVA\numVarVB\numVarVA)$
and a $|\numVarVA|$-tuple $\tup{\termA}$ of closed numeric terms,
the evaluation of a counting term  is defined as follows:
\[\eval{\Struct}{(\#\uniVarVA\numVarVA \leq \tup{\termA}.\qspace\formA)}(\tup{\vertA}\tup{m}) := \left|\setcond*{\tup{\vertC}\tup{n} \in \StructV^{|\uniVarVA|} \times \nat^{|\numVarVA|}}{n_i \leq \termA_i^\StructA \text{ for all } i \in [|\numVarVA|] , \tup{\vertA}\tup{\vertC}\tup{m}\tup{n} \in \eval{\Struct}{\formA}}\right|.\]

\paragraph{Finite Variable Counting Logic}
The $k$-variable first-order logic with counting $\Ck{k}$ 
extends the $k$-variable fragment of first-order logic (\FO{}) with counting quantifiers $\exists^{\geq j} \uniVarA.\qspace\formA$
stating that at least $j$ distinct vertices satisfy $\formA$ (see~\cite{Otto1997}).
Bounded variable logics with counting are a useful tool to prove
\IFPC{}-undefinability.
For every $n \in \nat$,
every \IFPC{}-formula using~$k$ variables
is equivalent on structures of order up to~$n$
to a $\Ck{\mathcal{O}(k)}$-formula.

Let~$\StructA$ and~$\StructB$ be two $\sig$-structures and $\tup{\vertA} \in \StructVA^\ell$ and $\tup{\vertB}\in\StructVB^\ell$.
We say that a logic~$L$ \defining{distinguishes} $(\StructA,\tup{\vertA})$
from $(\StructB,\tup{\vertB})$ if there is an $L$-formula~$\formA$ with~$\ell$ free variables such that ${\tup{\vertA} \in \formA^\StructA}$ and ${\tup{\vertB}\notin \formA^\StructB}$. 
Otherwise, the structures are called \defining{$L$-equivalent}.
We write $(\StructA,\tup{\vertA})\kequiv{k} (\StructB,\tup{\vertB})$
if $(\StructA,\tup{\vertA})$ and $(\StructB,\tup{\vertB})$ are $\Ck{k}$-equivalent.
The logics~$\Ck{k}$ are used to prove \IFPC{}-undefinability as follows:
Let $(\StructA_k, \StructB_k)$ be a sequence of finite structures for every $k \in \nat$
such that~$\StructA_k$ has a property~$P$ but~$\StructB_k$ does not.
If $\StructA_k \kequiv{k} \StructB_k$ for every~$k$, then \IFPC{} does not define~$P$.

The logic $\Ck{k}$ can be characterized by an Ehrenfeucht-Fraïssé-like pebble game -- the \defining{bijective $k$-pebble game}~\cite{Hella96}.
The game is played on two structures~$\StructA$ and~$\StructB$
by two players called Spoiler and Duplicator.
There are~$k$ pebble pairs $(p_i,q_i)$, one for each $i \in [k]$.
Positions in the game are tuples $(\StructA, \tup{\vertA}; \StructB, \tup{\vertB})$
for tuples $\tup{\vertA}\in \StructVA^{\leq k}$ and
$\tup{\vertB} \in \StructVB^{\leq k}$ of the same length.
For every $i \in |\tup{\vertA}|$,
a pebble~$p_j$ is placed on the atom~$\vertA_i$
and the pebble~$q_j$ is placed on~$\vertB_i$ from some $j \in[k]$.
It will not matter which pebble pair $(p_j,q_j)$ is used for the $i$-th entries
of~$\tup{\vertA}$ and~$\tup{\vertB}$.
The game proceeds as follows.
If $|\StructVA| \neq |\StructVB|$, then Spoiler wins.
Otherwise,
Spoiler picks up a pair of pebbles $(p_i, q_i)$ (may it be already placed on the structures are not).
Duplicator answers with a bijection $\gamebij \colon \StructVA \to \StructVB$.
Spoiler places the pebble~$p_i$ on a vertex $\vertC \in \StructVA$
and~$q_i$ on $\gamebij(\vertC) \in \StructVB$.
If in the resulting position $(\StructA, \tup{\vertA}; \StructB, \tup{\vertB})$
there is no \defining{pebble-respecting} local isomorphism,
that is, the map defined via $\vertA_i \mapsto \vertB_i$
is not an isomorphism $(\StructA[\tup{\vertA}],\tup{\vertA}) \to (\StructB[\tup{\vertB}],\tup{\vertB})$,
then Spoiler wins.
Otherwise, the game continues with the next round.
Duplicator wins if Spoiler never wins.
We say that Spoiler or Duplicator has a \defining{winning strategy} in
position $(\StructA, \tup{\vertA}; \StructB, \tup{\vertB})$
if Spoiler or Duplicator, respectively, can always win the game  regardless of the moves of the other player.

For all finite $\sig$-structures $\StructA$ and $\StructB$
and all tuples $\tup{\vertA} \in \StructVA^{\leq k}$ and $\tup{\vertB} \in \StructVB^{\leq k}$,
Spoiler has a winning strategy in the bijective $k$-pebble game
in position $(\StructA, \tup{\vertA}; \StructB, \tup{\vertB})$
if and only if  $(\StructA,\tup{\vertA})\not\kequiv{k} (\StructB,\tup{\vertB})$~\cite{Hella96}.

\paragraph{Logical Interpretations}
A logical interpretation is the logical correspondence to an algorithmic reduction.
It transforms a relational structure to another one.
An $\IFPC[\sigA,\sigB]$-interpretation
defines a partial map from $\sigA$-structures to $\sigB$-structures,
where the map is defined in terms of \IFPC{}-formulas
operating on tuples of the input $\sigA$\nobreakdash-structure.
In the case of \IFPC{}, these tuples not only contain vertices
but also numbers.
For the sake of readability,
we use in the following~$\tup{\uniVarA}$,~$\tup{\uniVarB}$, and~$\tup{\uniVarC}$~for a tuple of both element and numeric variables
and~$\tup{\vertA}$ and~$\tup{\vertB}$ for a tuple of both vertices and numbers.

Let $\sigA$ and $\sigB=\set{\rel_1,\dots, \rel_\ell}$ be relational signatures.
A \defining{$d$-dimensional $\IFPC[\sigA,\sigB]$-interpretation}  $\interpret(\tup{\uniVarC})$
with parameters $\tup{\uniVarC}$
is a tuple
\[\interpret(\tup{\uniVarC}) = \left(\formA_{\text{dom}}(\tup{\uniVarC}\tup{\uniVarA}),
 \formA_{\cong}(\tup{\uniVarC}\tup{\uniVarA}\tup{\uniVarB}),
 \formA_{\rel_1}(\tup{\uniVarC}\tup{\uniVarA}_1\cdots \tup{\uniVarA}_{\arity{\rel_1}}), \dots,
 \formA_{\rel_\ell}(\tup{\uniVarC}\tup{\uniVarA}_1\cdots \tup{\uniVarA}_{\arity{\rel_\ell}}), \tup{\termA}\right)
\]
of $\IFPC{}[\sig]$-formulas
and a $j$-tuple $\tup{\termA}$ of closed numeric $\IFPC{}[\sig]$-terms,
where $j$ is the number of numeric variables in $\tup{\uniVarC}$.
The tuples of variables $\tup{\uniVarA}$, $\tup{\uniVarB}$, and all the $\tup{\uniVarA}_i$ are of length $d$ and agree on whether the $k$-th variable is an element or numeric variable.
Let $\StructA$ be a $\sig$-structure
and $\tup{\vertA} \in (\StructVA \cup \nat)^{|\tup{\uniVarC}|}$ match the types of the parameter variables~$\tup{\uniVarC}$ (element or numeric).
We now define $\interpret(\StructA, \tup{\vertA})$.
Assume that up to reordering the first $j$ variables in $\tup{\uniVarC}$ are numeric variables
and set $D := \set{0,..., \termA_1^\Struct}\times\dots\times\set{0,..., \termA_j^\Struct} \times \StructVA^{d-j}$.
We define a $\sigB$-structure $\StructB = (\StructVB, \rel_1^\StructB, \dots, \rel_\ell^\StructB)$ via
\begin{align*}
	\StructVB&:= \setcond*{\tup{\vertB}\in D}{\tup{\vertA}\tup{\vertB} \in \formA_{\text{dom}}^\Struct}, \\
	\rel_i^\StructB &:= \setcond*{(\tup{\vertB}_1,\dots,\tup{\vertB}_{\arity{\rel_i}})\in \StructVB^{\arity{\rel_i}}}{\tup{\vertA}\tup{\vertB}_1\dots\tup{\vertB}_{\arity{\rel_i}}\in \formA_{\rel_i}^\Struct} &\text{for all } i \in [\ell].
	\intertext{Finally, using the relation $E$,
		we define the image of the interpretation:}
	E &:= \setcond*{(\tup{\vertB}_1,\tup{\vertB}_2) \in \StructVB^2}{\tup{\vertA}\tup{\vertB}_1\tup{\vertB}_2 \in \formA_{\cong}^\Struct},\\
	\interpret(\StructA,\tup{\vertA}) &:= \begin{cases}
		\StructB / E & \text{if } E \text{ is a congruence relation on } \StructB,\\
		\text{undefined}& \text{otherwise.}
	\end{cases}
\end{align*}
Here $\StructB/E$ is the quotient of $\StructB$ by $E$,
its vertices are the equivalence classes of $E$.
Intuitively, $\StructB/E$ is obtained from $\StructB$ by contracting every equivalence class of $E$ to a single vertex.
The congruence condition ensures that the $\sigB$-relations are well-defined on the quotient, that is, independent of the choice of representatives of the classes.
An interpretation is called \defining{equivalence-free}
if $\formA_{\cong}(\tup{\uniVarC}\tup{\uniVarA}\tup{\uniVarB})$
is the formula $\tup{\uniVarA} = \tup{\uniVarB}$.

When we consider a logic $L$ that is an extension of \IFPC{},
then the notion of an $L[\sigA,\sigB]$-interpretation
is defined exactly in the same way
by replacing \IFPC{}-formulas or terms with $L$-formulas or terms, respectively.
For logics, which do not posses numeric variables like \FO{} or \IFP{},
the notion of an interpretation is similar and just omits the numeric part,
i.e., there is no numeric term $\termA$ bounding the range of numeric variables.

A property $P$ of $\sig$-structures is \defining{$L$-reducible} to 
a property $Q$ of $\sigB$-structures
if there is an $L[\sig, \sigB]$-interpretation $\interpret$
such that, for every $\sig$-structure
$\Struct$, it holds that $\Struct \in P$ if and only if $\interpret(\Struct) \in Q$.
A logic $L'$ is \defining{closed under $L$-interpretations} (or $L$-reductions)
if for every property~$P$ that is $L$-reducible to an $L'$-definable property~$Q$, the property~$P$ itself is $L'$-definable (cf.~\cite{ebbinghaus1985,Otto1997}).
We say that $L'$ is closed under interpretations
if $L'$ is closed under $L'$-interpretations.

\section{Witnessed Symmetric Choice}
\label{sec:witnessed-symmetric-choice}

We extend \IFPC{} with an inflationary fixed-point operator with witnessed symmetric choice.
Let~$\sig$ be a relational signature
and~$R$,~$R^*$, and~$S$ be new relation symbols not contained in~$\sig$
of arities $\arity{R}=\arity{R^*}$ and $\arity{S}$.
The letter $p$ is used for an element parameter in this section.
We define the WSC-fixed-point operator with parameters $\tup{p}\tup{\numVarA}$.
If
\begin{itemize}
	\item 
	$\formStep(\tup{p}\tup{x}\tup{\numVarA})$
	is an \ifpcwscsig{\sig, R, S}-formula such that $|\tup{x}| = \arity{R}$, 
	\item
	$\formChoice(\tup{p}\tup{y}\tup{\numVarA})$ is an \ifpcwscsig{\sig, R}-formula
	such that $|\tup{y}| = \arity{S}$,
	\item
	$\formWit(\tup{p}\tup{y}\tup{y}'z_1z_2\tup{\numVarA})$ is an \ifpcwscsig{\sig  , R, R^*}-formula where $|\tup{y}|=|\tup{y}'|=\arity{S}$, and
	\item
	$\formOut(\tup{p}\tup{\numVarA})$ is an \ifpcwscsig{\sig, R^*}-formula,
\end{itemize}
then
\[\formA(\tup{p}\tup{\numVarA}) = \ifpwscbig{R,\tup{x}}{R^*}{S,\tup{y}, \tup{y}'}{z_1z_2}{\formStep(\tup{p}\tup{x}\tup{\numVarA})}{\formChoice(\tup{p}\tup{y}\tup{\numVarA})}{\formWit(\tup{p}\tup{y}\tup{y}'z_1z_2\tup{\numVarA})}{\formOut(\tup{p}\tup{\numVarA})}\]
is an \ifpcwscsig{\sig}-formula.
The formulas~$\formStep$,~$\formChoice$,~$\formWit$, and~$\formOut$ are
called \defining{step formula}, \defining{choice formula}, \defining{witnessing formula}, and \defining{output formula}, respectively.
The free variables of~$\formA$ are 
the ones of~$\formStep$ apart from~$\tup{x}$,
the ones of~$\formChoice$ apart from~$\tup{y}$,
the ones of~$\formWit$ apart from~$\tup{y}$,~$\tup{y}'$,~$z_1$ and~$z_2$, and
the ones of~$\formOut$.
That is,~the variables~$\tup{x}$ are bound in~$\formStep$,~the variables~$\tup{y}$ are bound in~$\formChoice$ and~$\formWit$,
and~the variables $\tup{y}'$,~$z_1$, and~$z_2$ are bound in~$\formWit$.
Note that only element variables are used for defining the fixed-point
in the WSC-fixed-point operator.
This suffices for our purpose in this article and 
increases readability.
We expect that our arguments also work with numeric variables in the fixed-point.
For the sake of readability, we will now omit the free numeric variables~$\tup{\numVarA}$
when defining the semantics of the WSC-fixed-point operator.
Fixing numeric parameters does not change orbits or automorphisms.

\paragraph{Evaluation with Choices }
Intuitively, the WSC-fixed-point operator~$\formA$ is evaluated as follows:
Let $\StructA$ be a $\sig$\nobreakdash-structure
and $\tup{\vertA} \in \StructVA^{|\tup{p}|}$ be a tuple of parameters.
We define a sequence of relations called \defining{stages} $\emptyset =: R_1^\Struct \subset \dots \subset  R_\ell^\Struct = R_{\ell+1}^\Struct =: (R^*)^\Struct$.
Given the relation $R_i^\Struct$,
the choice formula defines the choice-set $T_{i+1}^\Struct$
of the tuples $\tup{\vertB}$ satisfying $\formChoice$, i.e.,
\[T_{i+1}^\Struct := \setcond*{\tup{\vertB}}{\tup{\vertA}\tup{\vertB} \in \formChoice^{\Struct, \smash{R_i^\StructA}}}.\]
We pick an arbitrary tuple $\tup{\vertB} \in T_{i+1}^\Struct$
and set $S_{i+1}^\Struct := \set{\tup{\vertB}}$
or $S_{i+1}^\Struct := \emptyset$ if no such $\tup{\vertB}$ exists.
The step formula
is used on the structure $(\StructA, R_i^\Struct, S_{i+1}^\Struct)$ to define the next stage in the fixed-point iteration:
\[R_{i+1}^\Struct :=  R_{i}^\Struct \cup \setcond[\Big]{\tup{\vertC}}{\tup{\vertA}\tup{\vertC} \in \formStep^{(\StructA, R_i^\Struct, S_{i+1}^\Struct)}}.\]
We proceed in that way until a fixed-point $(R^*)^\Struct$ is reached.
Because we define an inflationary fixed-point, it is guaranteed to exist.
This fixed-point is in general ill-defined
because it depends on the choices made.

We ensure that $\formA$ is still isomorphism-invariant as follows:
First, we only allow choices from orbits, which the witnessing formula has to certify.
A set $N \subseteq \autGroup{(\StructA, \tup{\vertA})}$ \defining{witnesses a relation $R \subseteq \StructVA^k$ as $(\StructA, \tup{\vertA})$-orbit},
if for every $\tup{\vertB}, \tup{\vertB}' \in R$,
there is a $\autoA \in N$
satisfying $\tup{\vertB} = \autoA(\tup{\vertB}')$.
Because we need the notion of witnessing orbits only for isomorphism-invariant sets, we do not need to check whether $R$ is a proper subset of an orbit.
We require that $\formWit$ defines a set of automorphisms.
Intuitively, for $\tup{\vertB}, \tup{\vertB}' \in T_{i+1}^\Struct$, a map $\autoA_{\tup{\vertB}, \tup{\vertB}'}$ is defined via
\[ \vertC \mapsto \vertC' \text{ whenever } \tup{\vertA}\tup{\vertB}\tup{\vertB}'\vertC\vertC' \in \formWit^{(\Struct, R_i^\Struct, (R^*)^\Struct)}.\]
The set of all these maps for all $\tup{\vertB}, \tup{\vertB}' \in T_{i+1}^\Struct$
has to witness $T_{i+1}^\Struct$ as $(\Struct, \tup{\vertA}, R_1^\Struct, \dots, R_i^\Struct)$-orbit.
Note here, that the witnessing formula always has access to the fixed-point.
Actually, we do not require that $\autoA_{\tup{\vertB}, \tup{\vertB}'}$ maps $\tup{\vertB}$ to $\tup{\vertB}'$ but only the set of all $\autoA_{\tup{\vertB}, \tup{\vertB}'}$ has to witness the orbit.
If some choice-set is not witnessed, $\formA$ is not satisfied by $\tup{\vertA}$.
Otherwise, the output formula is evaluated on the defined fixed-point: \[\formA^\Struct :=\formOut^{(\Struct, (R^*)^\Struct)}.\]
Because all choices are witnessed, all possible fixed-points (for different choices)
are related by an automorphism of $(\Struct, \tup{\vertA})$
and thus either all of them or none of them satisfy the output formula.

\paragraph{An Example}
We give an illustrating example (by adapting an example for \CPTWSC{} in~\cite{LichterSchweitzer24}).
We show that the class of threshold graphs (i.e, graphs that can be reduced to the empty graph by iteratively deleting universal or isolated vertices)
is \IFPCWSC{}-definable.
Threshold graphs are actually \IFP{}-definable, but our formula below illustrates the WSC-fixed-point operator.
The set of all isolated or universal vertices of a graph $G$ forms a $1$-orbit
(note that there cannot be an isolated and a universal vertex at the same time).
We choose one vertex of this orbit, collect the chosen vertex in a unary relation~$R$, and repeat pretending that the vertices in $R$ are deleted.
If all vertices are contained in the obtained fixed-point~$R^*$,
then the graph $G$ is a threshold graph.
The choice formula~$\formChoice$ defines the set of all isolated or universal vertices in~$G - R$.
The step formula~$\formStep$ adds the chosen vertex, which is the only vertex in the relation~$S$, to~$R$.
The output formula $\formOut$ checks whether~$R^*$ contains all vertices
and so defines whether it was possible to delete all vertices:
\begin{align*}
	\formChoice(y) &:= \neg R(y) \land \Big(\big(\forall z.\qspace\neg R(z) \Rightarrow E(y,z)\big) \lor \big(\forall z.\qspace\neg R(z) \Rightarrow \neg  E(y,z)\big)\Big),\\
	\formStep(x) &:= R(x) \lor S(x) , \\
	\formOut &:= \forall x.\qspace R^*(x).
\end{align*}
Witnessing orbits is easy: To show that two isolated (or universal, respectively) vertices $y$ and $y'$ are related by an automorphism, it suffices to define their transposition
as follows:
\begin{align*}
	\formWit(y,y', z_1,z_2) &:= 
		(z_1 = y \land z_2 = y') \lor  
		(z_2 = y \land z_1 = y') \lor 
		(y \neq z_1 = z_2 \neq y').
\end{align*}
To the end, the formula
$\ifpwsc{R,x}{R^*}{S,y, y'}{z_1z_2}{\formStep}{\formChoice}{\formWit}{\formOut}$
defines the class of threshold graphs.

\paragraph{Formal Semantics}
To define the semantics of the WSC-fixed-point operator formally,
we use the $\wscStarSym$-operator defined in~\cite{LichterSchweitzer24}.
The $\wscStarSym$-operator captures the idea of fixed-point iterations with choices from orbits for arbitrary isomorphism-invariant functions.
Let $\StructA$ be a $\sig$-structure and $\tup{\vertA} \in \StructVA^k$.
We denote by $\HF{\StructVA}$ the set of all hereditary finite sets over $\StructVA$.
The $\wscStarSym$\nobreakdash-operator defines for isomorphism-invariant functions
${\stepF^{\Struct, \tup{\vertA}},\autF^{\StructA, \tup{\vertA}} \colon \HF{\StructVA} \times \HF{\StructVA} \to \HF{\StructVA}}$ 
and $\choiceF^{\StructA, \tup{\vertA}} \colon \HF{\StructVA} \to \HF{\StructVA}$,
the set
\[W=\wsc{\stepF^{\Struct, \tup{\vertA}}}{\autF^{\StructA, \tup{\vertA}}}{\choiceF^{\StructA, \tup{\vertA}}}\]
of all $\HF{\StructVA}$-sets obtained in the following way.
Starting with $b_0 := \emptyset$,
define a sequence of sets as follows:
Given $b_i$,
define the choice-set $c_i := \choiceF^{\StructA, \tup{\vertA}}(b_i)$,
pick an arbitrary $d_i \in c_i$ (or $d=\emptyset$ if $c_i = \emptyset$),
and set 
$b_{i+1} := \stepF^{\StructA, \tup{\vertA}}(b_i, d_i)$.
Let $b^* := b_\ell$ for the smallest
$\ell$ satisfying $b_\ell = b_{\ell+1}$ (if it exists, which in our case of inflationary fixed-points is always the case).
Then we include $b^*$ in $W$ if, for every $i \in [\ell]$,
the set $\autF^{\StructA, \tup{\vertA}}(b_i, b^*)$
is a set of automorphisms
witnessing $c_i$ as $(\StructA, \tup{\vertA}, b_1, \dots, b_i)$-orbit.
Of course $b^*$ is not unique and depends on the choices of the $d_i$
and thus $W$ is not necessarily  a singleton set.
It is proven in~\cite{LichterSchweitzer24}
that $\autF^{\StructA, \tup{\vertA}}$
witnesses all choice-sets~$c_i$ 
for either every possible~$b^*$ obtained in the former way
or for none of them.
In particular,~$W$ is an $(\StructA, \tup{\vertA})$-orbit.
Formally, the set $W$ is defined using trees capturing all possible choices.
For more details, we refer to~\cite{LichterSchweitzer24}.
We now define the semantics of the WSC+fixed-point operator
\[\formA(\tup{p}) = \ifpwscbig{R,\tup{x}}{R^*}{S,\tup{y}, \tup{y}'}{z_1z_2}{\formStep(\tup{p}\tup{x})}{\formChoice(\tup{p}\tup{y})}{\formWit(\tup{p}\tup{y}\tup{y}'z_1z_2)}{\formOut(\tup{p})}\]
using tuples (implicitly encoded as $\HF{\StructV}$-sets).
Let $\tup{\vertA} \in \StructVA^{|\tup{p}|}$.
We set
\begin{align*}
	\stepF^{\Struct, \tup{\vertA}}(R^\Struct,S^\Struct) &:= 
	R^\Struct\cup \setcond*{\tup{\vertC}}{
		\tup{\vertA}\tup{\vertC} \in \formStep^{(\StructA, R^\Struct, S^\Struct)}},
\\
	\choiceF^{\Struct, \tup{\vertA}}(R^\Struct) &:= 
	\setcond*{\tup{\vertB}}{\tup{\vertA}\tup{\vertB} \in \formChoice^{(\StructA,R^\Struct)}}, \text{ and}\\
	\autF^{\Struct, \tup{\vertA}}(R^\Struct,(R^*)^\Struct) &:=
	 \setcond*{ \autoA_{\tup{\vertB}\tup{\vertB}'}}{\tup{\vertB},\tup{\vertB}' \in \choiceF^{\Struct, \tup{\vertA}}(R^\Struct)} \text{ where}   \\
	 	&\hiddenEqq \quad\autoA_{\tup{\vertB}\tup{\vertB}'} = \setcond*{(\vertC,\vertC') \in \StructVA^2}{\tup{\vertA}\tup{\vertB}\tup{\vertB}'\vertC\vertC' \in \formWit^{(\StructA, R^\Struct, (R^*)^\Struct)}} .
\end{align*}
The function $\stepF^{\Struct, \tup{\vertA}}(R^\Struct,S^\Struct)$
evaluates the step formula $\formStep$ and adds its output to~$R^\Struct$,
which eventually defines the inflationary fixed-point.
The function $\choiceF^{\Struct, \tup{\vertA}}(R^\Struct)$
defines the choice-set by evaluating the choice formula $\formChoice$.
Finally,
the function $\autF^{\Struct, \tup{\vertA}}(R^\Struct,(R^*)^\Struct)$ defines a set of automorphisms
by evaluating the witnessing formula $\formWit$ for all tuples
in the current relation.
Note that the set can formally contain arbitrary binary relations,
but if they are not automorphisms, the choice will not be witnessed.
Now set $W_\formA^{(\StructA,\tup{\vertA})} := \wsc{\stepF^{\Struct, \tup{\vertA}}}{\choiceF^{\Struct, \tup{\vertA}}}{\autF^{\Struct, \tup{\vertA}}}$.

We define the semantics of the WSC-fixed-point operator~$\formA$ as follows.
For a $\sig$-structure~$\StructA$ we define
\begin{equation}
	\label{eqn:wsc-semantics}
	\formA^\Struct := \setcond*{\tup{\vertA}}{\tup{\vertA} \in \formOut^{(\Struct, (R^*)^\Struct)} \text{ for some } (R^*)^\Struct \in W_\formA^{(\StructA,\tup{\vertA})}}. \tag{WSC}
\end{equation}
Note that if $W_\formA^{(\StructA,\tup{\vertA})} =\emptyset$, that is, not all choices could be witnessed,
then we have $\tup{\vertA} \notin \formA^\Struct$.
Also note that because $W_\formA^{(\StructA,\tup{\vertA})}$ is an $(\StructA, \tup{\vertA})$-orbit,
$\tup{\vertA} \in \formOut^{(\StructA, (R^*)^\Struct)}$
holds for either every $(R^*)^\Struct \in W_\formA^{(\StructA,\tup{\vertA})}$ or for no $(R^*)^\Struct \in W_\formA^{(\StructA,\tup{\vertA})}$.
Finally, we conclude that \IFPCWSC{} is isomorphism-invariant:

\begin{lem}
	For every \ifpcwscsig{\sig}-formula $\formA$
	and every $\sig$-structure $\Struct$,
	the set $\formA^\Struct$ is
	a union of $\StructA$-orbits.
\end{lem}
\begin{proof}
	The proof is straight-forward by induction on the formula
	using that $W_\formA^{(\StructA,\tup{\vertA})}$ is an $(\StructA,\tup{\vertA})$-orbit.
\end{proof}

\newcommand{\IFPCWSCsub}{\ensuremath{%
			\mathrm{IFPC{+}WSC^{\mathrm{sub}\upharpoonright}}%
}}
\newcommand{\IFPCWSCglob}{\ensuremath{%
		\mathrm{IFPC{+}WSC^{\mathrm{glob}\upharpoonright}}%
}}
\newcommand{\IFPCWSCIglob}{\ensuremath{%
		\mathrm{IFPC{+}WSC^{\mathrm{glob}\upharpoonright}{+}I}%
}}
\newcommand{\ifpcwscsubsig}[1]{\ensuremath{\IFPCWSCsub{}[#1]}}
\paragraph{The Reduct Property}
According to Ebbinghaus~\cite{ebbinghaus1985}, a reasonable logic
has the property that whether a structure $\StructA$ satisfies a formula $\formA$
only depends on the relation symbols used in $\formA$.
This property is called the \emph{reduct property}:
formally, if~$L$ is a logic,~$\formA$ is an $L[\sigA]$-sentence, and $\StructA$ a $\sigB$-structure with $\sigA \subseteq \sigB$,
then $\StructA$ satisfies $\formA$ if and only if the reduct $\reduct{\StructA}{\sigA}$ satisfies $\formA$.

This is not the case for the WSC-fixed-point operator: the orbits of $\StructA$ also depend on the relation symbols in $\sigB\setminus \sigA$,
and thus the evaluation of WSC-fixed-point operators depends on all relations.
Hence, $\IFPCWSC$ does not have the reduct property.
This is also the case for the logic introduced by Gire and Hoang~\cite{GireHoang98}.
We now discuss two possibilities to solve this issue.
For this, we will define two variants of \IFPCWSC{}:
their formulas
are syntactically the ones of \IFPCWSC{}, but with two different semantics:
subformula and global reduct semantics.
The semantics as defined in the previous paragraph is referred as \emph{non-reduct semantics}%
\footnote{The conference version of this article~\cite{Lichter2023b} only considered one of the two variants, \IFPCWSC{} with subformula reduct semantics.
However, for this logic some of the statements were false, essentially caused by Lemma~\ref{lem:ipfcwsc-subformula-reduct-closed-under-one-dim-interpretations}.
This article corrects this mistake by a more detailed analysis of Ebbinghaus' reduct property for witnessed symmetric choice, but also has to weaken some claims from subformula to global reduct semantics.
In particular, one question of Dawar and Richerby~\cite{DawarRicherby03} posed for subformula reduct semantics remains open and is only answered for global reduct semantics (see the end of Section~\ref{sec:separating-wsc-wsci}).
The author apologizes for this mistake.
}.

First, we introduce \IFPCWSC{} with \emph{subformula reduct semantics}.
Dawar and Richerby solved the problem by evaluating a (W)SC-fixed-point operator~$\formA$ not on the input structure~$\StructA$
but on the reduct~$\reduct{\StructA}{\sigOf{\formA}}$, where $\sigOf{\formA}$ is the set of free relation symbols in $\formA$ (the ones not bound by fixed-point or WSC-fixed-point operators and interpreted by~$\StructA$).
On a formal level, we would
have to replace
\begin{align*}
	\formOut^{(\Struct, (R^*)^\Struct)} \text{ and } 
	W_\formA^{(\StructA, \tup{\vertA})}  \quad \text { by } \quad
	\formOut^{(\reduct{\StructA}{\sigOf{\formA}}, (R^*)^\Struct)} \text{ and }
	W_\formA^{(\reduct{\StructA}{\sigOf{\formA}}, \tup{\vertA})},
\end{align*}
respectively, in Equation~(\ref{eqn:wsc-semantics})
to define the subformula reduct semantics.
Then a WSC-fixed-point operator $\formA$ is evaluated on the $\sigOf{\formA}$-reduct.
We refer to \IFPCWSC{} with subformula reduct semantics by \IFPCWSCsub{}.
With the subformula reduct semantics, the evaluation of the WSC-fixed-point operator does not depend on unused relation symbols anymore
and hence this logic has the reduct property.
In particular, if different WSC-fixed-point operators in a formula use different relation symbols, they are evaluated on different reducts.

The second approach is \IFPCWSC{} with \emph{global reduct semantics}.
The idea is a straight-forward way to ensure that a logic has the reduct property: to evaluate a formula $\formA$ on a structure $\StructA$, the formula $\formA$ is immediately evaluated on the reduct $\reduct{\StructA}{\sigOf{\formA}}$.
In particular, the transition to the reduct is only done once and then never again for subformulas.
This approach could be applied to any logic to ensure that it has the reduct property.
Formally, a formula $\formA$, under global reduct semantics, evaluates on $\StructA$ to $\formA^{\reduct{\StructA}{\sigOf{\formA}}}$ (where $\formA^{\reduct{\StructA}{\sigOf{\formA}}}$ refers to evaluation under non-reduct semantics).
We write \IFPCWSCglob{} for \IFPCWSC{} with global reduct semantics.
This logic evaluates all WSC-fixed point operators in a formula  on the same structure.

In what follows, we will explicitly state when formulas are evaluated under the global or subformula reduct semantics.
The notation $\formA^\StructA$ will always refer to the non-reduct semantics.
Comparing the approaches, the subformula reduct semantics
is more natural because evaluation is defined recursively
and no initial ``exception'' to consider the reduct is needed
as it is for global reduct semantics.
This implies that subformulas under the subformula reduct semantics have the reduct property, too.
The same is not true for global reduct semantics (here only the ``global'' formula has the reduct property).
The subformula reduct semantics is potentially more expressive:
\begin{lem}
	$\IFPCWSCglob{} \leq \IFPCWSCsub{}$.
\end{lem}
\begin{proof}
	Let~$\formA$ be an $\IFPCWSCglob{}$-formula.
	Note that~$\formA$ is also an $\IFPCWSCsub{}$-formula,
	but potentially has different semantics under the subformula reduct semantics than under the global one.
	Because \IFPCWSCglob{} and \IFPCWSCsub{}
	have the reduct property,
	it suffices to find an  \IFPCWSCsub{}-formula $\formA'$ such that $\sigOf{\formA'} = \sigOf{\formA}$ and $\formA'$, under the subformula reduct semantics, is equivalent to $\formA$ with the global reduct semantics for all $\sigOf{\formA}$-structures.
	Note that for $\sigOf{\formA}$-structures in the global reduct semantics, we can omit the initial reduct step when evaluating $\formA$
	and consider the non-reduct semantics.

	We will prove the following:
	For every signature $\sigA$
	and every \IFPCWSC{}-formula (or term) $\formA$, there is an \IFPCWSCsub{}-formula (or term, respectively) $\formA^\sigA$ such that
	\begin{enumerate}
		\item $\sigOf{\formA^\sigA} = \sigOf{\formA} \cup \sigA$
		and
		\item $\formA$ with non-reduct semantics and $\formA^\sigA$ with subformula reduct semantics are equivalent for all $\sigOf{\formA^\sigA}$-structures.
	\end{enumerate}
	For $\sigA = \sigOf{\formA}$, this claim proves the lemma.
	
	So let $\sigA$ be some signature.
	For simplicity, we assume that relation symbols are always renamed such that $\sigA$-symbols are never bound in~$\formA$.
	The base case of the translation is the case that $\formA$ is an $\IFPC{}$-formula or term.
	Let $\formB$ be an arbitrary \IFPC{}-tautology such that $\sigOf{\formB} = \sigA$.
	If $\formA$ is a formula, then $\formA^\sigA := \formA \land \formB$ satisfies the claim
	because for \IFPC{}-formulas the two semantics do not differ.
	If $\formA = \termA$ is a term, then $\termA^\sigA = \termA + \# x.\qspace x \neq x \land \formB$ satisfies the claim.
	
	Next, consider the case that a formula or term $\formA$ is obtained by \IFPC{}-formula-formation rules from subformulas $\formB_1,\dots, \formB_\ell$ and subterms $\termB_1, \dots, \termB_k$ (which are not all \IFPC{}-formulas or terms).
	Let $\sigB := \sigOf{\formA} \cup \sigA$.
	By the inductive hypothesis there are $\formB_1^\sigB, \dots, \formB_\ell^\sigB$ and $\termB_1^\sigB,\dots,\termB_k^\sigB$.
	From them, construct $\formA^\sig$ in the same way as $\formA$ is constructed from $\formB_1,\dots, \formB_\ell$ and $\termB_1, \dots, \termB_k$.
	For example, $\#\uniVarVA\numVarVA \leq (\termB_1,\dots,\termB_\ell).\qspace\formA$
	is turned into $ \#\uniVarVA\numVarVA \leq (\termB_1^\sigB,\dots, \termB_\ell^\sigB).\qspace\formA^\sigB$.
	By the inductive hypothesis, $\formB_i^\sigB$ with subformula reduct semantics
	is equivalent to $\formB_i$ with non-reduct semantics for all $\sigOf{\formB_i} \cup \sigB=\sigOf{\formB_i^\sigB} $-structures for all $i \in [\ell]$ (and similar for the terms).
	If $\formA$ is not a (non-WSC) fixed-point operator, no relation symbol is bound in $\formA$ and thus $\sigOf{\formB_i} \cup \sigB = \sigOf{\formA} \cup \sigA = \sigB$.
	Thus, $\formA^\tau$ is equivalent to $\formA$ on all $\sigB$-structures.
	Otherwise, $\formA = [\ifp \rel \tup{x}.\qspace \formB_1](\tup{y})$ is a fixed-point operator binding the relation symbol $\rel$.
	If $\rel \notin \sigOf{\formB}$, the claim follows by the inductive hypothesis.
	Otherwise, by the inductive hypothesis,
	$\formB_1^\sigB$ with subformula reduct semantics
	is equivalent to $\formB_1$ with non-reduct semantic
	for all $ \sigB \cup \set{\rel}=\sigOf{\formB_1^\sigB}$-structures.
	Hence, for all $\sigB$-structures, initially extended by $R$ to be empty,
	they will define the same inflationary fixed-point.
	Because $\sigOf{\formA^\sigB} = \sigOf{\formB_1^\sigB} \setminus \set{\rel} = \sigB = \sigOf{\formA} \cup \sigA$,
	it follows that $\formA^\sigA$  with subformula reduct semantics is equivalent to $\formA$ with non-reduct semantics for all $\sigB$-structures.
	
	Lastly, consider WSC-fixed-point operators. Let
	\[\formA(\tup{p}) = \ifpwscbig{R,\tup{x}}{R^*}{S,\tup{y}, \tup{y}'}{z_1z_2}{\formStep(\tup{p}\tup{x})}{\formChoice(\tup{p}\tup{y})}{\formWit(\tup{p}\tup{y}\tup{y}'z_1z_2)}{\formOut(\tup{p})}\]
	be an WSC-fixed-point operator
	and let $\sigB := \sigOf{\formA} \cup \sigA$.
	By the inductive hypothesis, there are formulas $\formStep^\sigB$, $\formChoice^\sigB$, $\formWit^\sigB$, and $\formOut^\sigB$ that under the subformula reduct semantics are equivalent to $\formStep$, $\formChoice$, $\formWit$, and $\formOut$, respectively, under the non-reduct reduct semantics for all $\sigB$-structures.
	Define 
	\[\formA^\sigA(\tup{p}) := \ifpwscbig{R,\tup{x}}{R^*}{S,\tup{y}, \tup{y}'}{z_1z_2}{\formStep^\sigB(\tup{p}\tup{x})}{\formChoice^\sigB(\tup{p}\tup{y})}{\formWit^\sigB(\tup{p}\tup{y}\tup{y}'z_1z_2)}{\formOut^\sigB(\tup{p})}.\]
	We show that $\formA^\sigA$ satisfies the claim.
	First note that $\sigOf{\formA^\sigA} = \sigOf{\formA} \cup \sigA =\sigB $
	because $\sigOf{\formA_{\text{x}}^\sigB} \setminus \set{R,R^*,S} = (\sigOf{\formA_{\text{x}}} \setminus \set{R,R^*,S} ) \cup \sigB = \sigOf{\formA} \cup \sigA$
	for every $\text{x} \in \set{\text{step},\text{choice},\text{wit},\text{out}}$.
	Now let~$\StructA$ be a $\sigB$-structure.
	The evaluation of the WSC-fixed-point operator under the subformula reduct semantics does not consider a proper reduct of~$\StructA$ because $\sigOf{\formA^\sigA} = \sigB$.
	By the inductive hypothesis for the $\formA_{\text{x}}^\sigB$
	and using a similar reasoning as for the plain fixed-point,
	the formula $\formA^\sigA$ evaluates under the subformula reduct semantics as~$\formA$ under the non-reduct one on~$\StructA$.
\end{proof}
Whether global reduct semantics is as expressive as the subformula one is not as easy to see and will be answered in this article.
We also compare the two semantics with respect to closure under interpretations.

\begin{lem}
	\label{lem:ipfcwsc-subformula-reduct-closed-under-one-dim-interpretations}
	The logic \IFPCWSCsub{} is closed under one-dimensional and equivalence-free   \IFPCWSCsub{}-interpretations.
\end{lem}
\begin{proof}
	Let $\sigA$ and $\sigB = \set{\rel_1,\dots, \rel_\ell}$ be signatures, let \[\interpret(\tup{\uniVarC}) = \left(\formA_{\text{dom}}(\tup{\uniVarC}\uniVarA),
	\formA_{\cong}(\tup{\uniVarC}\uniVarA\uniVarB),
	\formA_{\rel_1}(\tup{\uniVarC}\uniVarA_1\cdots \uniVarA_{\arity{\rel_1}}), \dots,
	\formA_{\rel_\ell}(\tup{\uniVarC}\uniVarA_1\cdots \uniVarA_{\arity{\rel_\ell}})\right)
	\]
	be a one-dimensional and equivalence-free  $\ifpcwscsubsig{\sigA,\sigB}$-interpretation, and let~$\formB$ be an $\ifpcwscsubsig{\sigB}$-sentence.
	Define, for each $i \in [\ell]$, the formula $\formA'_{\rel_i} := \formA_{\rel_i} \land \bigwedge_{j \in [\arity{R_i}]} \formA_{\text{dom}}(\tup{z}x_j)$
	to be $\formA_{\rel_i}$  that is only satisfied by vertices satisfying $\formA_{\text{dom}}$.
	Let $\formB_{\text{dom}}$ be obtained form $\formB$ by
	restricting all quantifiers, counting operators, (WSC)-fixed-points, etc.~(wherever first-order variables are bound) to vertices in~$R_{\text{dom}}$, for example, turn $\exists x.\qspace \Pi$ to $\exists x \in R_{\text{dom}}.\qspace \Pi$.
	
	We introduce, for \IFPCWSCsub{}-formulas $\Pi$ and $\Pi'$, the following shorthand notation:
	\[[\ifpwscSym R\tup{x}. \qspace\Pi](\Pi') := 
	\ifpwscbig{R,\tup{x}}{Y}{Z,(),()}{y_1y_2}{\Pi}{\bot}{\bot}{\Pi'},\]
	where $Y,Z$ are arbitrary fresh nullary relation symbols and $\bot$ is a contradictory formula with $\sigOf{\bot}=\emptyset$.
	The formula $(\ifpwscSym R\tup{x}. \Pi)[\Pi']$ on a structure $\StructA$
	defines the inflationary fixed-point defined by $\Pi$ and then evaluates $\Pi'$ using it (formulas are evaluated on $\reduct{(\StructA, R^\StructA)}{(\sigOf{\Pi'}\cup \sigOf{\Pi})}$).
	Because the choice and witnessing formula are contradictory, no choices are made or needed to be witnessed.
	
	Now consider the formula
	\begin{align*}
		\formB' &:= \left[\ifpwscSym \rel_{\text{dom}} \uniVarA.\qspace \formA_{\text{dom}}\right]\Big(\\
			&\qquad\qquad \left[\ifpwscSym \rel_1 \uniVarA_1\cdots \uniVarA_{\arity{\rel_1}}.\qspace\formA'_{\rel_1} \right] \left( \dots \left[\ifpwscSym \rel_\ell \uniVarA_1\cdots \uniVarA_{\arity{\rel_\ell}}.\qspace\formA'_{\rel_\ell} \right](\formB_{\text{dom}}) \dots\right)\Big).
	\end{align*}
	In the $\sigA$-formula $\formB'$, the fixed-points are used to define the relation symbols $\rel_{\text{dom}}, \rel_1,\dots,\rel_\ell$ of the interpretation,
	on which $\formB_{\text{dom}}$ is evaluated.
	
	Let~$\StructA$ be a $\sigA$-structure.
	We claim that, with subformula reduct semantics,
	the formula~$\formB'$ evaluates on~$\StructA$ as~$\formB$ on $\StructB := \interpret(\StructA) = (B, \rel_1^{\StructB}, \dots, \rel_\ell^\StructB)$.
	The first fixed-point in~$\formB'$ interprets~$\rel_{\text{dom}}$ as $B$.
	The remaining fixed-points in~$\formB'$ interpret $\rel_i$ as $\rel_i^{\StructB}$:
	Because \IFPCWSCsub{} has the reduct property,
	it does not matter whether $\formA'_{\rel_i}$ is evaluated on $(\StructA, \rel_{\text{dom}}^\StructA)$ or on $(\StructA, \rel_{\text{dom}}^\StructA, \rel_1^\StructB, \dots, \rel_{i-1}^\StructB)$
	(or possibly some reduct of it if not all relations in $\sigA$ are used).
	Finally, $\formB_{\text{dom}}$ is evaluated on  $(\StructA, \rel_{\text{dom}}^\StructA, \rel_1^\StructB, \dots, \rel_\ell^\StructB)$.
	Because also subformulas in \IFPCWSCsub{} have the reduct property,
	$\formB_{\text{dom}}$ can equivalently be evaluated on the $(\sigOf{\formB} \cup  \set{\rel_{\text{dom}}})$-reduct.
	Let this reduct be~$\StructC$.
	In~$\StructC$, every $\rel_i$-relation satisfies $\rel_i^{\StructC} = \rel_i^\StructB \subseteq (\rel_{\text{dom}}^\StructA)^{\arity{\rel_i}}$ by the definition of~$\formA'_{\rel_i}$.
	This means that $\StructC$ is the disjoint union of $\StructB$ (after putting all vertices into $R_{\text{dom}}$)
	and some isolated vertices that are in no relation (and in particularly not in $\rel_{\text{dom}}$).
	This means that the orbits of~$\StructC$ are the orbits of~$\StructB$ plus one orbit for the isolated vertices.
	The formula $\formB_{\text{dom}}$ restricts binding first-order variables to vertices in $\rel_{\text{dom}}$.
	This in particularly happens in the step formula of every WSC-fixed-point operator that is a subformula of $\formB_{\text{dom}}$.
	This implies that the subformula reduct semantics never removes the $\rel_{\text{dom}}$ relation.
	Hence, the isolated vertices can also be removed since they are never assigned to some variable and never change the orbits during the evaluation.
	Thus,~$\formB'$ on~$\StructC$ evaluates as~$\formB$ on $\StructB$ (both with subformula reduct semantics)\footnote{We used WSC-fixed-points to define the relations~$\rel_{\text{dom}}$ and~$\rel_i$ because WSC-fixed-points directly allow to evaluate~$\formB_{\text{dom}}$ using these relations, which with plain (non-WSC) fixed-points requires more coding:
	The WSC-fixed-points can easily be replaced by a simultaneous plain fixed-point.
	However, encoding the simultaneous fixed-point back into non-simultaneous ones is possible, but requires more coding
	than the translation for simultaneous fixed-points in \IFPC{}
	because we have to ensure that subformulas can move to proper reducts.}.
\end{proof}
The same approach does not work for higher-dimensional interpretations
because the reducts of WSC-fixed-point operators only consider orbits on the original universe but not on the larger universe defined by a higher-dimensional interpretation.
For global reduct semantics, it is unclear whether it is closed under one-dimensional interpretations.
Thus, it allows for a more fine-grained analysis of the interplay of witnessed symmetric choice and interpretations.
We can understand the expressive power gained with witnessed-symmetric choice 
better and to which extend this power depends on subtleties like the question of global versus subformula reduct semantics.
To do so, we now build in logical interpretations into the logic.

\paragraph{An Operator for Logical Interpretations}
We extend the logic \IFPCWSC{} with another operator for interpretations.
First, every \IFPCWSC{}-formula is an \IFPCWSCI{}-formula.
Second, if $\interpret(\tup{p}\tup{\numVarA})$ is an \ifpcwscisig{\sigA,\sigB}-interpretation
with parameters $\tup{p}\tup{\numVarA}$
and~$\formA$ is an \ifpcwscisig{\sigB}-sentence,
then the \defining{interpretation operator}
\[\formB(\tup{p}\tup{\numVarA}) = \iop{\interpret(\tup{p}\tup{\numVarA})}{ \formA}\]
is an $\ifpcwscisig{\sigA}$-formula with free variables $\tup{p}\tup{\numVarA}$.
The (non-reduct) semantics is defined as follows:
\[\iop{\interpret(\tup{p}\tup{\numVarA})}{\formA}^\StructA := 
\setcond*{\tup{\vertA}\tup{n}\in \StructV^{|\tup{p}|}\times\nat^{|\tup{\numVarA}|}}{ \formA^{\interpret(\StructA, \tup{\vertA}\tup{n})} \neq \emptyset}.\]
Note that $\formA^{\interpret(\StructA, \tup{\vertA}\tup{n})} \neq \emptyset$
if and only if $\formA$ is satisfied.
The interpretation operator allows to evaluate a subformula in the image of an interpretation.
Thus, by definition, \IFPCWSCI{} is closed under interpretations.
For \IFPC{}, such an operator does not increase the expressive power because \IFPC{} is already closed under \IFPC{}-interpretations (see~\cite{Otto1997}).
For \IFPCWSC{}, this is not clear:
Because $\interpret(\StructA, \tup{\vertA}\tup{n})$ may have different automorphisms than~$\StructA$,
the formula~$\formA$ may exploit the WSC-fixed-point operator in a way
which is not possible in~$\StructA$.
The logic \IFPCWSCglob{}
is extended in the same way with the interpretation operator
yielding the logic \IFPCWSCIglob{}.
Its evaluation strategy is as before:
evaluate a formula~$\formA$ under the non-reduct semantics of \IFPCWSCI{} 
in the $\sigOf{\formA}$-reduct.
In particular, the global reduct semantics
will not consider a reduct of $\interpret(\StructA, \tup{\vertA}\tup{n})$
when evaluating the interpretation operator
but just does this once as the initial step.
The interpretation operator allows to express subformula and global reduct semantics:
\begin{lem}
	\label{lem:interpretation-operator-expresses-sub-and-glob}
$\IFPCWSCsub \leq \IFPCWSCIglob{} \leq  \IFPCWSCI$.
\end{lem}
\begin{proof}
	We show the first inclusion.
	For a signature $\sig$ and a tuple of variables $\tup{p}$, 
	let $\sig_{\tup{p}} := \sig \cup \set{P}$,
	where $P$ is a fresh $|\tup{p}|$-ary relation symbol.
	Let $\interpret^\sig_{\tup{p}}(\tup{p})$ be the one-dimensional and equivalence-free \FO{}-interpretation
	that defines the ``$\sig_{\tup{p}}$-reduct'' of all $\sigB$-structures with $\sig \subseteq \sigB$ as follows:
	The domain formula is tautological.
	All relations $\rel\in \sigA$ are preserved (defined by the formula $\formA_\rel(\tup{x}) := \rel(\tup{x})$).
	The relation~$P$ stores the assignment to~$\tup{p}$ (defined by the formula $\formA_{P_i}(\tup{x}) := \bigwedge_{i \in [|\tup{p}|]} x_i = p_i$).
	In particular, $\sigOf{\interpret^\sig_{\tup{p}}} = \sig$.
	
	Let~$\formB$ be an \IFPCWSCsub{}-formula.
	We construct
	an \IFPCWSCI{}-formula $\formB'$ equivalent to~$\formB$
	for all $\sigOf{\formB}$-structures.
	Consider a WSC-fixed-point operator 
	\begin{align*}
		\formA(\tup{p})&=\ifpwscbig{R,\tup{x}}{R^*}{S,\tup{y}, \tup{y}'}{z_1z_2}{\formStep(\tup{p}\tup{x})}{\formChoice(\tup{p}\tup{y})}{\formWit(\tup{p}\tup{y}\tup{y}'z_1z_2)}{\formOut(\tup{p})}\\
		\intertext{in $\formB$. It is important that $\tup{p}$ is exactly the tuple of free variables of $\formA$ and not more.
		The formula $\formA$ is inductively replaced with
		}
		\formA''(\tup{p}) &:= \ifpwscbig{R,\tup{x}}{R^*}{S,\tup{y}, \tup{y}'}{z_1z_2}{\formStep'(\tup{p}\tup{x})}{\formChoice'(\tup{p}\tup{y})}{\formWit'(\tup{p}\tup{y}\tup{y}'z_1z_2)}{\formOut'(\tup{p})}\\
		\formA'(\tup{p}) &:=\iop{\interpret^{\sigOf{\formA}}_{\tup{p}}}{\forall \tup{p}.\qspace P(\tup{p}) \Rightarrow \formA''},
	\end{align*}
	where~$\formA_{\text{x}}'$ is obtained by induction for~$\formA_\text{x}$
	for all $\text{x} \in \set{\text{step},\text{choice},\text{wit},\text{out}}$.
	In this way,~$\formA'$ simulates the subformula reduct semantics.
	Adding the relation~$P$ does not change the orbits of the structure
	because evaluating the WSC-fixed-point operator considers orbits stabilizing the vertices assigned to~$\tup{p}$.
	Also note that $\sigOf{\formA'} = \sigOf{\interpret^{\sigOf{\formA}}_{\tup{p}}} = \sigOf{\formA}$.

	We finally obtain an \IFPCWSCI{}-formula~$\formB'$
	such that $\sigOf{\formB'} = \formB$
	and $\formB'$ with non-reduct semantics is equivalent to $\formB$ with subformula reduct semantics for all $\sigOf{\formB}$-structures.
	We now evaluate $\formB'$ as a \IFPCWSCIglob{}-formula under global reduct semantics.
	On $\sigOf{\formB}$-structures,
	the global reduct semantics does not move to a proper reduct
	and~$\formB$ and~$\formB'$ are equivalent on these structures.
	Since both \IFPCWSCIglob{} and \IFPCWSCsub{} have the reduct property,
	equivalence for  $\sigOf{\formB}$-structures implies equivalence for all structures.

	To prove the second inclusion, the global reduct semantics is expressible
	in $\IFPCWSCI$ in a similar way: a formula~$\formB(\tup{p})$ with global reduct semantics
	is equivalent to the formula $\iop{\interpret^{\sigOf{\formA}}_{\tup{p}}}{\forall \tup{p}.\qspace P(\tup{p}) \Rightarrow \formB}$ with non-reduct semantics.
\end{proof}
Although \IFPCWSCI{} expresses global and subformula reduct semantics,
the logic still does not have the reduct property for the same reasons as for \IFPCWSC{}.
So there is a need to consider \IFPCWSCIglob{}.
Lemma~\ref{lem:interpretation-operator-expresses-sub-and-glob} shows that we do not need a subformula reduct semantics version
of \IFPCWSCI{}
since it is expressible by \IFPCWSCIglob{}.
In the proof, the translation potentially increases
the alternation depth between WSC-fixed-point and interpretation operators.
Whether this is indeed necessary, will be one of the concerns in this article
to understand the interplay between interpretations and witnessed symmetric choice.

In the following, we will focus on \IFPCWSC{} and \IFPCWSCI{} with the non-reduct semantics.
Compared to the global reduct semantics, this will mostly make no difference:
we are interested in properties of structures with a fixed signature.
Here, definability under non-reduct semantics implies definability under the global one.
\begin{lem}
	\label{lem:non-reduct-in-global-reduct-semantics-for-fixed-signature}
	Let $\GraphClass$ be some class of $\sig$-structures for some fixed signature $\sig$.
	Then  $\IFPCWSC \leq \IFPCWSCglob{}$ and
	$\IFPCWSCI \leq \IFPCWSCIglob{}$ on $\GraphClass$.
\end{lem}
\begin{proof}
	Let $\formB$ be a tautology such that $\sigOf{\formB} = \sig$.
	Then a formula $\formA$ under non-reduct semantics is equivalent to $\formA \land \formB$ under global reduct semantics on $\GraphClass$.
\end{proof}
When proving undefinability of a property,
we will have to argue that proper reducts
do not make property definable.
For the properties of interest, we will show that discarding relations
\begin{enumerate}
	\item makes structures with and without the property isomorphic (and thus the property cannot be defined), or
	\item does not change to orbit partition of the structures
	(and hence WSC-fixed point operator will evaluate equally).
\end{enumerate}

\section{The CFI Construction}
\label{sec:cfi-construction}
In this section we recall the CFI construction that was introduced by Cai, Fürer, and Immerman~\cite{CaiFI1992}.
At the heart of the construction are the so-called CFI gadgets.
These gadgets are used to obtain from a base graph two non-isomorphic CFI graphs.

\paragraph{The CFI Gadget}
The \defining{degree-$d$ CFI gadget} consists of $d$ pairs of \defining{edge vertices} $\set{a_{i0}, a_{i1}}$, one for each $i \in [d]$
and the set of \defining{gadget vertices} $\setcond{\tup{b} \in \FF_2^d}{b_1+\dots+b_d =0}$ (the names of the vertices will become more clear in the next paragraph).
There is an edge between $a_{ij}$ and $\tup{b}$ if and only if
$b_i = j$.
Edge vertices and gadget vertices receive different colors.
If, using $d$ additional colors, every edge-vertex-pair $\set{a_{i0}, a_{i1}}$ 
receives its own color for every $i \in[d]$,
then the CFI gadget realizes precisely the automorphisms
exchanging the vertices $\set{a_{i0}, a_{i1}}$ for an even number of $i \in[d]$.

We also need a variant of the CFI gadgets that does not use the gadget vertices
but has a $d$-ary relation instead.
Every gadget vertex $\tup{b}$ is replaced by the $d$-tuple $(a_{1b_1}, \dots, a_{db_d})$ in the said relation.
This gadget has the same automorphisms (with respect to the edge vertices).
It has the benefit that not the whole gadget can be fixed by fixing a single
gadget vertex but has the drawback that the arity of the relations depends on the degree.
For an overview of the different variants of CFI gadgets we refer to~\cite{Lichter2023}.

\paragraph{CFI Graphs}
A \defining{base graph} is a connected, simple, and possibly colored graph.
Let ${G=(V,E,\spleq)}$ be a colored base graph.
In the context of CFI graphs,
we call the vertices~$V$ of~$G$ \defining{base vertices},
call its edges~$E$ \defining{base edges},
and use fraktur letters for base vertices or edges.
For a function $f\colon E \to \FF_2$,
we construct the \defining{CFI graph} $\CFI{G, f}$ as follows:
First, replace every vertex of~$G$ by a CFI gadget of the same degree.
In that way, we obtain for every base edge $\set{\bVertA,\bVertB} \in E$
two edge-vertex-pairs $\set{a_{i0},a_{i1}}$ and $\set{a'_{j0},a'_{j1}}$.
The first one is given by the gadget of~$\bVertA$ and the second one by the gadget of~$\bVertB$ (which justifies the name ``edge vertex'' even if they are part of a CFI gadget).
Second, add edges such that $\set{a_{ik},a'_{j\ell}}$ is an edge if and only if
$k + \ell = f(\set{\bVertA,\bVertB})$.
We say that the edge vertices $\set{a_{i0},a_{i1}}$ \defining{originate} from $(\bVertA,\bVertB)$,
the edge vertices $\set{a'_{j0},a'_{j1}}$ originate from $(\bVertB,\bVertA)$, respectively, and
that the gadget vertices of the gadget for a base vertex~$\bVertA$ 
\defining{originate} from~$\bVertA$.
The color of edge and gadget vertices is obtained from the color of its origin given by~$\spleq$.

\def \baseNeighborRadius{1}
\def \edgePairDistAngle{10}
\def \edgePairRadius{2}
\def \gadgetVertexYDist{0.7}

\colorlet{edgeColA}{blue!50}
\colorlet{edgeColA}{edgeColA}
\colorlet{edgeColB}{red!50}
\colorlet{edgeColC}{green!50}
\colorlet{edgeColE}{yellow!70!black!50!white}
\colorlet{edgeColD}{magenta!50}
\colorlet{edgeColF}{cyan!50}

\tikzstyle{basevertex} = [draw=black, circle, inner sep=2pt, minimum size = 12pt]
\tikzstyle{edgevertex} = [draw=black, circle, inner sep=2pt, font=\footnotesize]
\tikzstyle{gadgetvertex} = [draw=black, circle, inner sep=1, font=\tiny, fill = gray!50]

\tikzset{
	bicolor fill/.code 2 args={
		\pgfdeclareverticalshading[%
		tikz@axis@top,tikz@axis@middle,tikz@axis@bottom%
		]{diagonalfill}{100bp}{%
			color(0bp)=(tikz@axis@bottom);
			color(60bp)=(tikz@axis@bottom);
			color(60bp)=(tikz@axis@middle);
			color(60bp)=(tikz@axis@top);
			color(100bp)=(tikz@axis@top)
		}
		\tikzset{shade, left color=#1, right color=#2, shading=diagonalfill}
	}
}

\begin{figure}
	\centering
	
	\subfloat[A part of a colored base graph.]{
		\begin{tikzpicture}
			
			\draw[draw = none, use as bounding box] (-3.6,-1.8) rectangle (4.6, 1.8);

			\node[basevertex, fill=edgeColA] (u1) at (0,0){};
			\node[basevertex, fill = edgeColB] (u2)  at (0:\baseNeighborRadius){};
			\node[basevertex, fill = edgeColC] (u3)  at (220:\baseNeighborRadius){};
			\node[basevertex, fill = edgeColD] (u4)  at (140:\baseNeighborRadius){};
			\begin{scope}[shift={(1,0)}, xscale=-1]
				\node[basevertex, fill = edgeColE] (u5)  at (220:\baseNeighborRadius){};
				\node[basevertex, fill = edgeColF] (u6)  at (140:\baseNeighborRadius){};
			\end{scope}
			
			\draw [-, thick]
			(u1) to (u2)
			(u1) to (u3)
			(u1) to (u4)
			(u2) to (u5)
			(u2) to (u6);
		\end{tikzpicture}
		\label{fig:cfi-gadget-edge-sub-base-graph}
	}%
	\hspace{1em}%
	\newcommand{\reducedstrut}{\vrule width 0pt height .9\ht\strutbox depth .9\dp\strutbox\relax}
	\newcommand{\pink}[1]{%
		\begingroup
		\setlength{\fboxsep}{1pt}%
		\colorbox{red!20}{#1}%
		\endgroup
	}
	\subfloat[Two connected CFI gadgets.]{
		\begin{tikzpicture}
			
			\draw[draw = none,use as bounding box] (-2.3,-1.8) rectangle (8.3, 1.8);

			\begin{scope}[name prefix=a]
				
				\node[edgevertex, bicolor fill={edgeColA}{edgeColB}] (v10) at (+\edgePairDistAngle:\edgePairRadius) {$0$};
				\node[edgevertex, bicolor fill={edgeColA}{edgeColB}] (v11) at (-\edgePairDistAngle:\edgePairRadius) {$1$};
				\node[edgevertex, bicolor fill={edgeColA}{edgeColC}, shading angle=310] (v20) at (220-\edgePairDistAngle:\edgePairRadius) {$0$};
				\node[edgevertex, bicolor fill={edgeColA}{edgeColC}, shading angle=310] (v21) at (220+\edgePairDistAngle:\edgePairRadius) {$1$};
				\node[edgevertex, bicolor fill={edgeColA}{edgeColD}, shading angle=230] (v30) at (140+\edgePairDistAngle:\edgePairRadius) {$0$};
				\node[edgevertex, bicolor fill={edgeColA}{edgeColD}, shading angle=230] (v31) at (140-\edgePairDistAngle:\edgePairRadius) {$1$};
				
				\setlength{\fboxsep}{0pt}
				
				\node[gadgetvertex, fill=edgeColA] (u000) at (0.4, 1.5*\gadgetVertexYDist) {$000$};
				\node[gadgetvertex, fill=edgeColA] (u011) at (0.4, 0.5*\gadgetVertexYDist) {$110$};
				\node[gadgetvertex, fill=edgeColA] (u101) at (0.4, -0.5*\gadgetVertexYDist) {$101$};
				\node[gadgetvertex, fill=edgeColA] (u110) at (0.4, -1.5*\gadgetVertexYDist) {$011$};
				
				\path [-, draw=black, thick]
				(u000) to (v10)
				(u000) to (v20)
				(u000) to (v30)
				(u011) to (v10)
				(u011) to (v21)
				(u011) to (v31)
				(u101) to (v11)
				(u101) to (v20)
				(u101) to (v31)
				(u110) to (v11)
				(u110) to (v21)
				(u110) to (v30);
				
			\end{scope}

			\begin{scope}[name prefix=b, shift={(6,0)}, xscale = -1]
				
				\node[edgevertex, bicolor fill={edgeColB}{edgeColA}, shading angle = -90] (v10) at (+\edgePairDistAngle:\edgePairRadius) {$0$};
				\node[edgevertex, bicolor fill={edgeColB}{edgeColA}, shading angle=-90] (v11) at (-\edgePairDistAngle:\edgePairRadius) {$1$};
				\node[edgevertex, bicolor fill={edgeColB}{edgeColE}, shading angle=50] (v20) at (220-\edgePairDistAngle:\edgePairRadius) {$0$};
				\node[edgevertex, bicolor fill={edgeColB}{edgeColE}, shading angle=50] (v21) at (220+\edgePairDistAngle:\edgePairRadius) {$1$};
				\node[edgevertex, bicolor fill={edgeColB}{edgeColF}, shading angle=130] (v30) at (140+\edgePairDistAngle:\edgePairRadius) {$0$};
				\node[edgevertex, bicolor fill={edgeColB}{edgeColF}, shading angle=130] (v31) at (140-\edgePairDistAngle:\edgePairRadius) {$1$};

				
				\node[gadgetvertex, fill=edgeColB] (u000) at (0.4, 1.5*\gadgetVertexYDist) {$000$};
				\node[gadgetvertex, fill=edgeColB] (u011) at (0.4, 0.5*\gadgetVertexYDist) {$011$};
				\node[gadgetvertex, fill=edgeColB] (u101) at (0.4, -0.5*\gadgetVertexYDist) {$101$};
				\node[gadgetvertex, fill=edgeColB] (u110) at (0.4, -1.5*\gadgetVertexYDist) {$110$};
				
				\path [-, draw=black, thick]
				(u000) to (v10)
				(u000) to (v20)
				(u000) to (v30)
				(u011) to (v10)
				(u011) to (v21)
				(u011) to (v31)
				(u101) to (v11)
				(u101) to (v20)
				(u101) to (v31)
				(u110) to (v11)
				(u110) to (v21)
				(u110) to (v30);
				
			\end{scope}
			
			\path[thick, -, draw = black]
			(av10) to (bv10)
			(av11) to (bv11);

		\end{tikzpicture}
		\label{fig:cfi-gadget-edge-sub-gadgets}
	}
	\caption[CFI gadgets with gadget and edge vertices]{Construction of a CFI graph from a base graph:
		Figure~\protect\subref{fig:cfi-gadget-edge-sub-base-graph} shows a part of a colored base graph~$G$.
		Figure~\protect\subref{fig:cfi-gadget-edge-sub-gadgets}
		shows the two gadgets for the red and the blue base vertex and in the CFI graph $\CFI{G,f}$.
		The figure  assumes that $f$ sets their incident edges to $0$.
		Gadget vertices inherit the color from their origins.
		The origin of an edge vertex is an ordered pair $(\bVertA,\bVertB)$ of base vertices and an edge vertex inherits the color of this pair.
		This color is  by asymmetrically coloring the edge vertices in the figure by the colors of $\bVertA$ and $\bVertB$.
		The $0/1$-triples in the gadget vertices refer to their adjacent edge vertices: for the order of entries in such a tuple in the left gadget,
		we use the order  $\textcolor{edgeColD}{\bullet} < \textcolor{edgeColC}{\bullet} < \textcolor{edgeColB}{\bullet}$
		on the neighbors of the blue base vertex.
		For the right gadget, the figure uses
		$\textcolor{edgeColA}{\bullet} < \textcolor{edgeColE}{\bullet} < \textcolor{edgeColF}{\bullet}$.
	}
	\label{fig:cfi-gadget-edge}
\end{figure}

It is well-known~\cite{CaiFI1992} that, for every $f,g \colon E \to \FF_2$,
we have $\CFI{G,g} \iso \CFI{G,f}$ if and only if $\sum g := \sum_{\bEdge\in E} g(\bEdge) = \sum f$.
Hence, if we are only interested in the graph up to isomorphism,
we also write $\CFI{G,0}$ and $\CFI{G,1}$.
A CFI graph $\CFI{G,f}$ is \defining{even} if ${\sum f = 0}$ 
and \defining{odd} otherwise.
A base edge $\bEdge \in E$ is called \defining{twisted} by $f$ and $g$ if ${g(\bEdge) \neq f(\bEdge)}$.
Twisted edges can be ``moved around'' using path isomorphisms (see e.g.~\cite{Lichter2023}):
If $\bVertA_1, \dots, \bVertA_\ell$ is a path in~$G$ for $\ell > 2$,
then there is an isomorphism $\autoA \colon \CFI{G,g} \to \CFI{G,g'}$,
where $g'(\bEdge) = g(\bEdge)$ apart from $\bEdge_1:=\set{\bVertA_1,\bVertA_2}$ and $\bEdge_2:= \set{\bVertA_{\ell-1}, \bVertA_\ell}$ which satisfy $g'(\bEdge_i) = g(\bEdge_i) +1$.
The isomorphism~$\autoA$ is the identity on all vertices apart from 
the gadget vertices, whose origin is contained in $\bVertA_2, \dots, \bVertA_{\ell-1}$, and from edge edges, whose origin is $(\bVertA_i,\bVertA_{i+1})$ or $(\bVertA_{i+1},\bVertA_{i})$ for $2\leq i < \ell-1$.
For an edge-vertex-pair of such an origin, $\autoA$ exchanges the two vertices of that pair.
This determines the image of the gadget vertices.
If $\bVertA_1, \dots, \bVertA_\ell$ is actually a cycle, we obtain an automorphism in this way.
If $G$ is totally ordered,
then every isomorphism is composed of path-isomorphisms
and
every automorphism of $\CFI{G,g}$ is the composition of such cycle-automorphisms.
For a class of base graphs~$\GraphClass$,
set
\[\CFI{\GraphClass} := \setcond[\big]{\CFI{G,g}}{G=(V,E) \in \GraphClass, g\colon E\to \FF_2}\]
to be the class of CFI graphs over $\GraphClass$.
The \defining{CFI query} for $\CFI{\GraphClass}$
is to decide whether a given CFI graph in $\CFI{\GraphClass}$ is even.
We collect some properties of CFI graphs:
\begin{lemC}[{\cite[Theorem~3]{DawarRicherby07}}]
	\label{lem:k-connected-treewidth-CFI-equiv}
	If $G$ is of  minimum degree $2$ and has treewidth at least $k$,
	in particular if $G$ is $k$-connected,
	then $\CFI{G,0} \kequiv{k} \CFI{G,1}$.
\end{lemC}
In particular, the CFI query for a class of base graphs
of unbounded treewidth or connectivity is not \IFPC{}-definable.
We now consider orbits of CFI graphs.
The following lemma is well-known: 
\begin{lem}
	\label{lem:edge-vertices-same-orbit}
	Let $G=(V,E,\spleq)$ be a colored base graph,
	let $\set{\bVertA,\bVertB} \in E$,
	let $\StructA = \CFI{G, f}$ for some $f\colon E\to \FF_2$,
	let $\tup{\vertC} \in \StructVA^*$ be an $\ell$-tuple of edge vertices,
	and let the origin of $\vertC_i$ be $(\bVertC_i,\bVertC'_i)$ for every $i \in [\ell]$.
	Then both edge vertices with origin $(\bVertA,\bVertB)$ are in the same $1$-orbit of $(\StructA,\tup{\vertC})$
	if and only if $\set{\bVertA,\bVertB}$ is part of a cycle in $G - \setcond{\set{\bVertC_i,\bVertC'_i}}{i \in [|\tup{\vertC}|]}$.
\end{lem}
\begin{proof}
	Set $E' := \setcond{\set{\bVertC_i,\bVertC'_i}}{i \in [\ell]}$.
	Assume that there is a cycle in $G -E'$
	containing $\set{\bVertA,\bVertB}$.
	Then we can use a cycle-automorphism for that cycle
	to exchange the two edge vertex pairs with origin $(\bVertA,\bVertB)$.
	Because this automorphism is the identity on all vertices apart from the ones whose origin is contained in the cycle,
	it in particular fixes~$\tup{\vertC}$.
	
	For the other direction, assume that there is an automorphism~$\autoA$ mapping
	one edge vertex with origin $(\bVertA,\bVertB)$ to the other one.
	Then in particular~$\autoA$ has to exchange both.
	We can assume that~$\autoA$ is base-vertex-respecting,
	that is,~$\autoA$ maps vertices to vertices of the same origin.
	If~$\autoA$ was not base-vertex-respecting,~$\autoA$ induces
	a non-trivial automorphism of the base graph,
	whose inverse can be combined with~$\autoA$ to a base-vertex-respecting automorphism (see~\cite{Pago23}).
	But a base-vertex-respecting automorphism of $(\StructA, \tup{\vertC})$ is composed out of cycle-automorphisms
	not using the edges~$E'$.
	So there is in particular a single cycle in $G -E'$ containing $\set{\bVertA,\bVertB}$.
\end{proof}
We sketch the proof of the following lemma
to illustrate the requirement of high connectivity on the base graph.

\begin{lemC}[{\cite[Lemma~3.14]{GradelPakusa19}}]
	\label{lem:orbits-connectivity}
	Let $G=(V,E,\leq)$ be an ordered and $(k+2)$-connected base graph
	and let $\StructA = \CFI{G, f}$ for some $f\colon E\to \FF_2$.
	Let $\tup{\vertC} \in \StructVA^{\leq k}$ 
	and $\set{\bVertA,\bVertB} \in E$ be a base edge
	such that no vertex in $\tup{\vertC}$ has origin~$\bVertA$,~$\bVertB$, $(\bVertA,\bVertB)$, or $(\bVertB,\bVertA)$.
	Then the two edge vertices with origin $(\bVertA,\bVertB)$ are contained in the same orbit of $(\Struct, \tup{\vertC})$.
\end{lemC}
\begin{proof}
	We first assume that $\tup{\vertC}$ only consists of gadget vertices.	
	Let $V_{\tup{\vertC}} \subseteq V$ be the set of origins of all vertices in $\tup{\vertC}$.
	Because every vertex in~$G$ has degree at least $k+2$
	(since $G$ is $(k+2)$-connected),
	the vertices~$\bVertA$ and~$\bVertB$
	have degree at least $2$ in $G\setminus V_{\tup{\vertC}}$.
	Because $G$ is $(k+2)$-connected, there is a $\bVertA$-$\bVertB$-path in $G\setminus V_{\tup{\vertC}}$ not using the edge $\set{\bVertA,\bVertB}$
	(removing  $\bVertB$ from $G\setminus V_{\tup{\vertC}}$ removes at most $k+1$ vertices from $G$).
	So there is a cycle in $G\setminus V_{\tup{\vertC}}$ using the edge $\set{\bVertA,\bVertB}$ and thus there is
	an automorphism exchanging the two edge vertices with origin~$(\bVertA,\bVertB)$.
	
	Now assume that there is an edge vertex with origin $(\bVertA, \bVertB)$
	in $\tup{\vertC}$.
	Let $\tup{\vertC}'$ be obtained from $\tup{\vertC}$
	by replacing this vertex with a gadget vertex with origin $\bVertA$.
	Every automorphism fixing $\tup{\vertC}'$ also fixes $\tup{\vertC}$.
\end{proof}

\section{Canonization of CFI Graphs in \IFPCWSCI{}}
\label{sec:canonize-cfi-graphs}

In this section we show that with respect to canonization
the CFI construction ``loses its power'' in \IFPCWSCI{}
in the sense that canonizing CFI graphs is not harder
than canonizing the base graphs.
In the following, we work with a class of base graphs closed under
individualization.
Intuitively, this means that the class is closed under assigning some
vertices unique new colors.
We adapt a result of~\cite{LichterSchweitzer24} from \CPTWSC{} to \IFPCWSC{}
and show that it suffices to define a single orbit
for every individualization of a graph
to obtain a definable canonization.
The canonization approach adapts Gurevich's canonization algorithm~\cite{Gurevich97} and requires the WSC-fixed-point operator.
We then show that once we define orbits of (the closure under individualization of the) base graphs,
we can define orbits of CFI graphs
and hence canonize them.
Here we need the interpretation operator to reduce a CFI graph to its base graph, which is needed to define its orbits.

Intuitively, individualizing vertices
is just a tuple of parameters.
However, the number of vertices to individualize is not bounded
and we need to encode them via relations in the first-order setting.

\begin{defi}[Individualization of Vertices]
	Let $\Struct$ be a relational structure.
	A binary relation ${\indRel^\Struct} \subseteq \StructV^2$ 
	is an \defining{individualization} of a set of vertices $W \subseteq \StructVA$
	if~$\indRel^\StructA$ is a total order on~$W$
	and~${\indRel^\Struct} \subseteq W^2$.
	We say that $\indRel^\Struct$ is an individualization
	if it is an individualization of some $W \subseteq \StructVA$
	and that the vertices in $V$ are \defining{individualized} by $\indRel^\Struct$.
\end{defi}

The relation $\indRel^\Struct$ defines a total order on the individualized vertices.
Intuitively, it assigns unique colors to these vertices.
Instead of $(\Struct, \indRel^\Struct)$,
one can think of $(\Struct, \tup{\vertA})$,
where  $\vertA_1 \indRel^\StructA \cdots \indRel^\StructA \vertA_{|\tup{\vertA}|}$
are the vertices individualized by~$\indRel^\StructA$.

\begin{defi}[Closure under Individualization]
	Let $\GraphClass$ be a class of $\sig$-structures.
	The \defining{closure under individualization} $\indClosure{\GraphClass}$ of $\GraphClass$
	is the class
	of $(\sig \disunion \set{\indRel})$-structures
	such that ${(\StructA, \indRel^\StructA) \in \indClosure{\GraphClass}}$
	for every $\StructA \in \GraphClass$
	and every individualization $\indRel^\StructA$.
\end{defi}

In the following,
let $L$ be one of the logics \IFPC{}, \IFPCWSC{}, or \IFPCWSCI{}.
We adapt some notions related to canonization from~\cite{LichterSchweitzer24}
to the first-order setting.
Note that all following definitions implicitly include the closure under individualization.
\begin{defi}[Canonization]
	Let $\GraphClass$ be a class of $\sig$-structures.
	An \defining{$L$-canonization} for~$\GraphClass$ is an
	$L[\sig\disunion\set{\indRel}, \sig\disunion\set{\indRel,\leq}]$-interpretation $\interpret$
	satisfying the following:
	\begin{enumerate}
		\item $\leq^{\interpret(\StructA)}$ is a total order on $\interpret(\StructA)$ for every $\StructA \in \indClosure{\GraphClass}$,
		\item $\StructA \iso \reduct{\interpret(\StructA)}{(\sig\disunion\set{\indRel})}$ for every $\StructA \in \indClosure{\GraphClass}$, and
		\item $\interpret(\StructA) \iso \interpret(\StructB)$ if and only if $\StructA \iso \StructB$
		for every $\StructA,\StructB \in \indClosure{\GraphClass}$.
	\end{enumerate}
	The structure $\interpret(\Struct)$ is called the \defining{$\interpret$-canon} (or just the canon if unambiguous) of $\Struct$.
	We say that $L$ \defining{canonizes} $\GraphClass$
	if there is an $L$-canonization for $\GraphClass$.
\end{defi}

We make same remarks to the former definition.
Condition~3 requires isomorphism of \emph{ordered} structures.
For a logic $L$ possessing numbers (as in our case),
Condition~3 can equivalently be stated with equality.
While Condition~3 is essential for algorithmic canonizations,
it is implied by Condition~2 for $L$-definable canonizations:
Let $\StructA, \StructB \in \indClosure{\GraphClass}$.
If~$\StructA \iso \StructB$,
then $\interpret(\StructA) \iso \interpret(\StructB)$
because~$L$ is isomorphism-invariant.
If $\interpret(\StructA) \iso \interpret(\StructB)$,
then by Condition~2,
$\StructA \iso \reduct{\interpret(\StructA)}{(\sig\disunion\set{\indRel})}
\iso \reduct{\interpret(\StructB)}{(\sig\disunion\set{\indRel})} \iso \StructB$.

\begin{defi}[Distinguishable $k$-Orbits]
	The logic $L$ \defining{distinguishes the $k$-orbits} of a class of $\sig$-structures $\GraphClass$,
	if there is an $L[\sig\disunion\set{\indRel}]$-formula $\formA(\tup{\uniVarA},\tup{\uniVarB})$
	that has~$2k$ free variables~$\tup{\uniVarA}$ and~$\tup{\uniVarB}$ such that $|\tup{\uniVarA}| = |\tup{\uniVarB}| = k$
	and that defines, for every $\StructA \in \indClosure{\GraphClass}$, a total preorder on $\StructVA^k$ 
	whose equivalence classes coincide with the $k$-orbit partition of~$\StructA$.
\end{defi}
Note that because $\formA$ defines a total preorder,
$L$ does not only define the $k$-orbit partition
but also orders the $k$-orbits.

\begin{defi}[Ready for Individualization]
	A class of $\sig$-structures $\GraphClass$ is \defining{ready for individualization in $L$}
	if there is an $L[\sig\disunion\set{\indRel}]$-sentence $\formA$
	defining, for every $\Struct \in \indClosure{\GraphClass}$,
	a set of vertices $O = \formA^\Struct$ such that
	\begin{itemize}
		\item $O$ is a $1$-orbit of $\Struct$,
		\item $|O|>1$ if $\Struct$ has a non-trivial $1$-orbit, and
		\item if $O =\set{\vertA}$ is a singleton set,
		then $\vertA$ is not individualized by $\indRel^\Struct$ unless $\indRel^\Struct$ individualizes~$\StructV$, i.e., all vertices.
	\end{itemize}
\end{defi}

The following is a statement similar to \CPTWSC{} in~\cite{LichterSchweitzer24}
and the proof is analogous (note that~\cite{LichterSchweitzer24} includes definable isomorphism, which we do not do here):
\begin{lem}
	\label{lem:canon-orbit-ready-iff}
	Let $L$ be one of the logics \IFPCWSC{} and \IFPCWSCI{} and
	let $\GraphClass$ be a class of $\sig$-structures.
	The following are equivalent:
	\begin{enumerate}
		\item $L$ defines a canonization for $\GraphClass$.
		\item $L$ distinguishes the $k$-orbits of $\GraphClass$ for every positive $k \in \nat$.
		\item $\GraphClass$ is ready for individualization in $L$.
	\end{enumerate}
\end{lem}
\begin{proof}
	We show (1) $\Rightarrow$ (2) $\Rightarrow$ (3) $\Rightarrow$ (1).
	To show (1) $\Rightarrow$ (2),
	two $k$-tuples are ordered
	according to the lexicographical order on the canons
	when individualizing these tuples.
	For (2) $\Rightarrow$ (3),
	one orders the $1$-orbits and picks the minimal (according to that order)
	non-trivial orbit.
	If such an orbit does not exist, the minimal
	singleton $1$-orbit whose vertex is not individualized is picked.
	Finally, to show (3) $\Rightarrow$ (1),
	one defines a variant of Gurevich's canonization algorithm~\cite{Gurevich97}
	using a WSC-fixed-point operator.
	This is done exactly as in~\cite{LichterSchweitzer24} for \CPTWSC{}.
	We sketch the approach here and refer to the original work for further details.
	Vertices are individualized iteratively: distinguish the orbits of the input structure with the current individualization, 
	choose a vertex from the defined orbit, and
	individualize the chosen one.
	This procedure is repeated until all vertices are individualized.
	Then the individualization is a total order on the input.
	Because all choices were made from orbits, this order defines a canonization.
	To witness the choices, this approach is essentially applied again,
	starting from two vertices~$u$ and~$v$ in the same orbit.
	Intuitively,  each of both vertices gets individualized in a ``copy''
	of the current structure.
	For both copies, we obtain an order in the way before.
	Since both copies are isomorphic because~$u$ and~$v$ are in the same orbit,
	these two orders induce an isomorphism mapping~$u$ to~$v$.
	This isomorphism can be turned into an  automorphism of the original structure
	witnessing that~$u$ and~$v$ are in the same orbit.
	Because the witnessing formula has access to the total order initially obtained using choices,
	no further choices are needed to define these automorphism-inducing orders.
	In particular, it suffices to use one WSC-fixed-point operator.
	It is easy to see there that once the formula defining the orbit is an $L$-formula,
	this approach can easily be expressed in~$L$.
\end{proof}

While stating Lemma~\ref{lem:canon-orbit-ready-iff}
requires some technical definitions,
it simplifies defining canonization using WSC-fixed-point operators.
Their use is hidden in Gurevich's canonization algorithm:
Witnessing automorphisms do not have to be defined explicitly.

\begin{lem}
	\label{lem:CFI ready-individualization}
	Let $\GraphClass$ be a class of colored base graphs.
	\hspace{-0.5pt}If \IFPCWSCI{} distinguishes $2$-orbits of $\GraphClass$,
	then $\CFI{\GraphClass}$ is ready for individualization in \IFPCWSCI{}.
\end{lem}
\begin{proof}
	Let $G=(V,E,\spleq) \in \GraphClass$,
	let $g\colon E \to \FF_2$,
	and let $\Struct = (\CFI{G,g}, \indRel^\Struct)$,
	where $\indRel^\Struct$ individualizes some vertices of
	$\CFI{G,g}$.
	We assume that only edge vertices are individualized.
	Instead of individualizing a gadget vertex with origin $\bVertA$,
	one can individualize one edge vertex
	with origin $(\bVertA, \bVertB)$ for every $\bVertB \in \neighbors{G}{\bVertA}$, namely the neighbors of the gadget vertex.
	This translation  is \IFPC{}-definable in both directions.
	
	We denote by $\orig{\indRel^\Struct}$
	the set of all (directed) base edges $(\bVertA,\bVertB)$
	that are the origin of some edge vertex individualized by $\indRel^\Struct$.
	We denote by $G-\orig{\indRel^\Struct}$
	the graph obtained from~$G$ by deleting
	the edges in $\orig{\indRel^\Struct}$ viewed as undirected edges.
	We can turn $\indRel^\Struct$ into an individualization for $G$:
	the individualization $\indRel^\Struct$ defines a total order on $\orig{\indRel^\Struct}$ 
	and for each directed base edge $(\bVertA,\bVertB) \in \orig{\indRel^\Struct}$ we individualize $\bVertA$ and $\bVertB$.
	In that way, we denote by $(G, \orig{\indRel^\Struct})$
	the graph~$G$ with this individualization.
	
	We analyze the cases in which there are non-trivial  $1$-orbits of $(\Struct, \indRel^\Struct)$.
	If there is a $2$-orbit 
	of $(G, \orig{\indRel^\Struct})$
	containing base edges part of a cycle in 
	$G -  \orig{\indRel^\Struct}$,
	we obtain a non-trivial $1$-orbit of $(\Struct, \indRel^\Struct)$ as follows:
	\begin{clm}
		\label{clm:orbit-cycle}
		Let $O$ be a $2$-orbit of $(G, \orig{\indRel^\Struct})$.
		If every (directed) edge in $O$ is part of a cycle in $G - \orig{\indRel^\Struct}$,
		then the set of edge-vertex-pairs
		$\setcond{\vertA}{\text{the origin of } \vertA \text{ is in } O}$
		is a $1$-orbit of $(\StructA,\indRel^\Struct)$.
	\end{clm}
	\begin{claimproof}
		Because $O$ is a $2$-orbit,
		either all or none of the directed base edges in $O$ are part of a cycle in $G - \orig{\indRel^\Struct}$.
		By Lemma~\ref{lem:edge-vertices-same-orbit},
		the edge-vertex-pairs with origin in $O$ are in the same $1$-orbit.
		Using an automorphism of the base graph,
		the edge vertices with origin $(\bVertA,\bVertB) \in O$
		can be mapped to the edge vertices with origin $(\bVertA',\bVertB') \in O$
		for every $(\bVertA,\bVertB),(\bVertA',\bVertB') \in O$.
		That is, $\setcond{\vertA}{\text{the origin of } \vertA \text{ is in } O}$
		is a subset of a $1$-orbit of $(\StructA,\indRel^\Struct)$.
		It cannot be a strict subset
		because then an edge vertex with origin in $O$
		has to be mapped to an edge vertex with origin not in $O$,
		which contradicts that $O$ is an orbit.
	\end{claimproof}

	In the case that no such $2$-orbit exists, 
	there are possibly other non-trivial $1$-orbits.
	They arise from automorphisms of the base graph.	
	An \defining{edge-vertex-pair-order} of a set of base edges  $E' \subseteq E$
	is a set of edge vertices~$R$ 
	such that, for every edge-vertex-pair $\set{\vertA_1,\vertA_2}$ with origin $(\bVertA,\bVertB)$ such that $\set{\bVertA,\bVertB} \in E'$,
	exactly one of~$\vertA_1$ and~$\vertA_2$ is contained in~$R$.
	Intuitively,~$R$ defines an order per edge-vertex-pair with origin in~$E'$,
	but does not order edge-vertex-pairs of different origins.
	
		\begin{figure}
		\centering
		\begin{tikzpicture}

			\begin{scope}[name prefix=a]
				
				\node[edgevertex, fill=gray!30!white] (v10) at (+\edgePairDistAngle:\edgePairRadius) {$0$};
				\node[edgevertex, fill=red!15!white, draw =red, ultra thick, dotted] (v11) at (-\edgePairDistAngle:\edgePairRadius) {$1$};
				\node[edgevertex, fill=red!30!white, draw =red, ultra thick] (v20) at (220-\edgePairDistAngle:\edgePairRadius) {$0$};
				\node[edgevertex, fill=gray!30!white] (v21) at (220+\edgePairDistAngle:\edgePairRadius) {$1$};
				\node[edgevertex, fill=gray!30!white] (v30) at (140+\edgePairDistAngle:\edgePairRadius) {$0$};
				\node[edgevertex, fill=red!30!white, draw =red, ultra thick] (v31) at (140-\edgePairDistAngle:\edgePairRadius) {$1$};
				
				\node[gadgetvertex, fill=gray] (u000) at (0.4, 1.5*\gadgetVertexYDist) {$000$};
				\node[gadgetvertex, fill=gray] (u011) at (0.4, 0.5*\gadgetVertexYDist) {$110$};
				\node[gadgetvertex, fill=blue!30!white, draw =blue, ultra thick] (u101) at (0.4, -0.5*\gadgetVertexYDist) {$101$};
				\node[gadgetvertex, fill=gray] (u110) at (0.4, -1.5*\gadgetVertexYDist) {$011$};
				
				\path [-, draw=black, thick]
				(u000) to (v10)
				(u000) to (v20)
				(u000) to (v30)
				(u011) to (v10)
				(u011) to (v21)
				(u011) to (v31)
				(u101) to (v11)
				(u101) to (v20)
				(u101) to (v31)
				(u110) to (v11)
				(u110) to (v21)
				(u110) to (v30);

			\end{scope}

		\end{tikzpicture}
		\caption[Extending Edge-Vertex-Pair-Orders]{
			Extending an edge-vertex-pair-order through one gadget as in Claim~\ref{clm:edge-vertex-pair-order-all} in the proof of Lemma~\ref{lem:CFI ready-individualization}:
			Gadget vertices are drawn in gray, edge vertices in light gray
			(cf.~Figure~\ref{fig:cfi-gadget-edge}).
			The two red-circled vertices are contained in an edge-vertex-pair-order.
			Because the shown gadget is of degree $3$,
			there is a unique blue gadget vertex adjacent to the two red vertices.
			The blue vertex has a unique additional adjacent vertex
			which is shown in red-dotted.
			By adding this canonical vertex to the edge-vertex-pair-order,
			it can be extended to the remaining edge-vertex-pair of the gadget.
		}
		\label{fig:cfi-extend-edge-vertex-pair-order}
	\end{figure}

	\begin{clm}
		\label{clm:edge-vertex-pair-order-all}
		There is an \IFPC{}-formula (uniformly in $G$)
		that defines an edge-vertex-pair-order on all base edges~$E$
		if $G - \orig{\indRel^\Struct}$ has no cycles.
	\end{clm}
	\begin{claimproof}
		We first show the following:
		Whenever for a base vertex
		there is an edge-vertex-pair-order~$R$ of all 
		incident edges apart from one,
		then~$R$ can be extended to the remaining incident edge. This is done as follows (cf.~Figure~\ref{fig:cfi-extend-edge-vertex-pair-order}).
		Let $\bVertA \in V$ be a base vertex
		and let $\neighbors{G}{\bVertA} = \set{\bVertB_1, \dots, \bVertB_d}$.
		Assume without loss of generality that $R$ orders the edge-vertex-pairs of the base edges $\set{\bVertA,\bVertB_i} \in E$ for every $i \in [d-1]$
		and let~$\vertB_i \in R$ be the edge vertex with origin $(\bVertA,\bVertB_i)$ for every $i \in [d-1]$.
		By construction of the CFI graphs,
		there is exactly one gadget vertex~$\vertA$ of the gadget of~$\bVertA$
		adjacent to~$\vertB_i$ for all $i \in [d-1]$.
		Because every gadget vertex is adjacent to exactly one edge vertex per incident edge-vertex-pair,
		we can add this unique edge vertex~$\vertB_d$ with origin $(\bVertA, \bVertB_d)$  and the unique edge vertex~$\vertB_d'$
		with origin $(\bVertB_d, \bVertA)$ adjacent to~$\vertA_d$  to~$R$.
		These two vertices can clearly be defined in \IFPC{} (without an order on the~$\bVertB_i$).
		
		We propagate this approach through gadgets.
		Assume there is an edge-vertex-pair-order~$R$ of~$E'$
		such that $G - E'$ has no cycles.
		Then $G-E'$ is a forest and there are vertices of degree one in $G-E'$
		(unless $E' = E$).
		For all degree-one vertices of $G - E'$,
		we extend~$R$ to the remaining incident edge as shown before.
		So unless $E' = E$, we added more edges to~$E'$.
		Surely, $G -E'$ has still no cycles.
		So
		we can repeat this process using a fixed-point operator
		to define an edge-vertex-pair-order of~$E$.
		
		It is clear that we can turn~$\indRel^\Struct$ into an
		edge-vertex-pair-order~$R$ of $\orig{\indRel^\Struct}$ (seen as undirected edges):
		For every  $(\bVertA,\bVertB) \in \orig{\indRel^\Struct}$,
		at least for one of the edge-vertex-pairs with origin $(\bVertA,\bVertB)$
		is individualized by~$\indRel^\Struct$.
		We put the $\indRel^\Struct$-minimal such vertex $\vertA$ into~$R$.
		For the edge-vertex-pair with origin $(\bVertB,\bVertA)$, we add the unique edge vertex adjacent to $\vertA$ to~$R$
		(if both $(\bVertA,\bVertB),(\bVertB,\bVertA) \in \orig{\indRel^\Struct}$,
		we start with the directed edge containing the $\indRel^\Struct$\nobreakdash-minimal vertex).
		One easily sees that $R$ is \IFPC{}-definable.
		By assumption, $G-\orig{\indRel^\Struct}$ has no cycles
		and thus we can define an edge-vertex-pair-order of~$E$.
	\end{claimproof}

	\begin{clm}
		\label{clm:orbit-no-cycle}
		Suppose~$R$ is an isomorphism-invariant edge-vertex-pair-order of~$E$,
		that is, $\autGroup{(\StructA,\indRel^\Struct)} = \autGroup{(\StructA,\indRel^\Struct, R)}$,
		and~$O$ is a $2$-orbit of $(G, \orig{\indRel^\Struct})$.
		Then the set of edge vertices
		$\setcond{\vertA \in R}{\text{the origin of } \vertA \text{ is in } O}$
		is a $1$-orbit of $(\StructA,\indRel^\Struct)$.
	\end{clm}
	\begin{claimproof}
		Let $(\bVertA,\bVertB),(\bVertA',\bVertB') \in O$,
		i.e., there is an automorphism $\autoA \in \autGroup{(G, \orig{\indRel^\Struct})}$
		such that $\autoA((\bVertA,\bVertB)) = (\bVertA',\bVertB')$.
		Every automorphism of $(G, \orig{\indRel^\Struct})$ induces an automorphism of $(\StructA, \indRel^\Struct)$.
		So there is an automorphism $\autoB \in \autGroup{(\StructA, \indRel^\Struct)}$ mapping the edge-vertex-pair with origin $(\bVertA,\bVertB)$
		to the one with origin $(\bVertA',\bVertB')$.
		Because~$R$ is isomorphism-invariant,
		$\autoB$ has to map edge vertices in~$R$ to edge vertices in~$R$.
		Hence, $\setcond{\vertA \in R}{\text{the origin of } \vertA \text{ is in } O}$
		is a subset of a $1$-orbit of $(\StructA,\indRel^\Struct)$.
		This set cannot be a proper subset of an orbit because~$R$ is isomorphism-invariant.
	\end{claimproof}

	Let $\formA_{2\text{-orb}}(\tup{x},\tup{y})$ be an $\IFPCWSCI$-formula distinguishing $2$-orbits of $\GraphClass$
	(that is, by definition, of $\indClosure{\GraphClass})$.
	We cannot evaluate $\formA_{2\text{-orb}}$ on $\StructA$
	to define $2$-orbits of $(G, \orig{\indRel^\Struct})$
	because~$\StructA$ has a more complicated automorphism structure
	than~$G$ and it is not clear how to witness orbits.
	Here we use the interpretation operator.
	We define an \IFPC{}-interpretation defining the base graph $(G, \orig{\indRel ^\Struct})$.
	Intuitively, we contract all gadgets to a single vertex,
	remove all edge vertices,
	and instead directly connect the contracted gadgets.
	
	\begin{figure}
		\centering
			\begin{tikzpicture}
				
				\draw[draw = none,use as bounding box] (-2.3,-1.8) rectangle (8.3, 1.8);

				\begin{scope}[name prefix=a]
					
					\node[edgevertex, fill=gray!30!white] (v10) at (+\edgePairDistAngle:\edgePairRadius) {$0$};
					\node[edgevertex, fill=gray!30!white] (v11) at (-\edgePairDistAngle:\edgePairRadius) {$1$};
					\node[edgevertex, fill=gray!30!white] (v20) at (220-\edgePairDistAngle:\edgePairRadius) {$0$};
					\node[edgevertex, fill=gray!30!white] (v21) at (220+\edgePairDistAngle:\edgePairRadius) {$1$};
					\node[edgevertex, fill=gray!30!white] (v30) at (140+\edgePairDistAngle:\edgePairRadius) {$0$};
					\node[edgevertex, fill=gray!30!white] (v31) at (140-\edgePairDistAngle:\edgePairRadius) {$1$};
					
					\node[gadgetvertex, fill=gray] (u000) at (0.4, 1.5*\gadgetVertexYDist) {$000$};
					\node[gadgetvertex, fill=gray] (u011) at (0.4, 0.5*\gadgetVertexYDist) {$110$};
					\node[gadgetvertex, fill=gray] (u101) at (0.4, -0.5*\gadgetVertexYDist) {$101$};
					\node[gadgetvertex, fill=gray] (u110) at (0.4, -1.5*\gadgetVertexYDist) {$011$};
					
					\path [-, draw=black, thick]
					(u000) to (v10)
					(u000) to (v20)
					(u000) to (v30)
					(u011) to (v10)
					(u011) to (v21)
					(u011) to (v31)
					(u101) to (v11)
					(u101) to (v20)
					(u101) to (v31)
					(u110) to (v11)
					(u110) to (v21)
					(u110) to (v30);
					
					\path[-, draw = red, ultra thick]
					(u000) to (v30) to (u110);
					\path[-, draw = blue, ultra thick]
					(u110) to (v21) to (u011) to (v31) to (u101);
					
				\end{scope}

				\begin{scope}[name prefix=b, shift={(6,0)}, xscale = -1]
					
					\node[edgevertex, fill=gray!30!white] (v10) at (+\edgePairDistAngle:\edgePairRadius) {$0$};
					\node[edgevertex, fill=gray!30!white] (v11) at (-\edgePairDistAngle:\edgePairRadius) {$1$};
					\node[edgevertex, fill=gray!30!white] (v20) at (220-\edgePairDistAngle:\edgePairRadius) {$0$};
					\node[edgevertex, fill=gray!30!white] (v21) at (220+\edgePairDistAngle:\edgePairRadius) {$1$};
					\node[edgevertex, fill=gray!30!white] (v30) at (140+\edgePairDistAngle:\edgePairRadius) {$0$};
					\node[edgevertex, fill=gray!30!white] (v31) at (140-\edgePairDistAngle:\edgePairRadius) {$1$};
					
					\node[gadgetvertex, fill=gray] (u000) at (0.4, 1.5*\gadgetVertexYDist) {$000$};
					\node[gadgetvertex, fill=gray] (u011) at (0.4, 0.5*\gadgetVertexYDist) {$011$};
					\node[gadgetvertex, fill=gray] (u101) at (0.4, -0.5*\gadgetVertexYDist) {$101$};
					\node[gadgetvertex, fill=gray] (u110) at (0.4, -1.5*\gadgetVertexYDist) {$110$};
					
					\path [-, draw=black, thick]
					(u000) to (v10)
					(u000) to (v20)
					(u000) to (v30)
					(u011) to (v10)
					(u011) to (v21)
					(u011) to (v31)
					(u101) to (v11)
					(u101) to (v20)
					(u101) to (v31)
					(u110) to (v11)
					(u110) to (v21)
					(u110) to (v30);
					
				\end{scope}
				
				\path[thick, -, draw = black]
				(av10) to (bv10)
				(av11) to (bv11);

				\path[-, draw = green!80!black, ultra thick]
				(au011) to (av10) to (bv10) to (bu000);
				\path[-, draw = brown!50!orange, ultra thick]
				(au101) to (av11) to (bv11) to (bu101) to (bv31) to (bu011);

			\end{tikzpicture}
		\caption[Distances in CFI Graphs]{
			Distances in CFI graphs as proven in Claim~\ref{clm:define-adjacent-gadgets}:
			Gadget vertices are drawn in gray and edge vertices in light gray
			(cf.~Figure~\ref{fig:cfi-gadget-edge}).
			Two gadget vertices in the same gadgets have distance~$2$ or~$4$
			as shown by the red and blue path.
			Two gadget vertices in gadgets for adjacent base vertices have distance $3$ or $5$ as shown by the green and brown path.
		}
		\label{fig:cfi-distances}
	\end{figure}
	
	When defining the base graph, base vertices of degree two will complicate matters. We first deal with base vertices of degree at least three.
	
	\begin{clm}
		\label{clm:define-gadget-vertices}
		For all vertices $\vertA\in \StructV$ the following holds:
		$\vertA$ is a gadget vertex whose origin has degree at least $3$ if, and only if,
			for all $\vertB \in \neighbors{\StructA}{\vertA}$,
			there are two distinct $\vertC,\vertC' \in \neighbors{\StructA}{\vertB} \setminus \set{\vertA}$ of different color.
	\end{clm}
	\begin{claimproof}
	 	Let $\vertA\in \StructV$.
		Assume that~$\vertA$ is a gadget vertex with origin~$\bVertA$
		and let~$\vertB \in \neighbors{\StructA}{\vertA}$.
		Then~$\vertB$ is an edge vertex with origin $(\bVertA,\bVertB)$.
		Because $\bVertA$ has degree at least~$3$,
		there is another gadget vertex $\vertA'\neq \vertA$ with origin~$\bVertA$
		adjacent to~$\vertB$ (cf.~Figure~\ref{fig:cfi-distances}).
		The edge vertex with origin $(\bVertB,\bVertA)$ adjacent to~$\vertB$
		has a different color than~$\vertA'$.
		To show the other direction,
		assume that
		for every neighbor~$\vertB\in \neighbors{\StructA}{\vertA}$
		there are two distinct~$\vertC,\vertC' \in \neighbors{\StructA}{\vertB} \setminus \set{\vertA}$ of different color.
		For a sake of contradiction,
		suppose that~$\vertA$ is an edge vertex
		with origin $(\bVertA,\bVertB)$.
		Consider the unique edge~$\vertB$ vertex with origin $(\bVertB,\bVertA)$
		adjacent to~$\vertA$.
		Every neighbor of~$\vertB$ is either~$\vertA$ or a gadget vertex with origin~$\bVertB$,
		but which all have the same color.
		Hence,~$\vertA$ is a gadget vertex.
	\end{claimproof}
	
	\begin{clm}
		\label{clm:define-adjacent-gadgets}
		Every two distinct gadget vertices $\vertA, \vertB \in \StructV$ 
		with origins~$\bVertA$ and~$\bVertB$, where at least one of them is of degree at least $3$,
		\begin{enumerate}[label=(\alph*)]
			\item have distance~$2$ or~$4$ if and only if $\bVertA = \bVertB$ and 
			\item have distance~$3$ or~$5$ if and only if $\set{\bVertA,\bVertB} \in E$. 
		\end{enumerate}
	\end{clm}
	\begin{claimproof}
		Let $\vertA=\tup{a}$ and $\vertB=\tup{b}$ be gadget vertices
		with origins~$\bVertA$ respectively~$\bVertB$, i.e., $\tup{a} \in \FF_2^{|\neighbors{G}{\bVertA}|}$ and $\tup{b} \in \FF_2^{|\neighbors{G}{\bVertB}|}$ such that $0 = a_1+\dots +a_{|\neighbors{G}{\bVertA}|} = b_1+\dots + b_{|\neighbors{G}{\bVertB}|}$.
		For Part~a),
		assume $\bVertA = \bVertB$.
		By construction of the CFI gadget,~$\vertA$ and~$\vertB$
		are not adjacent and have a common neighbor if and only if $a_i = b_i$ or some~$i \in [\tlength{a_i}]$.
		If $a_i \neq b_i$ for all $i\in [\tlength{a_i}]$,
		then there is another gadget vertex
		to which both~$\vertA$ and~$\vertB$ have distance~$2$
		because $\bVertA$ has degree at least~$3$ (cf.~Figure~\ref{fig:cfi-distances} again).
		For the other direction, assume~$\vertA$ and~$\vertB$ have distance~$2$ or~$4$.
		For every neighbor~$\vertC$ of~$\vertA$,
		there is an $\set{\bVertA,\bVertC} \in E$
		such that~$\vertC$ is an edge vertex with origin $(\bVertA,\bVertC)$
		and every neighbor of~$\vertC$ is a gadget vertex with origin~$\bVertA$
		or an edge vertex with origin $(\bVertC,\bVertA)$.
		Hence, if~$\vertA$ and~$\vertB$ have distance~$2$,
		then $\bVertB = \bVertA$.
		The case of distance~$4$ is similar,
		here all distance~$4$ vertices are either vertices of the same gadget or edge vertices.
		
		Proving Part~b) is similar:
		Assume $\set{\bVertA,\bVertB} \in E$.
		If $a_i = b_j$ where the $i$-th edge-vertex-pair of the gadget of $\bVertA$ is connected to the $j$-th one of the gadget of $\bVertB$,
		then~$\vertA$ and~$\vertB$ have distance~$3$
		(via two edge vertices with origin $(\bVertA,\bVertB)$ and $(\bVertB,\bVertA)$),
		otherwise they have distance~$5$.
		For the other direction, all distance~$3$ vertices of~$\vertA$
		are either edge vertices or gadget vertices of gadgets whose origin~$\bVertC$ satisfies $\set{\bVertA,\bVertC} \in E$.
	\end{claimproof}

	\noindent
	
	We first consider the case that the minimum degree of $G$ is three and later describe how to remove this assumption.
	By Claims~\ref{clm:define-gadget-vertices} and~\ref{clm:define-adjacent-gadgets}, 
	the following interpretation
	$\interpret = (\formA_{\text{dom}}, \formA_{\cong}, \formA_{E}, \formA_{\spleq}, \formA_\indRel)$ defines the base graph:
	\begin{align*}
		\formA_{\text{dom}}(x) &:= \forall y.\qspace E(x,y) \Rightarrow \exists z_1,z_2.\qspace z_1\neq x\land z_2 \neq x \land E(y,z_1)\land E(y,z_2)\land   z_1 \spless z_2,\\
		\formA_{\cong}(x,y) &:= x=y \lor \formA_{\text{dist}}^2 (x,y) \lor \formA_{\text{dist}}^4 (x,y),\\
		\formA_{E} (x,y)&:=  \formA_{\text{dist}}^3 (x,y) \lor \formA_{\text{dist}}^5 (x,y), \\
		\formA_{\spleq}(x,y) &:= x \spleq y.
	\end{align*}
	We used formulas $\formA_{\text{dist}}(x,y)^\ell$ defining that~$x$ and~$y$ have distance~$\ell$.
	It remains to define~$\formA_\indRel$,
	for which we omit a formal definition:
	We use~$\indRel$ to define an order on the base edges in $\orig{\indRel^\Struct}$ and individualize the corresponding base vertices using this order. 
	
	Now we are able to evaluate $\formA_{2\text{-orb}}$ on the base graph. 
	We extend $\interpret$ by two parameters~$z_1$ and~$z_2$
	for two edge vertices, for which we want to know whether their origins are in the same $2$-orbit of~$G$.
	We add two fresh binary relation symbols~$S_1$ and~$S_2$
	and let~$\interpret$ define~$S_i$ to contain the origin of the edge vertex of~$z_i$ for every $i \in [2]$.

	We lift the total preorder of the base vertices
	to the edge vertices using the interpretation operator:
	\[\formA_{\text{base-orb}}(z_1,z_2) := \iop{\interpret(z_1,z_2)}{
	\forall \tup{x}\tup{y}.\qspace S_1(\tup{x}) \Rightarrow S_2(\tup{y}) \Rightarrow \formA_{2\text{-orb}}(\tup{x},\tup{y})}.\]
	The formula~$\formA_{\text{base-orb}}$ defines whether the
	origins of two edge vertices are in the same orbit of $(G,\orig{\indRel^\Struct})$.
	Note that $\formA_{2\text{-orb}}$ does not mention the additional relations~$S_1$ and~$S_2$
	and thus every WSC-fixed-point operator in $\formA_{2\text{-orb}}$
	is evaluated on a structure not containing~$S_1$ and~$S_2$
	(by the reduct semantics) and so indeed on a $\indClosure{\GraphClass}$-graph.
	
	Finally, we can indeed check whether there is an orbit as required
	by Claim~\ref{clm:orbit-cycle}:
	Consider the equivalence classes induced by $\formA_{\text{base-orb}}$
	on the edge vertices and check the existence of the required cycle.
	We define the minimal such orbit if such one exists.
	
	If no such orbit exists,
	then $G-\orig{\indRel^\Struct}$ contains no cycles
	and we can define an edge-vertex-pair-order~$R$ of~$E$ by Claim~\ref{clm:edge-vertex-pair-order-all} in \IFPC{}.
	Because~$R$ is \IFPC{}-definable,~R is in particular isomorphism-invariant.
	If there is a non-trivial $2$-orbit of $G-\orig{\indRel^\Struct}$,
	we pick the minimal one to define a non-trivial $1$-orbit of $(\Struct$, $\indRel^\Struct)$ using~$R$ and Claim~\ref{clm:orbit-no-cycle}.
	
	If that is also not the case, then all $2$-orbit of $G-\orig{\indRel^\Struct}$ are trivial,
	that is, there is a definable total order on the (directed) base edges.
	Together with the edge-vertex-pair-order~$R$,
	we define a total order on all edge-vertex-pairs,
	which can be extended to a total order on~$\StructA$.
	We define the minimal vertex, which is not individualized
	in $\indRel^\Struct$, if it exists.
	Otherwise, all vertices are individualized.
	We pick the $\indRel$-minimal one.

	\begin{figure}
		\centering
			\begin{tikzpicture}

			\begin{scope}[name prefix=a]
				
				\node[edgevertex, bicolor fill={edgeColA}{edgeColB}] (v10) at (+\edgePairDistAngle:\edgePairRadius) {$0$};
				\node[edgevertex, bicolor fill={edgeColA}{edgeColB}] (v11) at (-\edgePairDistAngle:\edgePairRadius) {$1$};
				
				\node[edgevertex, bicolor fill={edgeColA}{edgeColC, shading angle=-90}] (v20) at (180-\edgePairDistAngle:\edgePairRadius) {$0$};
				\node[edgevertex, bicolor fill={edgeColA}{edgeColC, shading angle=-90}] (v21) at (180+\edgePairDistAngle:\edgePairRadius) {$1$};
				
				\node[gadgetvertex, fill=edgeColA] (u00) at (0, 0.5*\gadgetVertexYDist) {$00$};
				\node[gadgetvertex, fill=edgeColA] (u11) at (0, -0.5*\gadgetVertexYDist) {$11$};
%
				\path [-, draw=black, thick]
				(u00) to (v10)
				(u00) to (v20)
				(u11) to (v11)
				(u11) to (v21);
				
				\draw [thick, dashed, rounded corners]
					(-\edgePairRadius - 0.5, -0.5*\gadgetVertexYDist -0.5 ) rectangle (\edgePairRadius + 0.5, 0.5*\gadgetVertexYDist + 0.5);				
				
			\end{scope}
			
			\begin{scope}[name prefix=c, shift={(-6,0)}]
				
				\node[edgevertex, bicolor fill={edgeColC}{edgeColA}, shading angle = 90] (v10) at (+\edgePairDistAngle:\edgePairRadius) {$0$};
				\node[edgevertex, bicolor fill={edgeColC}{edgeColA}, shading angle=90] (v11) at (-\edgePairDistAngle:\edgePairRadius) {$1$};
				
				\begin{scope}[opacity = 0] 
					\node[gadgetvertex, fill=edgeColB] (u000) at (0.4, 1.5*\gadgetVertexYDist) {$000$};
					\node[gadgetvertex, fill=edgeColB] (u011) at (0.4, 0.5*\gadgetVertexYDist) {$011$};
					\node[gadgetvertex, fill=edgeColB] (u101) at (0.4, -0.5*\gadgetVertexYDist) {$101$};
					\node[gadgetvertex, fill=edgeColB] (u110) at (0.4, -1.5*\gadgetVertexYDist) {$110$};
				\end{scope}
				
				\path [-, draw=black, dashed, thick]
				(u000) to (v10)
				(u011) to (v10)
				(u101) to (v11)
				(u110) to (v11);
				
			\end{scope}
			
			\begin{scope}[name prefix=b, shift={(6,0)}, xscale = -1]
				
				\node[edgevertex, bicolor fill={edgeColB}{edgeColA}, shading angle = -90] (v10) at (+\edgePairDistAngle:\edgePairRadius) {$0$};
				\node[edgevertex, bicolor fill={edgeColB}{edgeColA}, shading angle=-90] (v11) at (-\edgePairDistAngle:\edgePairRadius) {$1$};
				
				\begin{scope}[opacity = 0] 
				\node[gadgetvertex, fill=edgeColB] (u000) at (0.4, 1.5*\gadgetVertexYDist) {$000$};
				\node[gadgetvertex, fill=edgeColB] (u011) at (0.4, 0.5*\gadgetVertexYDist) {$011$};
				\node[gadgetvertex, fill=edgeColB] (u101) at (0.4, -0.5*\gadgetVertexYDist) {$101$};
				\node[gadgetvertex, fill=edgeColB] (u110) at (0.4, -1.5*\gadgetVertexYDist) {$110$};
				\end{scope}
				
				\path [-, draw=black, dashed, thick]
				(u000) to (v10)
				(u011) to (v10)
				(u101) to (v11)
				(u110) to (v11);
				
			\end{scope}
			
			\path[thick, -, draw = black]
			(av10) to (bv10)
			(av11) to (bv11);
			
			\path[thick, -, draw = black]
			(av20) to (cv10)
			(av21) to (cv11);

		\end{tikzpicture}
		\caption{A degree two CFI gadget: the gadget is encircled.
		To its left and right, the edge-vertex-pair of the neighbored gadget is drawn. The degree two CFI gadget essentially subdivides a connection between these two other edge-vertex-pairs.
		By using multiple degree two gadgets, these paths can get arbitrarily long.}
		\label{fig:degree-two-gadget}
	\end{figure}
	We finally have to describe how to deal with base vertices of degree at most two.
	First, check whether the base graph~$G$ is of order at most two.
	If it is $2$-regular, then~$G$ is a cycle because~$G$ is connected.
	In this case, the CFI graph~$\StructA$ is either a cycle (if it is odd) or the disjoint union of two cycles (it it is even).
	If~$G$ is not $2$-regular, then~$G$ is a path because~$G$ is connected.
	In this case, the CFI graph~$\StructA$ is the disjoint union of two paths (of equal length if and only if~$\StructA$ is odd).
	Since both, disjoint unions of cycles or of paths, are \IFPC-definable, we can distinguish this case.
	It is easy to show that this class of colored graphs is ready for individualization -- here we actually do not need to define the base graph at all. But this can also be done in \IFPC.
	
	If the maximal degree of~$G$ is at least three, we follow the same approach as before:
	We first define the set of gadget vertices.
	For base vertices of degree at least three we use Claim~\ref{clm:define-gadget-vertices}.
	For base vertices of degree one,
	note that gadget vertices of a degree one gadget are the only vertices in a CFI graph of degree one.
	It remains to consider gadget vertices of degree two gadgets. 
	These base vertices subdivide the connection between two edge-vertex-pairs (cf.~Figure~\ref{fig:degree-two-gadget}).
	On these paths, every third vertex is a gadget vertex, which can be defined using a fixed-point operator.
	
	In the next step, we define which gadget vertices have the same origin.
	For base vertices of degree at least three,
	we exploit Claim~\ref{clm:define-adjacent-gadgets}, for degree one base vertices this is trivial (they only have one gadget vertex),
	and for degree two base vertices we can follow the ``subdivided'' paths again using fixed-points.
	Finally, we need to define which gadgets have adjacent origins.
	For base vertices of degree at least three we again use Claim~\ref{clm:define-adjacent-gadgets},
	for degree one this is trivial,
	and for degree two we exploit the structure of the paths.
	
	This high-level description omits some details, but a full formal proof would be lengthy and technical and not provide further insights. The main difference is that, due to the long paths, fixed-points are needed.
	Since degree two base vertices do no change the automorphism or isomorphism structure of CFI graphs,
	it usually suffices to consider base graphs of minimum degree three.
	It in particular would suffice for the main results of this article.
	This concludes the proof of Lemma~\ref{lem:CFI ready-individualization}.	
\end{proof}

We now can prove Theorem~\ref{thm:canonize-CFI if-base},
stating that if \IFPCWSCI{} canonizes $\GraphClass$,
then \IFPCWSCI{} canonizes $\CFI{\GraphClass}$, too.
\begin{proof}[Proof of Theorem~\ref{thm:canonize-CFI if-base}]
	Assume that \IFPCWSCI{} defines a canonization for $\GraphClass$.
	Then \IFPCWSCI{} distinguishes $2$-orbits of $\GraphClass$ by Lemma~\ref{lem:canon-orbit-ready-iff}.
	Hence, $\CFI{\GraphClass}$ is ready for individualization in \IFPCWSCI{} by Lemma~\ref{lem:CFI ready-individualization} 
	and so \IFPCWSCI{} defines a canonization for $\CFI{\GraphClass}$ by Lemma~\ref{lem:canon-orbit-ready-iff}.
\end{proof}

\begin{cor}
	If $\GraphClass$ is a class of base graphs of bounded degree,
	then \IFPCWSCI{} defines canonization for $\GraphClass$ if and only if 
	it defines canonization for $\CFI{\GraphClass}$.
\end{cor}
\begin{proof}
	One direction is by Theorem~\ref{thm:canonize-CFI if-base}.
	For the other direction,
	let $\GraphClass$ be graph class of maximal degree~$d$.
	Then there is a $(d+2)$-dimensional \IFPC{}-interpretation $\interpret$
	such that $\interpret(G)$ is the even CFI graph over $G$
	for every base graph $G \in \indClosure{\GraphClass}$.
	Individualized vertices are translated into the coloring
	and thus into gadgets of unique color.
	Together with the canonization-defining interpretation
	and the base-graph-defining one in the proof of Theorem~\ref{thm:canonize-CFI if-base},
	we obtain a canonization for $\indClosure{\GraphClass}$.
\end{proof}
Note that this approach of defining a CFI graph to canonize the base graph
cannot work for arbitrary base graphs
because for large degrees the CFI graphs get exponentially large
(which cannot be defined by an \IFPC{}-interpretation).

We make the following observation.
First, Theorem~\ref{thm:canonize-CFI if-base} can be applied iteratively:
If \IFPCWSCI{} canonizes~$\GraphClass$,
then \IFPCWSCI{} canonizes $\CFI{\GraphClass}$,
and so \IFPCWSCI{} canonizes $\CFI{\CFI{\GraphClass}}$.
Second, every application of Theorem~\ref{thm:canonize-CFI if-base}
adds one WSC-fixed-point operator (to define Gurevich's algorithm in Lemma~\ref{lem:canon-orbit-ready-iff})
and one interpretation operator (for the base-graph-defining interpretation in Lemma~\ref{lem:CFI ready-individualization}).
More precisely, the nesting depth of these operators increases.
We will show that this is in some sense necessary:
The CFI query on a variation of $\CFI{\CFI{\GraphClass}}$
cannot be defined without nesting two WSC-fixed-point and two interpretation operators.

\section{The CFI Query and Nesting of Operators}
\label{sec:nesting-of-operators}

In this section we show that the increased nesting depth of WSC-fixed-point
and interpretation operators used to canonize CFI graphs
in Theorem~\ref{thm:canonize-CFI if-base} is unavoidable.
Because \IFPC{} does not define the CFI query,
it is clear that even if the class of base graph has \IFPC{}-distinguishable orbits,
the CFI query cannot be defined in \IFPC{}
(e.g., on a class of all ordered graphs).
So, the nesting depth of WSC-fixed-point operators has to increase further.
However, for orbits distinguishable in \IFPCWSCI{} but not in \IFPC{}
it is not clear whether the nesting depth has to increase.
We now show that this is indeed the case.

Intuitively, we want to combine non-isomorphic CFI graphs
into a new base graph and apply the CFI construction again.
To define orbits of these double CFI graphs,
one has to define the CFI query for the base CFI graphs,
which cannot be done without a WSC-fixed-point operator.
However, parameters of WSC-fixed-point operators complicate matters.
If a parameter of a WSC-fixed-point operator is used to 
fix a vertex contained in of the base CFI graphs,
then this base CFI graph can be distinguished from all the others.
Hence, orbits of this base CFI graph can be defined without defining the CFI query.
To overcome this problem, we use multiple copies of the base CFI graphs
such that their number exceeds the number of parameters.

If we want to make choices from the base CFI graphs not containing parameters,
we have to define the CFI query for them, which increases the nesting depth.
If we do not do so, we can essentially make choices only in the base CFI graphs containing a parameter vertex.
That is, we can move the twist to the other base CFI graphs.
Proving that in this case the CFI query cannot be defined without more operators requires formal effort:
We introduce a logic which allows for quantifying over all individualizations of the base CFI graphs containing parameters.
This is (potentially) more powerful than making choices, but we still can prove non-definability of the CFI query.
This logic has the benefit that we can characterize it via a pebble game.

First, we introduce a formalism for nesting WSC-fixed-point and interpretation operators in Section~\ref{sec:nested-operators}.
Second, we provide a graph construction combining multiple base CFI graphs in Section~\ref{sec:color-class-joins-and-cfi}.
Third, we define a logic quantifying over certain individualizations in Section~\ref{sec:quantifying-pebbled-part-individualizations}
and prove that it cannot distinguish CFI graphs from the prior graph construction.
Last, we make the argument sketched above formal in Section~\ref{sec:nesting-operators-is-necessary} and show that
these CFI graphs require nested WSC-fixed-point operators.

\newcommand{\wscrank}[1]{\operatorname*{wscr}(#1)}
\subsection{Nested WSC-Fixed-Point and Interpretation Operators}
\label{sec:nested-operators}
We next consider \IFPCWSCI{}\nobreakdash-formulas with restricted nesting depth of
WSC-fixed-point and interpretation operators.
Let $\IFPC{} \subseteq L \subseteq \IFPCWSCI{}$ be a subset of \IFPCWSCI{}.
We write $\WSCof{L}$ for the set of formulas
obtained from $L$ using \IFPC{}-formula-formation rules
and WSC-fixed-point operators,
for which the step, choice, witnessing, and output formulas are $L$-formulas.
Likewise, we define $\Iof{L}$:
One can use an additional interpretation operator $I(\interpret, \formB)$,
where  the interpretation $\interpret$ is an $\IFPC{}$-interpretation and $\formB$ is an $L$-formula.
Note that $L \subseteq \WSCof{L}$, $L \subseteq \Iof{L}$,
and $\Iof{\IFPC} = \IFPC{}$ because \IFPC{} is closed under \IFPC{}-interpretations.

We abbreviate $\WSCIof{L} :=  \WSCof{\Iof{L}}$ and $\WSCIkof{k+1}{L} := \WSCIof{\WSCIkof{k}{L}}$.
Note that we restrict here to \IFPC{}-interpretations in the interpretation operator, which means that $\WSCIkof{n}{\IFPC}$ is the fragment of $\IFPCWSCI$
in which WSC-fixed-point operators can be nested at most~$n$ times,
interpretation operators are restricted to \IFPC{}-interpretations,
and WSC-fixed-point and interpretation operators alternate.
But also note that both, the WSC-fixed-point operator and the interpretation operator, can be made ``trivial'' (using contradictory choice and step formulas and a tautological output formula or the identity interpretation, respectively).
We restrict to \IFPC{}-interpretations because otherwise an interpretation operator $\iop{\interpret}{\formA}$ introduces ``hidden'' nestings of WSC-fixed-point operators: if both~$\interpret$ and~$\formA$ nest~$k$ WSC-fixed-point operators, then $\iop{\interpret}{\formA}$ evaluates~$\formA$ in the image of~$\interpret$ and so nests~$2k$ WSC-fixed-point operators.
This is not the desired notion for this section.

Restricting ourself to \IFPC{}-interpretations is a limitation,
but for the goal of this section it is not restrictive:
The goal is to prove $\WSCof{\IFPC{}} = \WSCIof{\IFPC{}} < \WSCIkof{2}{\IFPC{}}$ and the result only gets stronger when $\WSCIkof{2}{\IFPC{}}$ is limited to \IFPC{}-interpretations.
Recall the construction in Lemmas~\ref{lem:canon-orbit-ready-iff} and~\ref{lem:CFI ready-individualization}:
\begin{cor}
	\label{cor:canonize-CFI if-base-wsci}
	Let $\GraphClass$ be a class of base graphs.
	\begin{enumerate}
		\item If $L$ distinguishes $2$-orbits of $\GraphClass$,
		then $\CFI{\GraphClass}$ is ready for individualization in 
		$\Iof{L}$.
		\item If $\CFI{\GraphClass}$ is ready for individualization in $L$, then $\WSCof{L}$ defines a canonization for $\CFI{\GraphClass}$.
		\item If $L$ distinguishes $2$-orbits of $\GraphClass$,
		then $\WSCIof{L}$ defines a canonization for $\CFI{\GraphClass}$.
	\end{enumerate}
\end{cor}
\begin{proof}
	The first claim follows from the proof of Lemma~\ref{lem:CFI ready-individualization}.
	The second claim follows from the proof of Lemma~\ref{lem:canon-orbit-ready-iff}
	and the fact that implementing Gurevich's canonization algorithm
	requires one WSC-fixed-point operator~\cite{LichterSchweitzer24}.
	Combining the first and second claim yields the last one.
\end{proof}

\newcommand{\ccjoin}[1]{J_{\text{cc}}(#1)}
\newcommand{\ccjoink}[2]{J_{\text{cc}}^{#1}(#2)}
\subsection{Color Class Joins and CFI Graphs}
\label{sec:color-class-joins-and-cfi}
We define a composition operation of colored graphs:
Let $G_1, \dots , G_\ell$ be connected colored graphs,
such that all~$G_i$ have the same number of color classes~$c$.
The \defining{color class join} $\ccjoin{G_1, \dots, G_\ell}$ is defined as follows:
Start with the disjoint union of the~$G_i$ and add~$c$ additional vertices $\vertA_1, \dots, \vertA_c$.
Then add for each $i \in[c]$ edges between~$\vertA_i$ and each vertex~$\vertB$ in the $i$-th color class of every $G_j$ (cf.~Figure~\ref{fig:color-class-join}).
The resulting colored graph $\ccjoin{G_1, \dots, G_\ell}$ has~$2c$ color classes:
One color class for each~$\vertA_i$ and the union of the color classes of the~$G_j$.
We call the~$G_j$ the \defining{parts} of $\ccjoin{G_1, \dots, G_\ell}$.
The vertices~$\vertA_i$ are called the \defining{join vertices} and the other ones the \defining{part vertices}.
The \defining{part of} a part vertex~$\vertB$ is the graph~$G_j$ containing~$\vertB$.
The color class join has two crucial properties:
First, defining orbits of  $\ccjoin{G_1, \dots, G_\ell}$
is at least as hard as defining isomorphism of the~$G_j$.

\begin{figure}
	\centering
	\newcommand{\partgraph}[1]{
		\draw[gray, thick] (0,0) ellipse (1.8cm and 1.1cm);
		\draw[edgeColB, thick]  (-1,0) ellipse (0.3cm and 0.8cm);
		\draw[edgeColA, thick]  (0,0) ellipse (0.3cm and 0.8cm);
		\draw[edgeColC, thick] (1,0) ellipse (0.3cm and 0.8cm);
		
		\node[vertex, fill=edgeColB] (r1) at (-1,0.5){};
		\node[vertex, fill=edgeColB] (r2) at (-1,0){};
		\node[vertex, fill=edgeColB] (r3) at (-1,-0.5){};
		
		\node[vertex, fill=edgeColA] (b1) at (0,0.5){};
		\node[vertex, fill=edgeColA] (b2) at (0,0){};
		\node[vertex, fill=edgeColA] (b3) at (0,-0.5){};
		
		\node[vertex, fill=edgeColC] (g1) at (1,0.5){};
		\node[vertex, fill=edgeColC] (g2) at (1,0){};
		\node[vertex, fill=edgeColC] (g3) at (1,-0.5){};
		
		\node at (0,1.6) {#1};
	}
	\begin{tikzpicture}[font=\small, scale=1.1]
		
		\draw[use as bounding box,draw=none] (6.5,1.2) rectangle (-6.5,-2.9);
		\begin{scope}[name prefix = G1, scale=0.6,shift={(-5,-0.5)}, rotate = 20]
			\partgraph{$G_1$-part}
		\end{scope}
		\begin{scope}[name prefix = G2, scale=0.6]
				\partgraph{$G_2$-part}
		\end{scope}
		\begin{scope}[name prefix = G3, scale=0.6, shift={(5,-0.5)}, rotate = -20]
			\partgraph{$G_3$-part}
		\end{scope}
		
		\node[vertex, fill=edgeColE] (r) at (-2,-2.5) {};
		\node[vertex, fill=edgeColD] (b) at (0,-2.5) {};
		\node[vertex, fill=edgeColF] (g) at (2,-2.5) {};
		
		\draw [decorate,
		decoration = {brace}]
		(4.5,0.6) --  (4.5,-1.2)
		node[pos=0.5,black, right=0.2cm , align=left]
		{part vertices};
		
		\draw [decorate,
		decoration = {brace}]
		(4.5,-2.2) --  (4.5,-2.8)
		node[pos=0.5,black, right=0.2cm , align=left]
		{join vertices};
		
		\path[draw, black, thick]
		(G1r1) -- (r)
		(G1r2) -- (r)
		(G1r3) -- (r)
		(G2r1) -- (r)
		(G2r2) -- (r)
		(G2r3) -- (r)
		(G3r1) -- (r)
		(G3r2) -- (r)
		(G3r3) -- (r);
		
		\path[draw, black, thick]
		(G1g1) -- (g)
		(G1g2) -- (g)
		(G1g3) -- (g)
		(G2g1) -- (g)
		(G2g2) -- (g)
		(G2g3) -- (g)
		(G3g1) -- (g)
		(G3g2) -- (g)
		(G3g3) -- (g);
		
		\path[draw, black, thick]
		(G1b1) -- (b)
		(G1b2) -- (b)
		(G1b3) -- (b)
		(G2b1) to[bend left=9] (b)
		(G2b2) to[bend right=9] (b)
		(G2b3) -- (b)
		(G3b1) -- (b)
		(G3b2) -- (b)
		(G3b3) -- (b);
	\end{tikzpicture}
\caption{Visualization of color graphs joins:
For three graphs $G_1$, $G_2$, and $G_3$, each with three color classes,
the figure shows the color graph join $\ccjoin{G_1,G_2,G_3}$.
Edges inside the graphs $G_1$ to $G_3$ are not drawn.
For each color class,
one new vertex is added and connected to all existing vertices of that color class.
The new vertices receive a new unique color.
\label{fig:color-class-join}
}
 
\end{figure}

\begin{lem}
	\label{lem:color-class-join-orbits}
	If two part vertices $\vertB$ and $\vertB'$ are in the same orbit of $\ccjoin{G_1, \dots, G_\ell}$,
	then the part of $\vertB$ is isomorphic to the one of $\vertB'$.
\end{lem}
\begin{proof}
	Because every~$G_i$ is connected
	and the part and join vertices have different colors,
	an automorphism of $\ccjoin{G_1, \dots, G_\ell}$
	mapping~$\vertB$ to~$\vertB'$
	has to map the part of~$\vertB$ to the part of~$\vertB'$.
	In particular, these parts are isomorphic.
\end{proof}
Second, the automorphism structure of a part $G_j$ is independent of
individualizing vertices in other parts (or join vertices).
\begin{lem}
	Let $\tup{\vertC}$ be a tuple of vertices of $\ccjoin{G_1, \dots, G_\ell}$, let $j \in [\ell]$,
	and let $\vertB$ and~$\vertB'$ be vertices of~$G_j$.
	If there is no $i \in [|\tup{\vertC}|]$
	such that~$G_j$ is the part of~$\vertC_i$
	and if~$\vertB$ and~$\vertB'$ are in the same orbit of~$G_j$,
	then~$\vertB$ and~$\vertB'$ are in the same orbit of $(\ccjoin{G_1, \dots, G_\ell}, \tup{\vertC})$.
\end{lem}
\begin{proof}
	Because~$\vertB$ and~$\vertB'$ are in the same orbit of~$G_j$,
	there is a $\autoA \in \autGroup{G_j}$ such that $\autoA(\vertB) = \vertB'$.
	We extend~$\autoA$ to $\ccjoin{G_1, \dots, G_\ell}$
	by the identity on all join vertices and all parts apart from~$G_j$.
	Because~$\autoA$ respects the colors classes of~$G_j$
	and the join vertex~$\vertA_i$ is adjacent to all vertices in the $i$-th color class of~$G_j$ for every $i \in [c]$,
	adjacency of part vertices of~$G_j$ and join vertices is preserved by~$\autoA$.
	Thus, $\autoA \in \autGroup{(\ccjoin{G_1, \dots, G_\ell}, \tup{\vertC})}$
	because the part of~$\vertC_i$ is not~$G_j$ for every $i \in [|\tup{\vertC}|]$
	and~$\autoA$ is the identity on all other parts.
\end{proof}

We will be later interested in the case that the coloring of a graph is unnecessary with respect to automorphisms in the following sense.
The \defining{coloring of a graph $G$ is unnecessary}
if $G$ has the same automorphisms as the uncolored variant of $G$.
\begin{lem}
	\label{lem:unnecessary-colorings}
	If the coloring of a base graph $G$ is unnecessary,
	then the coloring of $\CFI{G,g}$ is unnecessary for every $g \in \FF_2$.
	If $G_1,\dots, G_\ell$ are base graphs of degree at most $d$ 
	whose coloring is unnecessary and $\ell+2 > d$,
	then the coloring of $\ccjoin{G_1,\dots, G_\ell}$ is unnecessary.
\end{lem}
\begin{proof}
	The first claim concerning CFI-graphs holds because automorphisms of a CFI
	graphs are composed of automorphisms of its base graph and ``CFI-automorphisms''~\cite{Pago23} (the automorphisms of CFI graphs for totally ordered base graphs). In particular, these CFI-automorphisms are independent of the coloring of the base graph.
	If the coloring of the base graph is unnecessary,
	then the uncolored base graph has the same automorphisms as the colored one.
	Hence both, the colored and uncolored CFI graph,
	have the same automorphisms of the base graph and the same CFI-automorphisms.
	This implies that the coloring of the CFI graph is unnecessary.
	
	For the second claim, we first show that also in the uncolored color class join, each join vertex is in its own orbit.
	Every part vertex $\vertA$ in $\ccjoin{G_1,\dots, G_\ell}$ has degree at most $d+1$: the vertex $\vertA$ has at most $d$ incident edges in its part $G_i$ and one additional incident edge to one join vertex.
	Each join vertex has degree $\ell$.
	Because $\ell+2 > d$, the join vertices in $\ccjoin{G_1,\dots, G_\ell}$ have degree larger than any part vertex.
	Two join vertices cannot be in the same orbit, either, because they are connected to vertices in different orbits in the uncolored version of~$G_i$,
	which is the case because the coloring of~$G_i$ is unnecessary.
	Hence, all automorphisms in both the colored and uncolored color graph join
	are composed of automorphisms of the parts and isomorphisms between parts.
	Because the coloring of the~$G_i$ is unnecessary, no additional automorphisms are added when removing the coloring.
\end{proof}

\paragraph{Color Class Joins of CFI Graphs}
Now we apply color class joins to CFI graphs:
For connected colored graphs $G$, $H$, and $K$ with the same number of color classes,
we define
\[\ccjoink{k}{G, H, K} := \ccjoin{\underbrace{G, \dots , G}_{k \text{ many}}, \underbrace{H, \dots , H}_{k \text{ many}}, \underbrace{K, \dots , K}_{k \text{ many}}}.\]
Let $\GraphClass$ be a class of colored base graphs.
For $G \in \GraphClass$ and $g \in \FF_2$, we define
\begin{align*}
	\CFImk{k}{G,g} &:= \ccjoink{k}{\CFI{G,0},\CFI{G,g},\CFI{G,1}},\\
	\CFImk{k}{\GraphClass} &:= \setcond*{\CFImk{k}{G,g}}{G \in \GraphClass, g \in \FF_2}, \text{ and}\\
	\CFIm{\GraphClass} &:= \bigcup_{k \in \nat} \CFImk{k}{\GraphClass}.
\end{align*}

\begin{lem}
	\label{lem:unnecessary-coloring-cfim}
	Let $G$ be a colored base graph of maximal degree $d$.
	If the coloring of $G$ is unnecessary,
	then for all $k > 2^{d-1}-2$ and all $g,h \in \FF_2$,
	the coloring of $\CFI{\CFImk{k}{G,g},h}$ is unnecessary, too.
\end{lem}
\begin{proof}
	The lemma follows by applying Lemma~\ref{lem:unnecessary-colorings} multiple times.
	Because the coloring of~$G$ is unnecessary, the one of $\CFI{G,g'}$ is unnecessary for all $g' \in \FF_2$.
	Because the maximal degree of $\CFI{G,g'}$ is $2^{d-1}$ and
	$k +2 > 2^{d-1}$, the coloring of $\CFImk{k}{G,g}$ is unnecessary as well.
	Hence, also the coloring of $\CFI{\CFImk{k}{G,g}, h}$ is unnecessary.
\end{proof}

\begin{lem}
	\label{lem:extended-CFI if-cfi}
	Let $\IFPC{} \subseteq L \subseteq \IFPCWSCI{}$
	be a subset of \IFPCWSCI{}
	closed under \IFPC{}-formula-formation rules.
	If $L$ canonizes  $\CFI{\GraphClass}$,
	then $L$ canonizes $\CFIm{\GraphClass}$.
\end{lem}
\begin{proof}
	First, we can easily distinguish join vertices from part vertices in \IFPC{}
	for graphs in $\indClosure{\CFIm{\GraphClass}}$
	because the join vertices are in singleton color classes and the part vertices are not.
	So for a given part vertex~$\vertA$ of $G$,
	we can define the set of all part vertices contained in the same part as~$\vertA$
	(namely the ones reachable without using a join vertex),
	that is, we can define the CFI graph in $\CFI{\GraphClass}$ containing~$\vertA$.
	
	Let $\interpret_{\text{can}}$ be an $L$-canonization of $\CFI{\GraphClass}$
	(that is, by definition, a canonization of $\indClosure{\CFI{\GraphClass}}$).
	Then we obtain an $L$-interpretation
	$\interpret_{\text{part-can}}(x)$
	that given a $\CFIm{\GraphClass}$-graph canonizes the part of $x$.
	It essentially is $\interpret_{\text{can}}$ but only considers the (definable)
	set of part vertices in the same part as $x$.
	So every choice-set in the evaluation of $\interpret_{\text{part-can}}(x)$
	will be a set of part vertices in the same part as $x$.
	Thus, a choice-set is an orbit if and only if
	the corresponding choice-set in the evaluation of $\interpret_{\text{can}}$ 
	is an orbit.
	The witnessing automorphisms are obtained by extending the
	witnessing automorphisms defined in $\interpret_{\text{can}}$ 
	with the identity on all vertices not in the part of $x$.
	In that way, exactly the same choices are successfully witnessed by
	$\interpret_{\text{part-can}}$ as by  $\interpret_{\text{can}}$.
	Note that $\interpret_{\text{part-can}}$ is an~$L$-interpretation
	because $L$ is closed
	under \IFPC{}-formula-formation rules. 
	
	We use $\interpret_{\text{part-can}}(x)$ to define the canon of every
	part containing an $\indRel$-individualized vertex.
	These canons can be ordered according to the order of the individualized vertices.
	For the remaining parts not containing individualized vertices,
	we use $\interpret_{\text{part-can}}(x)$ to determine how many of 
	them are even respectively odd CFI graphs.
	We obtain a canon for the even and odd graph (if they occur)
	and define as many copies as needed (using numeric variables).
	The copies can be ordered.
	We finally take the disjoint union of all these canons,
	and lastly add the join vertices,
	which all is \IFPC{}-definable.
\end{proof}

\paragraph{CFI Graphs of Color Class Joins}
Now, we use color class joins as base graphs.
We introduce terminology for CFI graphs over color class joins.
Let $G_1,\dots, G_\ell$ be colored base graphs with the same number of colors
and let $h \in \FF_2$.
We transfer the notion of part and join vertices from $H := \ccjoin{G_1, \dots, G_\ell}$ to $\StructA := \CFI{H,h}$.
The \defining{$G_i$-part} of $\StructA$
is the set of vertices originating from a vertex or edge of $G_i$ in $H$.
These vertices are called \defining{part vertices of $G_i$}.
A vertex is just a part vertex, if it is a part vertex of some $G_i$.
The remaining vertices are the \defining{join vertices}.

We consider a special class of individualizations of~$\StructA$.
Let $\tup{\vertA} \in \StructVA^*$.
A part of~$\Struct$ is \defining{$\tup{\vertA}$\nobreakdash-pebbled}
if $\vertA_i$ is a part vertex of that part for some $i \in [k]$.
Otherwise, the part is \defining{$\tup{\vertA}$\nobreakdash-unpebbled}.
The set of \defining{$\tup{\vertA}$\nobreakdash-pebbled-part vertices} $\partvert{\Struct}{\tup{\vertA}}$
is the set of all join vertices and all part vertices of all $\tup{\vertA}$\nobreakdash-pebbled parts.
The set of \defining{$\tup{\vertA}$\nobreakdash-pebbled-part individualizations} $\partind{\Struct}{\tup{\vertA}}$ is the set of all individualizations of $\partvert{\Struct}{\tup{\vertA}}$.
We now introduce a technical notion that we will use later.
\begin{defi}[Unpebbled-Part-Distinguishing]
For a tuple $\tup{\vertA} \in \StructVA^\ell$,
a relation $R \subseteq \StructV^k$ is \defining{$\tup{\vertA}$\nobreakdash-unpebbled-part-distinguishing} if there are $m \in [k]$
and $i \neq j \in [\ell]$
such that 
\begin{enumerate}
	\item both the $G_i$-part and the $G_j$-part of $\StructA$ are $\tup{\vertA}$-unpebbled,
	\item there is a $\tup{\vertB} \in R$
	such that $\vertB_m$ is a part vertex of $G_i$, and 
	\item for every $\tup{\vertC} \in R$,
	the vertex $\vertC_m$ is not a part vertex of $G_j$.
\end{enumerate}
\end{defi}
\begin{lem}
	\label{lem:orbit-unpebble-part-distinguishing}
	Let $\tup{\vertA} \in \StructVA^\ell$ and
	$R \subseteq \StructV^k$.
	Assume that $G_i \not\iso G_j$,
	that $G_i$ and $G_j$ are $\tup{\vertA}$\nobreakdash-unpebbled,
	and some $\tup{\vertB} \in R$ contains a vertex of a $\tup{\vertA}$-unpebbled part.
	If $R$ is an orbit of $(\StructA, \tup{\vertA})$,
	then $R$ is $\tup{\vertA}$-unpebbled-part-distinguishing.
\end{lem}
\begin{proof}
	Let the part of~$\vertB_m$ be $\tup{\vertA}$-unpebbled
	and let this part be the $G_n$-part.
	Assume that~$R$ is an orbit of $(\StructA, \tup{\vertA})$.
	If the part of~$\vertC_m$ is neither~$G_i$ nor~$G_j$ for all $\tup{\vertC}\in R$,
	then~${n \notin \set{i,j}}$.
	Thus,~$R$ is $\tup{\vertA}$-unpebbled-part-distinguishing
	because the $G_i$-part and the $G_n$-part are $\tup{\vertA}$\nobreakdash-unpebbled,~$\vertB_m$ is a part vertex of~$G_n$,
	and for every~$\tup{\vertC}\in R$,~$\vertC_m$ is not a part vertex of~$G_i$.
	
	Otherwise, there is a $\tup{\vertC} \in R$ such that
	the part of~$\vertC_m$ is without loss of generality~$G_i$.
	Then no automorphism can map~$\vertC_m$ to a vertex whose part is~$G_j$
	because $G_i \not\iso G_j$ (Lemma~\ref{lem:color-class-join-orbits}).
	Thus,~$\vertC_m$ is not a part vertex of~$G_j$ for every $\tup{\vertC} \in R$ because~$R$ is an orbit.
	It follows that~$R$ is $\tup{\vertA}$-unpebbled-part-distinguishing.
\end{proof}

\subsection{Quantifying over Pebbled-Part Individualizations}
\label{sec:quantifying-pebbled-part-individualizations}

We now define an extension of the $k$-variable counting logic $\Ck{k}$
which allows for quantifying over pebbled-part individualizations. 
Our (unnatural) extension $\Pk{k}$ can only be evaluated on CFI graphs over color class joins and will be a tool to show $\WSCof{\IFPC{}}$-undefinability.
The benefit of this logic is that we can characterize it via a pebble game.
The formulas of $\Pk{k}$ are defined as follows. Let~$\sig$ be a signature.
Whenever $\formA(\tup{x})$ is a $\Ck{k}[\sig,\indRel_P]$\nobreakdash-formula
(for a binary relation symbol~${\indRel_P}\notin \sig$), then \[(\existsP{\indRel_P}.\qspace \formA)(\tup{x})\]
is a $\Pk{k}[\sig]$\nobreakdash-formula.
$\Pk{k}[\sig]$-formulas can be combined as usual in $\Ck{k}$
with Boolean operators and counting quantifiers.
Note that $\exists^P$-quantifiers \emph{cannot} be nested.
The logic is evaluated on CFI graph over colors class joins:
Let $G_1,\dots ,G_\ell$ be colored base graphs with the same number of colors,
$g \in \FF_2$, and
$\Struct = \CFI{\ccjoin{G_1,\dots, G_n},g}$.
The $\exists^P$-quantifier has the following semantics:
\[(\existsP{\indRel_P}.\qspace\formA)^{\Struct} := \setcond*{\tup{\vertA} \in \StructV^{|\tup{x}|}}{
	\tup{\vertA} \in \formA^{(\Struct, \indRel_P^\Struct)}
	\text { for some } {\indRel_P^\Struct} \in \partind{\Struct}{\tup{\vertA}}}.\]
That is, the $\exists^P$-quantifier quantifies over a pebbled-part individualization for the free variables of $\formA$.
Note that the quantifier does not bind first-order variables.

We now characterize $\Pk{k}$ by an Ehrenfeucht-Fraïssé-like pebble game,
which is an extension of the bijective $k$-pebble game.
The \defining{$\Pk{k}$-game} is played by Spoiler and Duplicator.
There are two types of pebbles.
First, there a~$k$ pebble pairs $(p_i, q_i)$, one for each $i \in[k]$,
that will be used as in the bijective $k$-pebble game.
Second, there is a pebble pair $(a_i, b_i)$ for each $i \in \nat$.
The game is played on CFI graphs over color class joins~$\StructA$ and~$\StructB$
satisfying $|\StructVA| = |\StructVB|$ (otherwise Spoiler wins immediately).
A position in the game is a tuple $(\StructA, \indRel_P^\StructA, \tup{\vertA}; \StructB, \indRel_P^\StructB, \tup{\vertB})$,
where $\tup{\vertA} \in \StructVA^{\leq k}$ and
$\tup{\vertB} \in \StructVB^{\leq k}$ are of the same length,
and~$\indRel_P^\StructA$ and~$\indRel_P^\StructB$ individualize vertices (possibly none) originating from up to~$k$ part or join vertices in~$\StructA$ and~$\StructB$ respectively.
A pebble~$p_j$ is placed on~$\vertA_i$ and the corresponding pebble~$q_j$ is placed on~$\vertB_i$ for some $j \in [k]$
(the exact pebble pair $(p_j,q_j)$ placed on the $i$-th entries will not matter).
The pebble~$a_i$ is placed on the $i$-th vertex individualized by $\indRel_P^\StructA$
and~$b_i$ on the $i$-th vertex individualized by~$\indRel_P^\StructB$.
The initial position is $(\StructA, \emptyset, (); \StructB, \emptyset, ())$,
where~$()$ denotes the empty tuple.

Spoiler can perform two kinds of moves.
A \defining{regular move} proceeds as in the bijective $k$-pebble game:
Spoiler picks up a pair of pebbles $(p_i, q_i)$.
Then Duplicator provides a bijection $\gamebij\colon \StructVA \to \StructVB$.
Spoiler places $p_i$ on $\vertA \in \StructVA$ and $q_i$ on $\gamebij(\vertA)$.

A \defining{P-move} proceeds as follows and can only be performed once by Spoiler
(that is, if ${\indRel_P^\StructA}={\indRel_P^\StructB}=\emptyset$):
Spoiler chooses to play the $a_i$-pebbles exactly on the
$\tup{\vertA}$-pebbled-part vertices of~$\StructA$
or to play the $b_i$-pebbles exactly on the
$\tup{\vertB}$-pebbled-part vertices of~$\StructB$.
In the first case, 
Duplicator responds by placing the $b_i$-pebbles on exactly  the $\tup{\vertB}$-pebbled-part vertices of~$\StructB$.
In the second case, Duplicator places the $a_i$-pebbles on exactly  the $\tup{\vertA}$-pebbled-part vertices of~$\StructA$.

If after a round there is no pebble-respecting local isomorphism
of the pebble-induced substructures of $\StructA$ and $\StructB$, then Spoiler wins.
Duplicator wins if Spoiler never wins.

\begin{lem}
	\label{lem:pk-iff-pk-game}
	For every $k \geq 3$,
	Spoiler has a winning strategy in the $\Pk{k}$-game 
	at position $(\StructA, \emptyset,\tup{\vertA}; \StructB, \emptyset, \tup{\vertB})$
	if and only if the logic $\Pk{k}$ distinguishes $(\StructA, \tup{\vertA})$ and $(\StructB, \tup{\vertB})$.
\end{lem}
\begin{proof}
	Let $k \geq 3$.
	We first consider positions $(\StructA, \indRel_P^\StructA, \tup{\vertA};\StructB, \indRel_P^\StructB, \tup{\vertB})$ where Spoiler already has performed the $P$-move,
	i.e., $\indRel_P^\StructA,  {\indRel_P^\StructB} \neq \emptyset$.
	Then the remaining game is essentially just the 
	bijective $k$-pebble game,
	where for each~$i\in \nat$ the vertices pebbled by~$a_i$ and~$b_i$ are put in a unique singleton relation.
	Spoiler has a winning strategy at the position $(\StructA, \indRel_P^\StructA, \tup{\vertA}; \StructB, \indRel_P^\StructB, \tup{\vertB})$ if and only if $\Ck{k}$
	distinguishes $(\StructA, \indRel_P^\StructA, \tup{\vertA})$ and  $(\StructB, \indRel_P^\StructB, \tup{\vertB})$ because, for $k\geq 3$,
	we can define the $i$-th vertex individualized by $\indRel_P^\StructA$ or $\indRel_P^\StructB$.

	We now prove by induction on the number of rounds
	that if Spoiler has a winning strategy in the $\Pk{k}$-game at position $(\StructA, \emptyset,\tup{\vertA}; \StructB, \emptyset, \tup{\vertB})$,
	then $\Pk{k}$ distinguishes $(\StructA, \tup{\vertA})$ and $(\StructB, \tup{\vertB})$.
	If Spoiler wins without performing a $P$-move,
	then Spoiler wins the bijective pebble game and $\Ck{k}$ and thus also $\Pk{k}$ distinguishes $(\StructA, \tup{\vertA})$ and $(\StructB, \tup{\vertB})$.
	So assume Spoiler eventually performs a $P$-move.
	
	Assume that  Spoiler performs a regular move and picks up the $i$-th pebble pair $(p_i,q_i)$.
	The argument is essentially the same as for the bijective $k$-pebble game:
	For every bijection $\gamebij \colon \StructVA \to \StructVB$,
	there is a~$\vertC \in \StructVA$ such that 
	Spoiler wins the $\Pk{k}$-game 
	when placing~$p_i$ on $\vertC$ and~$q_i$ on~$\gamebij(\vertC)$.
	For all these positions, there is a $\Pk{k}$-formula distinguishing them by the inductive hypothesis.
	So some Boolean combination of these distinguishing formulas
	is satisfied by a different number of vertices in $(\StructA, \tup{\vertA})$ and $(\StructB, \tup{\vertB})$
	and we can distinguish them using a counting quantifier.
	
	Now assume that Spoiler performs a $P$-move:
	By symmetry, assume Spoiler places $a_i$\nobreakdash-pebbles on all $\tup{\vertA}$\nobreakdash-pebbled-part vertices of~$\StructA$ (inducing the individualization~$\indRel_P^\StructA$).
	Then, for every placement of the $b_i$\nobreakdash-pebbles by Duplicator on the $\tup{\vertB}$\nobreakdash-pebbled-part vertices of~$\StructB$
	(inducing~$\indRel_P^\StructB$),
	Spoiler has a winning strategy in the bijective $k$-pebble game at position
	$(\StructA, \indRel_P^\StructA, \tup{\vertA}; \StructB, \indRel_P^\StructB, \tup{\vertB})$ (no $P$-move is allowed anymore).
	Then, as argued before, there is a $\Ck{k}$\nobreakdash-formula $\formA$
	distinguishing $(\StructA, \indRel_P^\StructA, \tup{\vertA})$ and $(\StructB, \indRel_P^\StructB, \tup{\vertB})$.
	So the $\Pk{k}$-formula
	$\existsP{\indRel_P}.\qspace \formA$
	distinguishes $(\StructA, \tup{\vertA})$ and $(\StructB, \tup{\vertB})$.
	
	To show the other direction,
	we prove by induction on the quantifier depth that if
	a $\Pk{k}$\nobreakdash-formula~$\formA$ distinguishes $(\StructA, \tup{\vertA})$ and $(\StructB, \tup{\vertB})$,
	then Spoiler has a winning strategy in the $\Pk{k}$-game at position $(\StructA, \emptyset,\tup{\vertA}; \StructB, \emptyset, \tup{\vertB})$.
	If $\formA$ is actually a $\Ck{k}$-formula, then Spoiler wins the bijective $k$-pebble game and so in particular the $\Pk{k}$-game.
	If $\formA$ is $\formB \land \formB'$, $\formB \lor \formB'$, or $\neg \formB$, then one of $\formB$ and $\formB'$ distinguishes $(\StructA, \tup{\vertA})$ and $(\StructB, \tup{\vertB})$.
	If $\formA$ is a counting quantifier $\exists^{\leq i} x.\qspace \formB$,
	then Spoiler performs a regular move.
	Because $\formA$ has at most $k-1$ free variables, Spoiler can pick up a pair of pebbles $(p_i,q_i)$.
	Whatever bijection $\gamebij$ Duplicator chooses,
	there is a vertex~$\vertC$ such that, by symmetry,~$\vertC$ satisfies~$\formB$ 
	in~$\StructA$ but~$\gamebij(\vertC)$ does not satisfy~$\formB$ in~$\StructB$.
	That is,
	Spoiler places $p_i$ on $\vertC$ and $q_i$ on $\gamebij(\vertC)$
	and wins by the induction hypothesis.
	
	To the end, assume that $\formA$ is the $\Pk{k}$-formula $\existsP{\indRel_P}.\qspace\formB$.
	By symmetry, we assume that $(\StructA, \tup{\vertA})$ satisfies $\formA$
	but $(\StructB, \tup{\vertB})$ does not.
	So there is a ${\indRel_P^\StructA} \in \partind{\StructA}{\tup{\vertA}}$
	satisfying $\tup{\vertA} \in \formB^{(\StructA, \indRel_P^\StructA)}$
	such that, for every ${\indRel_P^\StructB} \in \partind{\StructB}{\tup{\vertB}}$,
	it holds that  $\tup{\vertB} \notin \formB^{(\StructB, \indRel_P^\StructB)}$.
	Because $\formB$ is a $\Ck{k}$ formula,
	Spoiler has a winning strategy in the $k$-bijective pebble game at position 
	$(\StructA, \indRel_P^\StructA, \tup{\vertA}; \StructB, \indRel_P^\StructB, \tup{\vertB})$.
	By performing a $P$-move and placing the $a_i$ pebbles according to $\indRel_P^\StructA$,
	Spoiler obtains a winning strategy in the $\Pk{k}$ game at position
	$(\StructA, \emptyset, \tup{\vertA}; \StructB, \emptyset, \tup{\vertB})$.
\end{proof}

Note that the former lemma only holds for $k\geq3$ because with fewer variables
we cannot define the $i$-th individualized vertex in $\Ck{k}$
and thus cannot check local isomorphisms.
Modifying the logic such that the lemma hold for every $k$
only complicates matters and is not needed in the following.

\begin{lem}
	\label{lem:pk-CFI undistinguishable}
	Let $G_1, \dots, G_{k+1}$ be colored base graphs,
	each with $c > k$ color classes,
	such that $\CFI{G_i,0}\kequiv{k}\CFI{G_i, 1}$ for each $i \in [k+1]$. 
	Then Duplicator has a winning strategy in the $\Pk{k}$-game
	played on $\CFI{\ccjoin{G_1, \dots, G_{k+1}},0}$ and $\CFI{\ccjoin{G_1, \dots, G_{k+1}},1}$.
\end{lem}
\begin{proof}
	Let $\StructA = \CFI{\ccjoin{G_1, \dots, G_{k+1}},0}$ and $\StructB = \CFI{\ccjoin{G_1, \dots, G_{k+1}},1}$.
	In this proof, we call the $G_j$-part in $\StructA$ and $\StructB$ just the $j$-th part.
	Let $V_J\subseteq \StructVA=\StructVB$ be the set of join vertices and $V_j\subseteq \StructVA=\StructVB$ be the set of part vertices of the $j$-th part for every $j \in [k+1]$
	(recall that CFI graphs are defined over the same vertex set).

	We show that
	Duplicator is able to maintain the following invariant.
	At every position ${(\StructA, \indRel_P^\StructA, \tup{\vertA}; \StructB, \indRel_P^\StructB, \tup{\vertB})}$ during the game there is  an isomorphism $\autoB \colon (\StructB, \tup{\vertB}, \indRel_P^\StructB) \to (\StructB', \tup{\vertB}, \indRel_P^\StructB)$ for some $\StructB'$,
	which moves the twisted edge, so $\StructVB' = \StructVB$,
	and there is a $j \in[k+1]$
	satisfying the following:
	\begin{enumerate}
		\item \label{itm:not-p-move} If Spoiler has not performed the $P$-move (i.e.,~the $a_i$ and $b_i$ pebbles are not placed), then the part of every $\vertA_i$ and every $\vertB_i$
		is not the $j$-th one.
		\item \label{itm:p-move} If Spoiler has performed the $P$-move, then
		no vertex of the $j$-th part is individualized by $\indRel_P^\StructA$ or $\indRel_P^\StructB$.
		\item \label{itm:auto-all-apart-j} There is an isomorphism
		$\autoA \colon (\StructA, \indRel_P^\StructA, \tup{\vertA})[\StructVA \setminus V_j] \to 
		 (\autoB(\StructB), \indRel_P^\StructB, \tup{\vertB})[\StructVB \setminus V_j]$ respecting the parts,
		 that is, $\autoA$ maps for all $i\neq j$ the $i$-th part of $\StructA$ 
		 to the $i$-th part of $\autoB(\StructB)$.
		 \item \label{itm:kequiv-part-j}
		 $(\StructA,\indRel_J, \tup{\vertA})[V_J \cup V_j] \kequiv{k}
		 (\autoB(\StructB),\autoB(\indRel_J), \tup{\vertB})[V_J \cup V_j]$ for some individualization~$\indRel_J$ of~$V_J$.
	\end{enumerate}
	Note that Property~\ref{itm:kequiv-part-j} is satisfied either by none or all such individualizations.
	Clearly, the invariant is satisfied initially.
	Assume that it is Spoiler's turn and
	Spoiler performs a regular move.
	Spoiler picks up a pebble pair $(p_i, q_i)$.
	Now Duplicator has to provide a bijection~$\gamebij$.
	The bijection~$\gamebij$ is constructed in the following way.
	We start with the isomorphism $\autoA$ of Property~\ref{itm:auto-all-apart-j} and extend it
	on the $j$-th part as follows:
	
	If Spoiler has performed the $P$-move already,
	Duplicator uses the bijection $\gamebij'$
	of Duplicator's winning strategy on $(\StructA,\indRel_J, \tup{\vertA})[V_J \cup V_j] \kequiv{k}
	(\autoB(\StructB),\autoB(\indRel_J), \tup{\vertB})[V_J \cup V_j]$
	for some and thus every $\indRel_J$ of~$V_J$
	(Property~\ref{itm:kequiv-part-j}) to extend $\autoA$ on $V_j$:
	\[\gamebij(\vertC) := \begin{cases}
		\auto(\vertC) & \text{if } \vertC \notin V_j,\\
		\gamebij'(\vertC) &\text{otherwise.}
	\end{cases}\]
	Because in this game the individualization~$\indRel_J$ is used for~$\StructA$
	and the image~$\autoB(\indRel_J)$ is used for~$\StructB$,
	the bijection~$\gamebij'$ necessarily has to map the $i$-th $\indRel_J$-individualized vertex 
	to the $i$-th ${\autoB(\indRel_J)}$\nobreakdash-individualized vertex.
	That is,~$\gamebij$ and~$\gamebij'$ necessarily agree on the join vertices,
	that is, $\gamebij(\vertC) = \gamebij'(\vertC)$ for all $\vertC \in V_J$.
	Spoiler places $p_i$ on $\vertC$ and $q_i$ on $\gamebij(\vertC)$.
	The pebbles still induce a local isomorphism
	(because~$\autoA$ is an isomorphism and $\gamebij'$ is given by a winning strategy).
	Properties~\ref{itm:not-p-move} and~\ref{itm:p-move}
	are obviously satisfied.
	If $\vertC \notin V_j$, then $\autoA(\vertC) = \gamebij(\vertC)$
	and the isomorphism~$\autoA$ still satisfies Property~\ref{itm:auto-all-apart-j}.
	If in particular $\vertC \notin V_J$,
	then Property~\ref{itm:kequiv-part-j} is satisfied because the new pebble is not placed on ${(\StructA,\indRel_J, \tup{\vertA})[V_J \cup V_j]}$ respectively on ${(\autoB(\StructB),\autoB(\indRel_J), \tup{\vertB})[V_J \cup V_j]}$.
	If $\vertC \in V_j$,
	then the pebbles are placed according to a winning strategy of Duplicator
	and thus Property~\ref{itm:kequiv-part-j} is satisfied, too.
	If otherwise $\vertC \in V_j$, then $\vertC$ is not in the domain of~$\autoA$ and thus~$\autoA$ satisfies Property~\ref{itm:auto-all-apart-j}.
	Property~\ref{itm:kequiv-part-j} is satisfied because $\gamebij(\vertC) = \gamebij'(\vertC)$
	and~$\gamebij'$ was obtained by a winning strategy of Duplicator.

	If Spoiler has not performed the $P$-move,
	then Duplicator extends~$\auto$ as follows.
	There is another $\tup{\vertA}$-unpebbled part different from the $j$-th one
	because there are at most $k-1$ pebbles placed
	(one pebble pair is picked up).
	Let this part be the $\ell$-th part for some $\ell \neq j$
	and let $\set{\bVertA,\bVertB}$
	be the twisted edge between $\StructA$ and $\autoB(\StructB)$, which is contained in the $j$-th part (by Property~\ref{itm:auto-all-apart-j}).
	There is a path from one of $\bVertA$ or $\bVertB$
	into the $\ell$-th part in $\ccjoin{G_1, \dots, G_{k+1}}$
	only using vertices of the $j$-th and $\ell$-th part and one join vertex $\bVertC$
	such that it does not use the origin of the pebbled vertices:
	Both~$G_j$ and~$G_\ell$ are connected, do not contain any pebbles,
	and there are $c > k$ color classes,
	so one join vertex is not pebbled.
	Hence, there is a path-isomorphism $\autoB' \colon (\StructB', \tup{\vertB}) \to (\StructB'', \tup{\vertB})$
	moving the twist from the $j$-th into the $\ell$-th part along that path.
	Thus, 
	we can extend the restriction \[\restrictVect{\autoB'}{\StructV \setminus V_j \setminus V_\ell}\circ\restrictVect{\auto}{\StructV \setminus V_\ell} \colon (\StructA, \tup{\vertA})[\StructVA \setminus V_j\setminus V_\ell] \to 
	(\autoB'(\autoB(\StructB)), \tup{\vertB})[\StructVB \setminus V_j\setminus V_\ell]\]
	to $V_j$ because the twist is now in the $\ell$-th part.
	That is, we obtain an isomorphism
	$\auto' \colon (\StructA, \tup{\vertA})[\StructVA \setminus V_\ell] \to 
	(\autoB'\circ\autoB(\StructB), \tup{\vertB})[\StructVB \setminus V_\ell]$
	which agrees with $\autoA$ on $\StructV\setminus V_j \setminus V_\ell$
	apart from gadget vertices originating from $\bVertC$
	and with the edge vertices originating from the edge incident to $\bVertC$
	into the $j$-th and $\ell$-th part.
	These are the only vertices in $\StructV \setminus V_j \setminus V_\ell$
	for which the isomorphism $\autoB'$ is not the identity.
	Duplicator extends~$\autoA$ on~$V_j$ using~$\autoA'$ to the bijection~$\gamebij$:
	\[\gamebij(\vertC) := \begin{cases}
		\autoA(\vertC) & \text{if } \vertC \notin V_j,\\
		\autoA' (\vertC) &\text{otherwise.}
	\end{cases}\]
	Spoiler places~$p_i$ on~$\vertC$ and~$q_i$ on~$\gamebij(\vertC)$.
	If $\vertC \notin V_j$
	(and so $\gamebij(\vertC)=\autoA(\vertC) \notin V_j$),
	then Properties~\ref{itm:not-p-move} to~\ref{itm:kequiv-part-j} are clearly satisfied (now for $\ell$ instead for $j$)
	because $\autoA$ is an isomorphism.
	For the same reason, the pebbles induce a local isomorphism.

	So assume $\vertC \in V_j$.
	By Property~\ref{itm:not-p-move},
	the first pebble is placed on the $j$-th part
	and the $\ell$-th part does not contain a pebble.
	Then the restriction of~$\gamebij$ to $\StructVA \setminus V_\ell$
	can be turned into 
	is an isomorphism $ (\StructA, \emptyset, \tup{\vertA}\vertC)[\StructVA \setminus V_\ell] \to 
	(\StructB'', \emptyset, \tup{\vertB}\gamebij(\vertC))[\StructVB \setminus V_\ell]$
	by applying a local automorphism of the gadget of $\bVertC$ (which is not pebbled)
	that moves the twist from the $j$-th into the $\ell$-th part according to $\autoB'$.
	In particular, the pebbles induce a local isomorphism
	and Property~\ref{itm:auto-all-apart-j} holds.
	Property~\ref{itm:kequiv-part-j} is satisfied
	because the $\ell$-th part does not contain a pebble:
	If $\CFI{G_\ell,0} \kequiv{k} \CFI{G_\ell,1}$,
	then $(\StructA, \indRel_J)[V_J \cup V_\ell]  \kequiv{k}
	(\StructB'',\autoB(\indRel_J))[V_J \cup V_\ell]$, too,
	because $(\StructA, \indRel_J)[V_J \cup V_\ell]$
	just extends $\CFI{G_\ell,0}$ by gadgets for the join vertices,
	which are all fixed by~$\indRel_J$
	(and likewise for $(\StructB'',\autoB(\indRel_J))[V_J \cup V_\ell]$ and $\CFI{G_\ell,1}$).
	
	Finally, let Spoiler perform the $P$-move.
	Assume by symmetry that Spoiler places the~$a_i$ pebbles on the~$\tup{\vertA}$-pebbled parts.
	Duplicator places the pebble~$b_i$ on~$\autoA(a_i)$ for all $i$.
	Because the pebbles are placed according to the isomorphism~$\autoA$,
	there is a pebble-respecting local isomorphism.
	Properties~\ref{itm:not-p-move} and~\ref{itm:p-move}
	are clearly satisfied.
	Property~\ref{itm:auto-all-apart-j} is satisfied by $\autoA$
	and Property~\ref{itm:kequiv-part-j} is satisfied
	because no pebble is placed in the $j$-th part
	(similar to the $\ell$-th part in the former case).
\end{proof}

\subsection{Nesting Operators to Define the CFI Query is Necessary}
\label{sec:nesting-operators-is-necessary}

We use CFI graphs over color class joins of CFI graphs
to construct a new class of base graphs,
for which $\WSCIkof{2}{\IFPC}$ defines the CFI query
but $\WSCIof{\IFPC}$ does not.
A graph is \defining{asymmetric} if it only has the trivial automorphism.
Fix a class $\GraphClass := \setcond{G_i}{i \in \nat}$ of totally ordered base graphs such that for all $i\in \nat$
\begin{enumerate}
	\item $G_i$ has maximal degree $3$,
	\item $G_i$ has treewidth at least~$i$, and
	\item the coloring of $G_i$ is unnecessary (i.e.,
the uncolored $E$-reduct of $G_i$ is asymmetric).
\end{enumerate} 
Such a class exists because we can obtain from some graph~$G_i'$ of treewidth~$i$
(e.g., a clique of size $i+1$)
a $3$-regular graph of treewidth at least~$i$ as follows:
If~$G_i'$ has a vertex $\vertA$ of degree greater than~$3$,
then we obtain a new graph~$G_i''$ by splitting~$\vertA$ off into two vertices
(onto which we equally distribute the edges incident to~$\vertA$)
and connecting them via an edge.
Contracting this edge yields back~$G_i'$.
Thus,~$G_i'$ is a minor of~$G_i''$ and thus the treewidth of~$G_i''$
is at least the treewidth of~$G_i'$.
We repeat this procedure until every vertex has degree $3$.
To make~$G_i$ asymmetric, we attach paths of distinct lengths to the vertices in~$G_i$.

\begin{lem}
	\label{lem:cfi-cfi-not-in-Ck}
	$\CFI{\CFI{G_k,g},0} \kequiv{k} \CFI{\CFI{G_k,g},1}$ for every $k \in \nat$ and $g \in \FF_2$.
\end{lem}
\begin{proof}
	Let $k \in \nat$ and $g \in \FF_2$.
	The graph $G_k$ has treewidth at least $k$.
	The CFI construction does not decrease the treewidth
	because $G_k$ is a minor of $\CFI{G_k, g}$ (cf.~\cite{DawarRicherby07}).
	Hence, $\CFI{G_k, g}$ has treewidth at least $k$ and 
	$\CFI{\CFI{G_k,g},0} \kequiv{k} \CFI{\CFI{G_k,g},1}$ by Lemma~\ref{lem:k-connected-treewidth-CFI-equiv}.
\end{proof}

\begin{lem}
	\label{lem:cfi-cfi-in-wsci-2}
	$\WSCIkof{2}{\IFPC}$ defines the CFI query for $\CFI{\CFIm{\GraphClass}}$.
\end{lem}
\begin{proof}
	\IFPC{} distinguishes $2$-orbits of $\GraphClass$
	because $\GraphClass$-graphs are totally ordered.
	By Corollary~\ref{cor:canonize-CFI if-base-wsci},
	$\WSCIof{\IFPC{}}$ canonizes $\CFI{\GraphClass}$.
	From Lemma~\ref{lem:extended-CFI if-cfi} it follows that $\WSCIof{\IFPC{}}$ canonizes $\CFIm{\GraphClass}$
	and so also distinguishes $2$-orbits of $\CFIm{\GraphClass}$.
	Again due to Corollary~\ref{cor:canonize-CFI if-base-wsci},
	$\WSCIkof{2}{\IFPC{}}$ canonizes $\CFI{\CFIm{\GraphClass}}$.
	In particular, $\WSCIkof{2}{\IFPC{}}$ defines the CFI query for $\CFI{\CFIm{\GraphClass}}$.
\end{proof}

To show that $\WSCIof{\IFPC}$ does not define the CFI query for 
$\CFI{\CFIm{\GraphClass}}$, we will use the following idea:
Suppose that a $\WSCIof{\IFPC}$-formula $\formA$ defines the CFI query for 
$\CFI{\CFIm{\GraphClass}}$
and we evaluate $\formA$
on $\CFI{\CFImk{\ell}{G}}$ for some $\ell > |\tup{p}|$.
If $\formA$ always defines choice-sets containing only tuples of vertices in parameter-pebbled parts,
then the twist can be moved in the parameter-unpebbled parts.
Because all choices are made in parameter-pebbled parts,
the output formula of $\formA$, which is an \IFPC{}-formula,
essentially has to define the CFI query for $\CFI{\CFI{\GraphClass}}$,
which is not possible (Lemma~\ref{lem:cfi-cfi-not-in-Ck}).
Otherwise, 
$\formA$ makes a choice in parameter-unpebbled parts.
But for that, $\formA$ has to distinguish $\CFI{G,0}$ from $\CFI{G,1}$
to define orbits of $\CFI{\CFImk{\ell}{G}}$.
So the choice \IFPC{}-formula
has to define the CFI query for $\CFI{\GraphClass}$,
which is also not possible.
Making this idea formal requires some effort.

\newcommand{\GraphClassOrb}{\GraphClass_{\text{orb}}}
\newcommand{\GraphClassCFI}{\GraphClass_{\text{cfi}}}

\begin{lem}
	\label{lem:cfi-cfi-not-in-wsci}
	$\WSCIof{\IFPC}$ does not define the CFI query for $\CFI{\CFIm{\GraphClass}}$.
\end{lem}
\begin{proof}	
	For the sake of a contradiction, suppose that $\formA$ is a $\WSCIof{\IFPC}$-formula defining the CFI query for $\CFI{\CFIm{\GraphClass}}$.
	Without loss of generality,~we assume that $\formA$ binds no variable twice.
	Because $\Iof{\IFPC{}}=\IFPC{}$,
	we can assume that $\formA$ is a $\WSCof{\IFPC{}}$-formula.
	
	Let $\formB_1(\tup{x}_1), \dots, \formB_m(\tup{x}_m)$ be all WSC-fixed-point operators that are subformulas of $\formA$.
	For the moment assume that all free variables $\tup{x}_i$ are element variables.
	Let the number of distinct variables of $\formA$ be $k$
	and let $\ell := \ell(k) \geq \max\set{k,3}$ for some function $\ell(k)$ to be defined later.
	We consider the subclass $\CFI{\CFImk{\ell+1}{\GraphClass}} \subseteq \CFI{\CFIm{\GraphClass}}$.
	We partition~$\GraphClass$ as follows:
	First, let $\GraphClassOrb$ be the set of all $G \in \GraphClass$
	such that, for every $g \in \FF_2$, there are $h\in \FF_2$,  $j \in [m]$,
	and a $|\tup{x}_j|$-tuple $\tup{\vertA}$ of vertices of $\CFI{\CFImk{\ell+1}{G,g},h})$
	such that 
	\begin{enumerate}
		\item all choice-sets during the evaluation of $\formB_j(\tup{\vertA})$ on $\CFI{\CFImk{\ell+1}{G,g},h}$
		are indeed orbits and
		\item one of these choice-sets is $\tup{\vertA}$-unpebbled-part-distinguishing.
	\end{enumerate}
	Second, set $\GraphClassCFI := \GraphClass\setminus \GraphClassOrb$.
	Clearly, at least one of $\GraphClassOrb$ and $\GraphClassCFI$ contains infinitely many graphs, which is a contradiction as shown in the two following claims.
	
	\begin{clm}
		\label{clm:orb-infinite}
		The size of $\GraphClassOrb$ is finite.
	\end{clm}
	\begin{claimproof}
		We claim that there is an \IFPC{}-formula defining the CFI query for $\CFI{\GraphClassOrb}$.
		First, we show that there is an \IFPC{}-interpretation that,
		for every $G \in \GraphClassOrb$ (and even in $\GraphClass$),
		maps a CFI graph $\CFI{G,g}$ and an $h \in \FF_2$ to the graph $(\CFI{\CFImk{\ell+1}{G,g},h},\indRel)$ such that~$\indRel$ individualizes the vertices of $\ell+1$ many $\CFI{G,0}$-parts,
		$\ell+1$ many $\CFI{G,1}$-parts, and all join vertices of $\CFI{\CFImk{\ell+1}{G,g},h}$.
		The following mappings are definable by \IFPC{}-interpretations:
		\begin{enumerate}[label=(\alph*)]
			\item Map $\CFI{G,g}$ to the base graph $G$ (which is an ordered graph).
			\item Map $G$ and $g'\in \FF_2$ to  $(\CFI{G,g'}, \leq)$, such that $\leq$ is a total order on $\CFI{G,g'}$.
			This map is \IFPC{}-definable because $G$ is of degree at most $3$.
			\item Map $\CFI{G,g}$, $(\CFI{G,0}, \leq_0)$, and $(\CFI{G,1}, \leq_1)$ to $(\CFImk{\ell+1}{G,g}, \indRel')$,
			where~$\leq_0$ and~$\leq_1$ are total orders and~$\indRel'$ individualizes all vertices of the $\ell+1$ parts of
			$\CFI{G,0}$ and $\CFI{G,1}$ as well as the join vertices.
			\item Finally, map $(\CFImk{\ell+1}{G,g}, \indRel')$ and $h\in\FF_2$ to $(\CFI{\CFImk{\ell+1}{G,g},h},\indRel)$ such that~$\indRel$ individualizes the required vertices.
			This map is \IFPC{}-definable since $\CFImk{\ell+1}{G,g}$ is of bounded degree:
			The graph $G$ has color class size $1$ and has maximal degree $3$,
			so $\CFI{G,h}$ has color class size $4$ and maximal degree $3$.
			That is, $\CFImk{\ell+1}{G,g}$ has color class size $4\ell + 4$
			and degree at most $4\ell+ 4$ (the join vertices),
			which is a constant.
		\end{enumerate}
		By composing these \IFPC{}-interpretations, we obtain the required one.
		We now show that we can simulate each WSC-fixed-point operator $\formB_j$
		on $(\CFI{\CFImk{\ell+1}{G,g},h},\indRel)$ in \IFPC{}
		such that we determine the parity of $\CFI{G,g}$. 
		
		For every WSC-fixed-point operator~$\formB_j(\tup{x}_j)$, we consider every $h \in \FF_2$
		and every possible $|\tup{x}_j|$-tuple~$\tup{\vertA}$ of $\indRel$\nobreakdash-individualized vertices for the parameters~$\tup{x}_j$.
		We simulate the evaluation of the WSC-fixed-point operator~$\formB_j$
		in $\IFPC{}$ as follows:
		Because~$\formA$ is a $\WSCof{\IFPC}$ formula,
		the step, choice, witnessing, and output formula of~$\formB_j$
		are $\IFPC$-formulas.
		We evaluate the choice formula and
		check whether all tuples in the defined relation
		are composed of $\tup{\vertA}$\nobreakdash-pebbled-part vertices.
		If that is the case, we resolve the choice deterministically
		using the lexicographical order of~$\indRel$ on the tuples
		(recall that~$\indRel$ individualizes all $\tup{\vertA}$-pebbled-part vertices by construction).
		Next, we evaluate the step formula.
		The simulation is continued
		until there is a choice-set not solely composed of vertices of  $\tup{\vertA}$-pebbled-part parts.
		Because the choice-set is by definition of~$\GraphClassOrb$ an orbit,
		it is $\tup{\vertA}$-unpebbled distinguishing
		by Lemma~\ref{lem:orbit-unpebble-part-distinguishing}.
		So the choice-set contains (at some index) vertices of the $\CFI{G,g}$-parts and either of the $\CFI{G,0}$-parts or of $\CFI{G,1}$-parts:
		Because isomorphic parts are in the same orbit, 
		either vertices of all isomorphic parts  or none of them occur
		because they cannot be distinguished.
		At least one of the $\CFI{G,0}$-parts and the $\CFI{G,1}$-parts each is not $\tup{\vertA}$-pebbled because $|\tup{\vertA}| \leq k < \ell+1$.
		So one of these is isomorphic to the $\CFI{G,g}$-parts.
		The graphs $\CFI{G,0}$ and $\CFI{G,1}$ were added by the interpretation,
		so we can actually remember their parity and thus defined the parity of $\CFI{G,g}$.
		
		It remains to prove that such a combination of~$j$,~$h$, and~$\tup{\vertA}$ always exists.
		By construction of~$\GraphClassOrb$,
		they exist when testing all possible $|\tup{x}_j|$-tuples for $\tup{\vertA}$
		(and not only those of $\indRel$\nobreakdash-individualized vertices).
		But, because $k < \ell+1 $, there exists always an automorphism~$\auto$ (ignoring the individualization) mapping $\tup{\vertA}$ to the vertices of the $\CFI{G,0}$-parts and the $\CFI{G,1}$-parts (because $\CFI{G,g}$ is isomorphic to one of them). 
		Because $\formB_j$ has no access to the individualization,~$\formB_j$ is satisfied by~$\tup{\vertA}$
		if and only if~$\formB_j$ is satisfied by $\auto(\tup{\vertA})$.
		So indeed \IFPC{} defines the CFI query for $\CFI{\GraphClassOrb}$.
		
		Now, for the sake of contradiction,
		suppose that the size of~$\GraphClassOrb$ is infinite.
		So for every~$k$, there is a $j\geq k$ such that $G_j \in \GraphClassOrb$
		and $\CFI{G_j,0} \kequiv{k} \CFI{G_j,1}$ by Lemma~\ref{lem:k-connected-treewidth-CFI-equiv}.
		This contradicts that \IFPC{} defines the CFI query for $\CFI{\GraphClassOrb}$.
	\end{claimproof}
	
	\begin{clm}
		\label{clm:cfi-infinite}
		The size of $\GraphClassCFI$ is finite.
	\end{clm}
	\begin{claimproof}
		Assume that $\GraphClassCFI$ is infinite.
		So there is an $\ell' > \ell$ such that $G:=G_{\ell'}\in \GraphClassCFI$.
		By definition of~$\GraphClassCFI$,
		there is a $g \in \FF_2$ such that for all $h \in \FF_2$, all $j \in [m]$, and all $|\tup{x}_j|$-tuples~$\tup{\vertA}$ of $\CFI{\CFImk{\ell+1}{G,g},h})$
		\begin{enumerate}
			\item some choice-set during the evaluation of $\formB_j(\tup{\vertA})$ on $\CFI{\CFImk{\ell+1}{G,g},h}$
			is not an orbit or
			\item all choice-sets are not $\tup{\vertA}$-unpebbled-part-distinguishing.
		\end{enumerate} 
		We claim that there is a $\Pk{\ell}$-formula equivalent to the CFI-query-defining formula~$\formA$
		on $\CFI{\CFImk{\ell+1}{G,g},0}$ and $\CFI{\CFImk{\ell+1}{G,g},1}$.
		We first translate every WSC-fixed-point operator $\formB_i(\tup{x}_i)$
		(for $i \in [m]$)
		into an equivalent \IFPC{}-formula which uses a fresh relation symbol~$\indRel_P$.
		The relation~$\indRel_P$ is intended to be interpreted as an $\tup{\vertA}$-pebbled-part individualization when using~$\tup{\vertA}$ for the parameters~$\tup{x}_i$.
		Let $i \in [m]$ be arbitrary.
		Again, the step, choice, witnessing, and output formulas of~$\formB_i$ are \IFPC{}-formulas because~$\formB_i$ is a $\WSCof{\IFPC}$-formula.
		We simulate~$\formB_i$ by an \IFPC{}-formula using the relation~$\indRel_P$.
		If all choice-sets during the evaluation for~$\tup{\vertA}$
		are not $\tup{\vertA}$-unpebbled-part-distinguishing,
		then all choice-sets contain solely tuples
		composed out of the vertices individualized by~$\indRel_P$
		(otherwise a choice-set would be $\tup{\vertA}$-unpebbled-part-distinguishing by Lemma~\ref{lem:orbit-unpebble-part-distinguishing}).
		So if all tuples in a choice-set are composed of the individualized vertices,
		we can resolve all choices deterministically using
		the lexicographical order on tuples given by~$\indRel_P$.
		Otherwise, some choice-set during the evaluation
		will not be an orbit by definition of~$\GraphClassCFI$
		and
		we immediately evaluate to false
		because we make (or will make) a choice out of a non-orbit.
		If this was never the case,
		we check in the end whether all choices were indeed witnessed.
		Let $\tilde{\formB}_{i}(\tup{x}_i)$ be an \IFPC{}-formula,
		which implements exactly this approach to simulate~$\formB_i$.
		The formula  $\tilde{\formB}_{i}(\tup{x}_i)$ can be constructed to use not more than $\ell(k)$ distinct variables
		such that no variable is bound twice.
		
		For every number~$n$,
		every $k$-variable \IFPC{}-formula not binding variables twice can be unwound into
		a $\Ck{k}$-formula that is equivalent on structures of size up to~$n$ (see~\cite{Otto1997}).
		So, for $\ell = \ell(k)$ and $n=|\CFI{\CFImk{\ell+1}{G,g},0}|$,
		we can unwind $\tilde{\formB}_{i}(\tup{x}_i)$ to a  $\Ck{\ell}$-formula $\tilde{\formB}_{i}^n(\tup{x}_i)$.
		Then the $\Pk{\ell}$-formula
		\[\formA_i(\tup{x}_i):= \existsP{\indRel_P}.\qspace\tilde{\formB}_{i}^n(\tup{x}_i)\]
		is equivalent to $\formB_i$ on $\CFI{\CFImk{\ell+1}{G,g},0}$ and $\CFI{\CFImk{\ell+1}{G,g},1}$.
		To see this, note that~$\tilde{\formB}_{i}$ evaluates equally for every pebbled-part individualization $\indRel_P$:
		The individualization~$\indRel_P$ is only used to resolve choices.		
		If all choice-sets were witnessed as orbits (not fixing~$\indRel_P$),
		then indeed~$\tilde{\formB}_{i}$ evaluates equally for all~$\indRel_P$ (neither the step, the choice, the witnessing, nor the output formula  use~$\indRel_P$).
		If a choice-set is not witnessed as orbit, this is surely also true for all~$\indRel_P$.
	
		We replace each WSC-fixed-point operator~$\formB_i$ by~$\formA_i$ in~$\formA$
		and continue to unwind the remaining \IFPC{}-part of~$\formA$
		yielding a $\Pk{\ell}$-formula equivalent to~$\formA$
		on $\CFI{\CFImk{\ell+1}{G,g},0}$ and $\CFI{\CFImk{\ell+1}{G,g},1}$,
		which by assumption distinguishes the two graphs.
		This contradicts Lemma~\ref{lem:pk-CFI undistinguishable}:
		The ordered graph~$G$ has more than~$\ell$ vertices (and thus color classes),
		so $\CFI{G,g'}$ has also more than~$\ell$ color classes for every $g' \in \FF_2$.
		Thus, we have
		$\CFI{\CFI{G,g'},0} \kequiv{k} \CFI{\CFI{G,g'},1}$ by Lemma~\ref{lem:cfi-cfi-not-in-Ck}.
		So, by  Lemma~\ref{lem:pk-CFI undistinguishable},
		Duplicator has a winning strategy in the $\Pk{\ell}$-game
		played on $\CFI{\CFImk{k+1}{G,g},0}$ and $\CFI{\CFImk{k+1}{G,g},1}$.
		Hence, $\Pk{\ell}$ does not distinguish the graphs
		by Lemma~\ref{lem:pk-iff-pk-game}, which is a contradiction.
	\end{claimproof}

	Finally, we have to consider the case of free numeric variables.
	Let the numeric variables in~$\formA$ be~$\tup{\numVarA}$.
	For each numeric variable, there is a closed numeric $\WSCof{\IFPC{}}$-term bounding its value. Let these terms be~$\tup{\termA}$.
	Because we cannot evaluate the $\WSCof{\IFPC{}}$-terms,
	we construct upper-bound-defining \IFPC{}-terms $\tup{\termB}$
	only depending on the size of the input structure.
	To obtain these, we construct upper-bound-defining terms recursively: For $0$,~$1$,~$\cdot$, and~$+$ this is obvious.
	For a counting quantifier $\#\tup{\uniVarA}\tup{\numVarA} \leq \tup{\termA}'.\qspace\formA$,
	the upper bound is defined by the \IFPC{}-term
	$(\#\tup{\uniVarA}.\qspace\tup{\uniVarA}=\tup{\uniVarA}) \cdot \termB'_1 \cdot \ldots \cdot \termB'_{|\numVarA|}$,
	where~$\termB'_i$ is the upper-bound-defining \IFPC{}-term recursively obtained for~$\termA'_i$ for every $i \in [|\numVarA|]$.
	Note that we do not recurse on~$\formA$ and in particular 
	not on a WSC-fixed-point operator.
	For $G \in \GraphClass$, set $N(G) := \set{0,...,\termB_1^\StructA} \times \dots \times \set{0,...,\termB_{|\tup{\numVarA}|}^\StructA}$ to be the possible values for
	the numeric variables for $\StructA = \CFI{\CFImk{\ell+1}{G,g},h}$ (which only depends on $|\StructVA|$).
	
	To partition~$\GraphClass$ into~$\GraphClassOrb$ and~$\GraphClassCFI$,
	we not only consider $|\tup{x}_j|$-tuples~$\tup{\vertA}$ of vertices of  $\CFI{\CFImk{\ell+1}{G,g},h}$
	but tuples $\tup{\vertA}\tup{a}$
	for $\tup{a} \in  N(G)$.
	To extend Claim~\ref{clm:orb-infinite} to numeric variables,
	we have to test all possible values for the free numeric variables
	according to the upper-bound-defining term~$\tup{\termB}$
	and find the unpebbled-part-distinguishing choice-set.
	To adapt the proof of Claim~\ref{clm:cfi-infinite},
	we obtain for every $i \in [m]$
	and every tuple of values $\tup{a} \in N(G)$ for the free numeric variables,
	a $\Pk{\ell}$-formula $
	\formA_{i}^{\tup{a}}(\tup{x}_i):= \existsP{\indRel_P}.\qspace\tilde{\formB}_{i}^{n,\tup{a}}(\tup{x}_i)$
	satisfying $\tup{\vertA} \in (\formA_{i}^{\tup{a}})^\Struct$ if and only if $\tup{\vertA}\tup{a} \in \formB_{i}^\Struct$
	for every $\Struct \in \setcond{\CFI{\CFImk{\ell+1}{G,g},h}}{h\in \FF_2}$.
	In the same way,
	free numeric variables of \IFPC{}-formulas are eliminated
	and we can use these formulas to construct the $\Pk{\ell}$-formula
	equivalent to~$\formA$.
\end{proof}

Now we can show that we cannot avoid the additional operators to canonize CFI graphs (which implies defining the CFI query) as shown in Corollary~\ref{cor:canonize-CFI if-base-wsci}.
\begin{proof}[Proof of Theorem~\ref{thm:cfi-wsci-wsci}]
	We consider the class of base graphs $\CFIm{\GraphClass}$.
	$\WSCIkof{2}{\IFPC}$ defines the CFI query for $\CFI{\CFIm{\GraphClass}}$
	by Lemma~\ref{lem:cfi-cfi-in-wsci-2}.
	But $\WSCIof{\IFPC}$ does not define the CFI query for $\CFI{\CFIm{\GraphClass}}$
	by Lemma~\ref{lem:cfi-cfi-not-in-wsci}.
\end{proof}

The reason why $\WSCIof{\IFPC{}}$ does not define this CFI query is not that $\WSCIof{\IFPC{}}$ does not have Ebbinghaus' reduct property.
Also $\WSCIof{\IFPC{}}$ with the global reduct semantics is not able to define it, either.
\begin{lem}
	$\WSCIof{\IFPC{}}$ with global reduct semantics does not define
	the CFI query for $\CFI{\CFIm{\GraphClass}}$.
\end{lem}
\begin{proof}
	To deal with the reduct semantics,
	we show that we can assume that every CFI-query-defining formula~$\formA$
	uses both the edge relation and the color-encoding preorder.
	Then the global reduct semantics does not move to a proper reduct
	and Lemma~\ref{lem:cfi-cfi-not-in-wsci} applies.
	
	If~$\formA$ does not use the edge relation, then~$\formA$ clearly cannot define the CFI query.
	So consider the case that~$\formA$ does not use the color-encoding preorder but uses the edge relation.
	For all $k$ large enough and all $g,h\in\FF_2$, the coloring of $\CFI{\CFImk{k}{G_i,g},h}$ is unnecessary
	by Lemma~\ref{lem:unnecessary-coloring-cfim} because
	the coloring of $G_i$ is unnecessary by the definition of~$\GraphClass$.
	Recall, that this means that the uncolored graph $\reduct{\CFI{\CFImk{k}{G_i,g},h}}{E}$ has the same orbit partition as the colored graph $\CFI{\CFImk{k}{G_i,g},h}$.
	Because the orbit partition does not change when omitting the coloring, defining and witnessing
	orbits only gets more difficult when not using the coloring.
	So we can assume that~$\formA$ uses both the edge relation and the preorder.
\end{proof}
\noindent We emphasize that our proofs only use \IFPC{}-interpretations in interpretation operators.
\begin{cor}
	\label{cor:wsci-le-wsci2}
	$\IFPC{} < \WSCIof{\IFPC{}} < \WSCIkof{2}{\IFPC{}}$.
\end{cor}

It seems natural that $\WSCIkof{n}{\IFPC{}} <  \WSCIkof{n+1}{\IFPC{}}$ for every $n\in \nat$.
Possibly, this hierarchy can be shown by iterating our construction (e.g.~$\CFI{(\CFImsym)^n(\GraphClass)}$,
where $(\CFImsym)^n$ denotes $n$ applications of the $\CFImsym$-operator).

Combining Corollary~\ref{cor:wsci-le-wsci2} with the fact that $\IFPC$ is closed under interpretations and thus $\WSCIof{\IFPC} \leq \IFPCWSC$, separates $\IFPC$ from $\IFPCWSC$.
\begin{cor}
	\label{cor:ifpc-le-wsc}
	$\IFPC{} < \IFPCWSC$.
\end{cor}

We have seen that nesting WSC-fixed-point and interpretation operators increases the expressiveness of \IFPC{}.
However, we have not shown that the interpretation operator is necessary for that
or whether WSC-fixed-point operators suffice.
We will show in the next section that the interpretation operator indeed increases the expressiveness.

\section{Separating \IFPCWSC{} from \IFPCWSCI{}}
\label{sec:separating-wsc-wsci}

This section separates \IFPCWSC{} from \IFPCWSCI{},
that is, the interpretation operator strictly increases expressiveness.
To show this, we will define a class of structures $\GraphClass$ without non-trivial automorphisms.
Thus, there are only singleton orbits
and the WSC-fixed-point operator can be expressed
by a (non-WSC) fixed-point operator:
Either all choice-sets are singletons or the WSC-fixed-point operator evaluates to false.
We will show that isomorphism of $\GraphClass$-structures is
not definable in \IFPC{} and thus not in \IFPCWSC{}.
However, we will show that there is an \IFPC{}-interpretation
reducing $\GraphClass$-isomorphism to
isomorphism of CFI graphs (over ordered base graphs)
and thus $\GraphClass$\nobreakdash-isomorphism is \IFPCWSCI{}-definable.

We will combine two known constructions.
We start with CFI graphs.
In the next step, we will modify the CFI graphs to remove all non-trivial automorphisms.
This will be achieved by gluing a CFI graph to a so-called multipede~\cite{GurevichShelah96}.
The multipedes are structures without non-trivial automorphisms, for which \IFPC{} fails to define the orbit partition.
To prove that isomorphism of this gluing is not \IFPC{}-definable, either,
we will combine winning strategies of Duplicator in the bijective $k$-pebble game
of CFI graphs and multipedes.
In order to successfully combine the strategies,
we will require that the base graphs of the CFI graphs have large connectivity.
For the multipedes, we will show the existence of sets of vertices
with pairwise large distance.
The multipede can be removed by an \FO{}-interpretation
reducing the isomorphism problem of the gluing to the CFI query
and hence the isomorphism problem is \IFPCWSCI{}-definable.

\subsection{Multipedes}

We now review the multipede structures of Gurevich and Shelah~\cite{GurevichShelah96}.
These structures are also based on the CFI gadgets.
Most importantly, these structures are \defining{asymmetric}, i.e.,
their only automorphism is the trivial one.

The base graph of a multipede is a bipartite graph $G=(\gadgetsb,\feetb,E, \leq)$,
where~$\leq$ is some total order on $\gadgetsb \cup \feetb$ and every vertex in~$\gadgetsb$ has degree~$3$.
We obtain the  \defining{multipede structure}~$\multipede{G}$ as follows.
For every base vertex $\bVertA \in \feetb$, there is a fresh vertex pair ${\feetOf{\bVertA} = \set{\vertA_0,\vertA_1}}$
called a \defining{segment}. 
We also call $\bVertA\in \feetb$ itself a segment.
A single vertex~$\vertA_i$ is called a \defining{foot}.
Vertices $\bVertB \in \gadgetsb$ are called \defining{constraint vertices}.

For every constraint vertex $\bVertB \in \gadgetsb$,
a degree-$3$ CFI gadget with three edge-vertex-pairs $\set{a_0^j,a_1^j}$ (for $j \in [3]$) is added.
Let $\neighbors{G}{\bVertB} = \set{\bVertA^1,\bVertA^2,\bVertA^3}$
(the order of the~$\bVertA^i$ is given by~$\leq$).
Then~$a_i^j$ is identified with the foot~$\vertA^j_i$ for every $j \in [3]$ and $i \in \FF_2$.
To construct multipedes, we use the relation-based CFI gadgets,
that is we do not add further vertices to the feet
but a ternary relation $R$ containing all triples $(\vertA_{i_1}^1,\vertA_{i_2}^2,\vertA_{i_3}^3)$ with $i_1+i_2+i_3 = 0$.
Because all constraint vertices in~$\gadgetsb$ have degree~$3$,
the construction yields ternary structures of the fixed signature $\set{R, \spleq}$.
The coloring~$\spleq$ is again obtained from~$\leq$
(the feet of a segment have the color of the segment).
Figure~\ref{fig:multipede} shows an example for a multipede.
We collect properties of multipedes.

\begin{figure}
		\newcommand{\vertexpair}[3]{
		\begin{scope}[#3]
			\node [vertex, fill=#1] (v0) at (-0.3,0) {};
			\node [vertex, fill=#1] (v1) at (0.3,0) {};
		\end{scope}
	}
	\subfloat[A base graph of a multipede.]{
		\begin{tikzpicture}
			
			\path[draw = none, use as bounding box] (-1.5,-0.5) rectangle (5.5, 2.5);

			\node[basevertex, fill=edgeColA] (u1) at (0,0){};
			\node[basevertex, fill = edgeColB] (u2)  at (1,0){};
			\node[basevertex, fill = edgeColC] (u3)  at (2,0){};
			\node[basevertex, fill = edgeColD] (u4)  at (3,0){};
			\node[basevertex, fill = edgeColE] (u5)  at (4,0){};
			
			\node[basevertex, dashed,very thick, fill = white] (v1)  at (1,2){};
			\node[basevertex, dotted,very thick, fill = white] (v2)  at (3,2){};
			\node[basevertex, very thick,fill = white] (v3)  at (2,2){};
			
			\node at (5,0) {$W$};
			\node at (5,2) {$V$};
			
			\draw [-, thick]
			(v1) edge (u1) edge (u2) edge (u3)
			(v2) edge (u3) edge (u4) edge (u5)
			(v3) edge (u3) edge (u5) edge (u1);
		\end{tikzpicture}
		\label{fig:multipede-base}
	}
	
	\subfloat[The corresponding multipede.]{
		\begin{tikzpicture}[rotate=-90, scale = 2, ]
			
			\path [draw=none, use as bounding box] (-0.8,-1) rectangle (0.5,5);
			\node at (-0.3,4.5) {$0$};
			\node at (0.3,4.5) {$1$};
			
			\begin{scope}[every node/.style={transform shape}]
		
		\vertexpair{edgeColA}{edgeColA!10!white}{name prefix = a, shift={(0,0)}}
		\vertexpair{edgeColB}{edgeColB!10!white}{name prefix = b, shift={(0,1)}}
		\vertexpair{edgeColC}{edgeColC!10!white}{name prefix = c, shift={(0,2)}}
		\vertexpair{edgeColD}{edgeColD!10!white}{name prefix = d, shift={(0,3)}}
		\vertexpair{edgeColE}{edgeColE!10!white}{name prefix = e, shift={(0,4)}}

		\draw[black, dashed, rounded corners, draw, line width = 0.4mm, ]
		($(av0)+(-0.04,0.04)$) -- ($(bv0)+(-0.02,0)$) -- ($(cv0)+(-0.05,-0.04)$)
		($(av0)+(0.04,0.04)$) -- ($(bv1)+(-0.01,0)$) -- ($(cv1)+(0.05,-0.04)$)
		($(av1)+(0.04,0.04)$) -- ($(bv1)+(0.09,0)$) -- ($(cv0)+(+0.05,-0.04)$)
		($(av1)+(-0.02,0.04)$) -- ($(bv0)+(0.04,0)$) -- ($(cv1)+(-0.05,-0.04)$);
		
		\draw[black, dotted, rounded corners, draw, line width = 0.4mm, ]
		($(cv0)+(-0.05,0.04)$) -- ($(dv0)+(-0.02,0)$) -- ($(ev0)+(-0.04,-0.04)$)
		($(cv0)+(0.05,0.04)$) -- ($(dv1)+(-0.01,0)$) -- ($(ev1)+(0.04,-0.04)$)
		($(cv1)+(0.05,0.04)$) -- ($(dv1)+(0.09,0)$) -- ($(ev0)+(+0.04,-0.04)$)
		($(cv1)+(-0.05,0.04)$) -- ($(dv0)+(0.04,0)$) -- ($(ev1)+(-0.02,-0.04)$);
		
		\draw[, rounded corners, draw, line width = 0.4mm, line cap=round]
		 ($(cv0)+(-0.06,0)$) edge[bend right=45] ($(av0)+(-0.04,-0.04)$) edge[bend left=45] ($(ev0)+(-0.04,0.04)$)
		($(cv1)+(-0.06,0)$) edge[bend right=50] ($(av0)+(0.04,-0.04)$) edge[bend left=40] ($(ev1)+(0.02,0.04)$)
		($(cv1)+(0.03,0)$) edge[bend right=45] ($(av1)+(0.02,-0.04)$) edge[bend left=55] ($(ev0)+(+0.04,0.04)$)
		($(cv0)+(0.02,0)$) edge[bend right=50] ($(av1)+(-0.04,-0.04)$) edge[bend left=50] ($(ev1)+(-0.04,0.04)$);
			
		\end{scope}
	\end{tikzpicture}
	\label{fig:multipede-multipede}
	}
	\caption{Construction of a multipede from a bipartite base graph: Figure~\protect\subref{fig:multipede-base} shows a base graph of a multipede
		with vertex set $V\cup W$.
	Each vertex in $W$ has a unique color and each vertex in $V$ has a unique line style.
	Figure~\protect\subref{fig:multipede-multipede} shows the resulting multipede:
	The vertices $u_0$ are shown on the top, the vertices $u_1$ at the bottom.
	Feet inherit the color of the base vertices, vertices of the same segment have the same color.
	Each tuple in the ternary relation of the multipede is drawn by a line.
	The line style of the constraint vertices indicate which tuple arises from which constraint vertex. 
	}
	\label{fig:multipede}
\end{figure}
 
\begin{defi}[Odd Graph]
	A bipartite graph $G=(\gadgetsb,\feetb,E, \leq)$ is \defining{odd}
	if, for every $\emptyset\neq \footbset \subseteq \feetb$,
	there exists a $\bVertB \in \gadgetsb$ such that
	$|\footbset \cap \neighbors{G}{\bVertB}|$ is odd.
\end{defi}
\begin{lemC}[\cite{GurevichShelah96}]
	\label{lem:odd-multipede-asymmetric}
	If $G$ is an odd and ordered bipartite graph,
	then $\multipede{G}$ is asymmetric.
\end{lemC}

\begin{defi}[$k$-Meager]
	A bipartite graph $G=(\gadgetsb,\feetb,E,\leq)$ is called \defining{$k$-meager},
	if every set $\footbset \subseteq \feetb$ of size $|\footbset| \leq 2k$
	satisfies $|\setcond{\bVertB \in \gadgetsb}{\neighbors{G}{\bVertB} \subseteq \footbset}| \leq 2|\footbset|$.
\end{defi} 

If $G$ is a $k$-meager bipartite graph,
then~$\Ck{k}$ cannot distinguish between the two feet of a segment in
the structure $\multipede{G}$~\cite{GurevichShelah96}.
We need to investigate the argument for this statement in more details.
For a set $\footbset \subseteq \feetb$, we define the \defining{feet-induced} subgraph
\[\footinduced{G}{\footbset} := G[\footbset \cup \setcond{\bVertB \in \gadgetsb}{\neighbors{G}{\bVertB} \subseteq \footbset}]\]
to be the subgraph induced by the feet in $\footbset$
and all constraint vertices only adjacent to feet in $\footbset$.
We extend the notation to the multipede: $\footinduced{\multipede{G}}{\footbset}$
is the substructure induced by all feet whose segment is contained in $\footbset$.
For a tuple $\tup{\vertA}$ of feet of $\multipede{G}$,
we define
\[\segmentOf{\tup{\vertA}} := \setcond[\big]{\bVertA \in \feetb}{\vertA_i \in \feetOf{\bVertA} \text{ for some } i\leq |\tup{\vertA}|}\]
to be the set of the segments of all feet in $\tup{\vertA}$.

\begin{lemC}[\cite{GurevichShelah96}]
	\label{lem:multipede-winning-strategy}
	Let $G$ be a $k$-meager bipartite graph,
	$\StructA = \multipede{G}$,
	and
	$\tup{\vertA},\tup{\vertB} \in \StructVA^k$.
	If there is a local automorphism $\autoA \in \autGroup{\footinduced{\StructA}{\segmentOf{\tup{\vertA}\tup{\vertB}}}}$
	with $\autoA(\tup{\vertA}) = \tup{\vertB}$,
	then $(\StructA, \tup{\vertA}) \kequiv{k} (\StructA, \tup{\vertB})$.
\end{lemC}

Gurevich and Shelah~\cite{GurevichShelah96} showed
that odd and $k$-meager bipartite graphs
exist.
However, we need a more detailed understanding of these graphs
for our construction.
In particular, we are interested in sets of segments which have pairwise large distance.

\begin{defi}[$k$-Scattered]
	For a bipartite graph $G=(\gadgetsb,\feetb,E)$,
	a set $\footbset \subseteq \feetb$ is called \defining{$k$\nobreakdash-scattered}
	if every distinct $\bVertA, \bVertB \in \footbset$ have distance at least
	$2k$ in $G$.
\end{defi}

We require pairwise distance~$2k$ for a $k$-scattered set $\footbset$
because we are actually only interested in segments
(which alternate with constraint vertices in paths in~$G$).

\begin{lem}
	\label{lem:odd-meager-distance-base-graph}
	For every $k \in \nat$,
	there is an odd and $k$-meager bipartite graph
	$G=(\gadgetsb,\feetb,E)$
	and a $k$-scattered set $\footbset \subseteq \feetb$
	of size $|\footbset| \geq k^2$.
\end{lem}
\begin{proof}
	The multipedes constructed in~\cite{GurevichShelah96}
	are sparse graphs generated by a random process.
	We show that these graphs contain a $k$-scattered set of size at least $k^2$:
	Fix an arbitrary $k \in \nat$.
	We start with a vertex set~$U$ of size~$n$.
	An event $E(n)$ is called \defining{almost sure}
	if the probability that $E$ occurs tends to~$1$ when~$n$ grows to infinity.
	Pick $\epsilon < \inv{(2k + 3)}$
	and add independently with probability $p=n^{-2+\epsilon}$ for every $3$-element subset $\set{\bVertA_1,\bVertA_2,\bVertA_3} \subseteq U$
	a vertex $\bVertB$ to $\gadgetsb$ adjacent to the $\bVertA_i$.
	Now, it is almost surely possible to remove at most $n/4$ vertices from $U$ forming subgraphs that are exceedingly dense
	(and all constraint vertices in~$\bVertB$ incident to them),
	which results in an odd and $k$-meager bipartite graph~\cite{GurevichShelah96}.
	We show that in this graph a $k$-scattered set $\footbset$ of size $k^2$ exists almost surely.
	
	\begin{clm}
		\label{clm:almost-sure-degree}
		Almost surely, every vertex $\bVertA \in \feetb$ has degree at most $n^{\inv{1.5k}}$.
	\end{clm}
	\begin{claimproof}
		Let $m$ be the least integer such that $m > n^{\inv{(1.5k)}}$.
		The probability that a given vertex has degree larger than $n^{\inv{(1.5k)}}$ is $\binom{n^2}{m} \cdot p ^{m}$.
		So the probability that some vertex has degree larger than $n^{\inv{(1.5k)}}$ is at most the following: 
		\begin{align*}
			n \cdot \binom{n^2}{m} \cdot p ^{m} &= n \cdot \binom{n^2}{m} \cdot (n^{-2+\epsilon})^m \\
			&\leq \binom{n^2}{m} \cdot n^{(-2 + \inv{(2k+3)})m + 1}\\
			&\leq \left(\frac{e n^2}{m}\right)^m \cdot n^{(-2 + \inv{(2k+3)})m + 1}\\
			&\leq \left(e n^{2-\inv{(1.5k)}}\right)^m \cdot n^{(-2 + \inv{(2k+3)})m + 1}\\
			&= e^m \cdot n^{(2-\inv{(1.5k)})m}\cdot n^{(-2 + \inv{(2k+3)})m + 1}\\
			&= e^m \cdot n^{(\inv{(2k+3)}-\inv{(1.5k)})m +1}\\
			&= n^{\inv{(\ln n)}\cdot m} \cdot n^{(\inv{(2k+3)}-\inv{(1.5k)})m +1}\\
			& = n ^ {(\inv{(2k+3)} -\inv{(1.5k)} + \inv{(\ln n)})m +1}\\
			&\leq n ^ {(\inv{(2k+3)} -\inv{(1.5k)} + \inv{(\ln n)})(n^{\inv{(1.5k)}}+1) +1}
			= o(1).
		\end{align*}
		The last step holds because $\inv{(2k+3)} -\inv{(1.5k)} < 0$  and thus for sufficiently large $n$ and
		$\ell := \inv{(2k+3)} - \inv{(1.5k)} + \inv{(\ln n)}<0$
		the term $\ell(n^{\inv{(1.5k)}}+1) +1$ tends to $-\infty$.
	\end{claimproof}
	\begin{clm}
		\label{ref:claim-distance}
		If every vertex $\bVertA \in \feetb$ has  degree at most $n^{\inv{(1.5k)}}$,
		then for every vertex $\bVertA \in \feetb$ at most $3^kn^{\frac{2}{3}}$  vertices $\bVertA'\in\feetb$
		have distance at most $2k$ to $\bVertA$.
	\end{clm}
	\begin{claimproof}
		By construction, every vertex in $\gadgetsb$ has degree $3$.
		So at most $3n^{\inv{(1.5k)}}$ vertices have distance $2$ to $\bVertA$.
		Repeating the argument, 
		at most $(3n^{\inv{(1.5k)}})^k = 3^k n^{\frac{2}{3}}$
		vertices in $\feetb$ have distance $2k$ to $\bVertA$.
	\end{claimproof}
	
	Almost surely, every vertex $\bVertA \in \feetb$ has degree at most $n^{\inv{(1.5k)}}$ (Claim~\ref{clm:almost-sure-degree}).
	We show that in this case a $k$-scattered set $\footbset$ of size at least $k^2$ exists.
	Note that we determined the probability before removing the $\frac{n}{4}$
	``bad'' vertices.
	So we first remove some set of size at most $\frac{n}{4}$ from~$U$
	(which only decreases the degree of the remaining vertices).
	We now repeatedly apply Claim~\ref{ref:claim-distance}.
	If~$\footbset$ is a $k$-scattered set, then at most $|\footbset| \cdot 3^kn^{\frac{2}{3}}$ vertices have distance at most $2k$ to a vertex in $\footbset$.
	Then pick one of the other vertices, add it to~$\footbset$, so~$\footbset$ is still $k$-scattered, and repeat.
	In that way we find a $k$-scattered set~$\footbset$
	of size $|\footbset| \geq \frac{3}{4}n \cdot \inv{(3^kn^{\frac{2}{3}})} = \frac{3}{4}3^{-k}n^{\frac{1}{3}}$.
	Finally, for sufficiently large $n$, we have that $k^2 \leq \frac{3}{4}3^{-k}n^{\frac{1}{3}}$.
\end{proof}

We want to use a $k$-scattered set $\footbset$ to ensure that
in the bijective $k$-pebble game played on multipedes 
placing a pebble on one foot of a segment in $\footbset$ has no effect on the other segments in $\footbset$.
To make this argument formal, we need to consider how information of pebbles is spread through the multipedes.
For now, fix an arbitrary ordered bipartite graph $G=(\gadgetsb,\feetb, E, \leq)$.

\begin{defi}[Closure]
	Let $\footbset \subseteq \feetb$.
	The \defining{attractor} of $\footbset$ is
	\[\attractor{\footbset} := \footbset \cup \bigcup_{\substack{\vertA \in \gadgetsb,\\ |\neighbors{G}{\vertA} \setminus \footbset| \leq 1}} \neighbors{G}{\vertA}. \]
	The set $\footbset$ is \defining{closed}
	if $\footbset = \attractor{\footbset}$.
	The \defining{closure} $\closure{\footbset}$
	of $\footbset$ is the inclusion-wise minimal closed superset of $\footbset$.
\end{defi}

\begin{lemC}[\cite{GurevichShelah96}]
	\label{lem:closure-meager-size}
	If $G$ is $k$-meager and $\footbset \subseteq \feetb$ of size at most $k$,
	then $|\closure{\footbset}| \leq 2|\footbset|$.
\end{lemC}

Let $\footbset$ be closed.
A set $\footbsetB \subseteq \footbset$ is a \defining{component} of $\footbset$
if $\footinduced{G}{\footbsetB}$
is a connected component of $\footinduced{G}{\footbset}$.
That is, every constraint vertex contained in $\footinduced{G}{\footbset}$ is contained in $\footinduced{G}{\footbsetB}$ or in $\footinduced{G}{\footbset \setminus \footbsetB}$.

\begin{lem}
	\label{lem:closure-components}
	Let $\footbsetA \subseteq \feetb$.
	If $\closure{\footbset} = \footbsetB_1 \cup \footbsetB_2$
	is the disjoint union of two components~$\footbsetB_1$ and~$\footbsetB_2$,
	then~$\footbsetA$ can be partitioned into $\footbsetA_1 \cup \footbsetA_2$
	such that $\closure{\footbsetA_i} = \footbsetB_i$ for every $i \in [2]$.
\end{lem}
\begin{proof}
	Define $\footbsetA_i := \footbsetB_i \cap \footbsetA$ for every $i \in [2]$.
	Let $i \in [2]$ be arbitrary. We show that $\closure{\footbsetA_i} = \footbsetB_i$.
	It is clear that $\closure{\footbsetA_i} \subseteq \footbsetB_i$.
	For the other direction, suppose that $\bVertA \in \footbsetB_i \setminus \closure{\footbsetA_i}$.
	Then there must be a constraint vertex in $\footinduced{G}{\footbsetB}$
	which has a neighbor in $\footbsetB_1$ and another one in $\footbsetB_2$.
	This contradicts that the $\footbsetB_i$ are components of $\closure{\footbset}$.
\end{proof}

\begin{lem}
	\label{lem:closure-singleton-component}
	Let $\footbsetB\subseteq \feetb$ be closed and
	$\bVertA \in \feetb$ have distance at least~$4$ to~$\footbsetB$.
	Then $\closure{\footbsetB \cup \set{\bVertA}} = \footbsetB \cup \set{\bVertA}$
	and~$\bVertA$ forms a singleton component of $\footbsetB \cup \set{\bVertA}$.
\end{lem}
\begin{proof}
	For the sake of contradiction,
	assume that there exists a $\bVertB \in \closure{\footbsetB \cup \set{\bVertA}}  \setminus (\footbsetB \cup \set{\bVertA})$.
	Then there is a constraint vertex~$\bVertC$,
	of which one neighbor is~$\bVertB$
	and the two others are contained in $\closure{\footbsetB \cup \set{\bVertA}}$.
	If the two other neighbors are contained in $\closure{\footbsetB}$,
	then~$\bVertB$ is already contained $\closure{\footbsetB}$,
	which is a contradiction.
	So one of the two neighbors is~$\bVertA$.
	But then~$\bVertA$ has distance~$2$ to~$\footbsetB$,
	which contradicts our assumption.
\end{proof}

\begin{lem}
	\label{lem:closure-scattered-singleton-components}
	Let $k \geq 2$,
	$G$ be $2k$-meager,
	$\footbsetB \subseteq \feetb$ be of size at most $k$,
	and $\footbset\subseteq \feetb$ be $6k$-scattered.
	Then at most $|\footbsetB|$ vertices of $\footbset$
	do not form singleton components in $\closure{\footbset \cup \footbsetB}$.
\end{lem}
\begin{proof}
	Because $\footbset\subseteq \feetb$ is $6k$-scattered,
	all distinct $\bVertA,\bVertB \in \footbset$
	have distance at least~$12k$ in~$G$.
	Hence, for every $\bVertC \in \footbsetB$,
	there is at most one $\bVertA \in \footbset$
	such that~$\bVertC$ has distance less than~$6k$ to~$\bVertA$.
	Let $\footbset' \subseteq \footbset$ be the set of vertices with distance at most~$6k$ to~$\footbsetB$. 
	Thus, $|\footbset'| \leq |\footbsetB|$.
	So, $|\footbsetA' \cup \footbsetB| \leq 2k$ and thus
	$|\closure{\footbsetA'\cup \footbsetB}| \leq 4k$
	by Lemma~\ref{lem:closure-meager-size} because~$G$ is $2k$-meager.
	It follows that $|\closure{\footbsetA'\cup \footbsetB} \setminus \footbsetA' \setminus\footbsetB | \leq 2k$.
	Because every component of $\closure{\footbsetA'\cup \footbsetB}$ contains a vertex of $\footbsetA' \cup \footbsetB$,
	every vertex in $\closure{\footbsetA' \cup \footbsetB}$
	has distance at most~$4k$ to every vertex in~$\footbsetA'\cup\footbsetB$.
	
	Since $|\footbset'| \leq |\footbsetB|$,
	it suffices to show that all vertices in $\footbset\setminus \footbset'$
	form singleton components of $\closure{\footbsetA \cup \footbsetB}$.
	Let $\footbset \setminus \footbset' = \set{\bVertC_1, \dots ,\bVertC_\ell}$
	and set $\footbset_i := \set {\bVertC_1,\dots,\bVertC_i}$ for every ${0 \leq i \leq \ell}$.
	We show by induction on $i \leq \ell$
	that all vertices in $\footbset_i$ form singleton components in ${\closure{\footbset_i\cup \footbset' \cup \footbsetB}}$.
	For $i = 0$, the claim trivially holds.
	So assume $i > 0$.
	As seen before, every vertex in $\closure{\footbset' \cup \footbsetB}$ has distance at most $4k$ to $\footbset' \cup \footbsetB$.
	By construction of $\footbset'$, the vertex~$\bVertC_i$ has distance at least $6k +1$ to $\footbsetB$.
	Hence, $\bVertC_i$ has distance at least $2k \geq 4$ to  $\closure{\footbset' \cup \footbsetB}$. 
	Because~$\footbset$ is  $6k$-scattered,~$\bVertC_i$ has distance at least $12k$ to all other $\bVertC_j$.
	Because, by the inductive hypothesis,
	all vertices in $\footbset_{i-1}$ form singleton components,
	$\bVertC_i$ has distance at least~$4$ to 
	$\closure{\footbset_i \cup \footbset' \cup \footbsetB}$.
	Hence, by Lemma~\ref{lem:closure-singleton-component},
	all vertices in $\footbset_i \setminus \footbset'$
	form singleton components of $\closure{\footbset \cup \footbsetB}$.
\end{proof}

\begin{lem}
	\label{lem:closure-distance}
	Let $k \geq 2$,
	$G$ be $2k$-meager,
	$\footbsetA \subseteq \feetb$ be $6k$-scattered,
	and let $\footbsetB \subseteq \feetb$ be of size at most $k$.
	Then $\footbsetA \cup \footbsetB$ can be partitioned
	into $\footbsetC_1, \dots, \footbsetC_\ell$
	such that $|\closure{\footbsetC_i} \cap \footbsetA| \leq 1$ for every $i\in [\ell]$
	and	$\closure{\footbsetC_1}, \dots, \closure{\footbsetC_\ell}$
	are the components of $\closure{\footbsetA \cup \footbsetB}$.
\end{lem}

\begin{proof}
	We partition $\footbsetA \cup \footbsetB$ using Lemma~\ref{lem:closure-components}
	into $\footbsetC_1, \dots, \footbsetC_\ell$
	such that $\closure{\footbsetC_1}, \dots, \closure{\footbsetC_\ell}$
	are precisely the components of $\closure{\footbsetA \cup \footbsetB}$.	
	Let $i \in [\ell]$ be arbitrary but fixed. We prove that $|\closure{\footbsetC_i} \cap \footbsetA| \leq 1$.
	By Lemma~\ref{lem:closure-scattered-singleton-components},
	all apart from $|\footbsetB|$ vertices in $\footbsetA$ are contained
	in a singleton component.
	If $\closure{\footbsetC_i}$ is such a singleton component, we are done.
	Otherwise, \[|\footbsetC_i| \leq |\footbsetC_i \cap \footbsetA| + |\footbsetB| \leq 2|\footbsetB| \leq 2k\]
	and $|\closure{\footbsetC_i}| \leq 4k$ by Lemma~\ref{lem:closure-meager-size}
	and because $G$ is $2k$-meager.
	Because $\closure{\footbsetC_i}$ is a component,
	all vertices in $\closure{\footbsetC_i}$ have pairwise distance at most $8k$.
	So $|\closure{\footbsetC_i} \cap \footbsetA| \leq 1$,
	because~$\footbset$ is $6k$-scattered (and thus vertices in $\footbset$ have distance at least $12k$).
\end{proof}

\begin{lemC}[\cite{GurevichShelah96}]
	\label{lem:closure-extension}
	Let $\footbset \subseteq \feetb$, $G$ be $k$-meager,
	$\autoA\in \autGroup{\footinduced{\multipede{G}}{\footbset}}$,
	and $|\footbset| \leq k$.
	Then $\autoA$ extends to an automorphism of $\footinduced{\multipede{G}}{\closure{\footbset}}$.
\end{lemC}
\begin{lem}
	\label{lem:closure-extension-new}
	Let $\footbsetB \subseteq \feetb$,~$G$ be $k$-meager,
	$\autoA\in \autGroup{\footinduced{\multipede{G}}{\closure{\footbsetB}}}$,
	and $|\footbsetB| < k$.
	Then for all $\bVertA \in \feetb \setminus \closure{\footbsetB}$
	and all feet $\vertB, \vertB' \in \feetOf{\bVertA}$ (possibly $\vertB=\vertB'$),
	there is an extension~$\autoB$ of~$\autoA$ to an automorphism of $\footinduced{\multipede{G}}{\closure{\footbsetB \cup \set{\bVertA}}}$
	satisfying $\autoB(\vertB) = \vertB'$.
\end{lem}
\begin{proof}
	The condition $\autoB(\vertB) = \vertB'$ defines $\autoB$ uniquely
	on the feet of~$\bVertA$,
	that is, given~$\autoA$,~$\vertB$, and~$\vertB'$ the map~$\autoB$ is determined.
	If $\vertB = \vertB'$, then the~$\autoB$ maps the feet of~$\bVertA$
	to itself and otherwise it exchanges them.
	Using Lemma~\ref{lem:closure-extension} 
	and noting that $|\footbsetB \cup \set{\bVertA}| \leq k$,
	it suffices to show that~$\autoB$ is an automorphism of 
	$\footinduced{\multipede{G}}{\footbsetB \cup \set{\bVertA}}$.
	Because $\bVertA \notin \closure{\footbsetB}$,
	there is no CFI gadget of which one edge-vertex-pair form the feet of~$\bVertA$
	and the other  edge-vertex-pairs are contained in $\closure{\footbsetB}$
	(otherwise~$\bVertA$ would be in the closure).
	So indeed,~$\autoB$ is a local automorphism.
\end{proof}

\begin{lem}
	\label{lem:distance-set-local-aut}
	Let $k \geq 2$,
	$G$ be $2k$-meager,
	$\footbset = \set{\bVertA_1,\dots, \bVertA_\ell} \subseteq \feetb$ be $6k$-scattered,
	${\footbsetB \subseteq \feetb}$ be
	such that $\footbset \cap \closure{\footbsetB} = \emptyset$
	and $|\footbsetB| < k$,
	and $\autoA\in \autGroup{\footinduced{\multipede{G}}{\footbsetB}}$.
	Then for all tuples ${\tup{\vertA},\tup{\vertA}' \in \feetOf{\bVertA_1} \times \cdots \times \feetOf{\bVertA_\ell}}$,
	there exists an extension $\autoB$ of $\autoA$ to an automorphism of
	$\footinduced{\multipede{G}}{\footbset \cup \footbsetB}$
	satisfying $\autoB(\tup{\vertA}) = \tup{\vertA}'$.
\end{lem}
\begin{proof}
	
	Let $\tup{\vertA},\tup{\vertA}' \in \feetOf{\bVertA_1} \times \cdots \times \feetOf{\bVertA_\ell}$.
	For every $i \in [\ell]$, there is by Lemma~\ref{lem:closure-extension-new} an extension~$\autoB_i$ of
	$\autoA$ to $\footinduced{\multipede{G}}{\closure{\footbsetB \cup \set{\bVertA_i}}}$
	satisfying $\autoB_i(\vertA_i) = \vertA'_i$.
	
	By Lemma~\ref{lem:closure-distance}, all $\bVertA_i$
	are in different components of $\closure{\footbsetB \cup \footbset}$
	because~$G$ is $2k$-meager and~$\footbset$ is $6k$-scattered.
	So we can extend every~$\autoB_i$ on all components  of $\closure{\footbsetB \cup \footbset}$,
	on which~$\autoB_i$ is not defined (in particular the ones containing all other~$\bVertA_j$ for $j\neq i$),
	by the identity map and obtain an automorphism of $\footinduced{\multipede{G}}{\closure{\footbsetB \cup \footbset}}$
	(two components are never connected by a CFI gadget).
	Hence, composing all the extended~$\autoB_i$ yields the desired automorphism.
\end{proof}

The previous lemma will be extremely useful in the bijective $k$-pebble game: If the pebbles are placed on the feet in $\footbsetB$,
we can simultaneously for all feet in $\footbset\setminus\closure{\footbsetB}$ place arbitrary pebbles 
and still maintain a local automorphism.
Such sets $\footbset$ will allow us to glue another graph to the multipede
at the feet in $\footbset$.
Whatever restrictions on placing pebbles are imposed by the other graph,
we still can maintain local automorphisms in the multipede.

\begin{lem}
	\label{lem:closure-new-vertices}
	Let $k\geq 2$,
	$G$ be $2k$-meager,
	$\footbset \subseteq \feetb$ be $6k$-scattered,
	and $\footbsetB \subseteq \feetb$ be of size at most~$k$.
	Then $|\closure{\footbsetB} \cap (\footbset \setminus\footbsetB) | \leq |\footbsetB\setminus \footbset|$.
\end{lem}
\begin{proof}
	We partition $\footbset \cup \footbsetB$
	using Lemma~\ref{lem:closure-distance} into $\footbsetC_1, \dots , \footbsetC_j$
	such that the $\closure{\footbsetC_i}$ are the components of $\closure{\footbset \cup \footbsetB }$
	and at most one vertex of $\footbset$
	is contained in one $\footbsetC_i$.
	Up to reordering, assume that for some $\ell\leq j$ 
	the components $\footbsetC_1, \dots , \footbsetC_\ell$
	are all components $\footbsetC_i$ such that $\footbsetC_i \cap \footbsetB \neq \emptyset$
	and, for some $m \leq \ell$,
	the $\footbsetC_1, \dots , \footbsetC_m$
	are all $\footbsetC_i$ such that additionally $\closure{\footbsetC_i} \cap (\footbset \setminus\footbsetB) \neq \emptyset$.
	
	Then clearly $\closure{\footbsetB} \subseteq \bigcup_{i \in [\ell]}
	\closure{\footbsetC_i}$
	and $\ell \leq |\footbsetB|$.
	Because~$G$ is $2k$\nobreakdash-meager,~$\footbset$ is $6k$\nobreakdash-scattered, and $|\footbsetB|\leq k$,
	it follows from Lemma~\ref{lem:closure-distance} that
	$|\closure{\footbsetC_i} \cap \footbset| \leq 1$ for every $i \in [\ell]$.
	Thus, $|\closure{\footbsetC_i} \cap (\footbset\setminus \footbsetB)| = 1$
	for every $i \in [m]$
	and $|\closure{\footbsetC_i} \cap (\footbset\setminus \footbsetB)| = 0$
	for every $m < i \leq \ell$.
	It follows that
	$\closure{\footbsetB} \cap (\footbset \setminus\footbsetB) \subseteq \bigcup_{i=1}^m \closure{\footbsetC_i} \cap (\footbset \setminus\footbsetB)$.
	Hence, $|\closure{\footbsetB} \cap (\footbset \setminus\footbsetB)| \leq m$.
	
	Every vertex in $\footbsetB \cap \footbsetA$
	has to be contained in some $\footbsetC_i$  such that $m < i \leq \ell$
	because, for every $i \leq m$, we have $|\closure{\footbsetC_i} \cap \footbset| \leq 1$ and so
	$\footbsetC_i \cap (\footbsetB \cap \footbset) = \emptyset$
	since $\footbsetC_i$ contains a vertex of $\footbset \setminus \footbsetB$.
	But for $i \leq \ell$, every $\footbsetC_i$ contains at least one vertex of $\footbsetB$.
	That is, $m \leq |\footbsetB\setminus\footbset|$.
\end{proof}

\subsection{Gluing Multipedes to CFI Graphs}

We now glue CFI graphs to multipedes.
First, we alter the CFI-construction used.
Instead of two edge-vertex-pairs for the same base edge $\set{\bVertA, \bVertB}$
(one with origin $(\bVertA,\bVertB)$ and one with origin $(\bVertB,\bVertA)$),
we contract the edges between these two edge-vertex-pairs (which form a matching)
and obtain a single edge-vertex-pair with origin $\set{\bVertA,\bVertB}$.
This preserves all relevant properties of CFI graphs.
In particular, there are parity-preserving \IFPC{}-interpretations
mapping CFI graphs with two edge-vertex-pairs per base edge to the corresponding CFI graphs  with only one edge-vertex-pair per base edge and vice-versa.
In this section, we write $\CFI{H,f}$ for CFI graphs of this modified construction.
Using only one edge-vertex-pair per base edge removes technical details from the following.

Let $G=(\gadgetsb^G,\feetb^G, E^G, \leq^G)$ be an ordered bipartite graph,
let $H=(V^H,E^H, \leq^H)$ be an ordered base graph,
let $f \colon E^H \to \FF_2$,
and let $\footbset \subseteq \feetb^G$ have size $|\footbset| = |E^H|$.
We define the ternary structure $\gluing{\multipede{G}}{\footbset}{\CFI{H, f}}$ called the \defining {gluing}
of the multipede $\multipede{G}  =  (\StructVA, R^{\multipede{G}}, \spleq^{\multipede{G}})$ and the CFI graph $\CFI{H, f}  = (\StructVB, E^{\CFI{H, f}}, \spleq^{\CFI{H, f}})$ at~$\footbset$ as follows (see Figure~\ref{fig:gluing-multipede-cfi-graph} for an illustration):
The $i$-th edge-vertex-pair of $\CFI{H, f}$ is the edge-vertex-pair such that its origin 
$\set{\bVertA,\bVertB}$ is the $i$-th edge in~$H$ according to $\leq^H$.
We start with the disjoint union of $\multipede{G}$ and $\CFI{H, f}$
and identify the $i$-th edge-vertex-pair of $\CFI{H, f}$ (according to $\leq^H$) with the $i$\nobreakdash-th segment in $\footbset$ (according to $\leq^G$).
We finally turn the edges $E^{\CFI{H, f}}$
into a ternary relation by extending every edge $(\vertA,\vertB)\in E^{\CFI{H, f}}$ to the triple $(\vertA,\vertB,\vertB)$.
In this way, we obtain the $\set{\rel,\spleq}$-structure $\gluing{\multipede{G}}{\footbset}{\CFI{H, f}}$:
The relation~$R$ is the union of $R^{\multipede{G}}$ with the triples $(\vertA,\vertB, \vertB)$ defined before.
The total preorder~$\spleq$ is obtained from joining~$\spleq^{\multipede{G}}$ and~$\spleq^{\CFI{H, f}}$
by moving all gadget vertices of $\CFI{H,f}$ to the end.

	\newcommand{\vertexpair}[3]{
	\begin{scope}[#3]
		\draw[#1,fill=#2] (0,0) ellipse (0.5cm and 0.25cm);
		\node [vertex, fill=#1] (v0) at (-0.2,0) {};
		\node [vertex, fill=#1] (v1) at (0.2,0) {};
	\end{scope}
}
\newcommand{\gadget}[3]{
	\begin{scope}[#3]
		\draw[#1, fill=#2] (0,0) ellipse (0.9cm and 0.35cm);
		\node [vertex, fill=#1] (u000) at (-0.6,0) {};
		\node [vertex, fill=#1] (u011) at (-0.2,0) {};
		\node [vertex, fill=#1] (u101) at (0.2,0) {};
		\node [vertex, fill=#1] (u110) at (0.6,0) {};
	\end{scope}
}
\begin{figure}[t]
	\centering

	\begin{tikzpicture}[font = \small]
		
		\draw[draw=none, use as bounding box]  (-6,-2.9) rectangle (8.5,5.2);
		
		\draw [edgeColA, fill = edgeColA!2!white, thick] (-0.4,-0.4) ellipse (5cm and 2.3cm);
		\draw [decorate,
		decoration = {brace}]
		(5,1.9) --  (5,-2.7)
		node[pos=0.5,black, right=0.2cm , align=left]
		{multipede $\multipede{G}$};
		
		\vertexpair{edgeColB}{edgeColB!20!white}{name prefix = a, shift={(1,-2)}}
		\vertexpair{edgeColB}{edgeColB!20!white}{name prefix = b, shift={(-4,0.5)}}
		\vertexpair{edgeColB}{edgeColB!20!white}{name prefix = c, shift={(-2.5,-1.5)}}
		\vertexpair{edgeColB}{edgeColB!20!white}{name prefix = d, shift={(-0.5,1)}}
		\vertexpair{edgeColB}{edgeColB!20!white}{name prefix = e, shift={(3.5,0)}}
		
		\vertexpair{edgeColA}{edgeColA!20!white}{name prefix = f, shift={(-0.1,-1.25)}}
		\vertexpair{edgeColA}{edgeColA!20!white}{name prefix = g, shift={(0,-0.5)}}
		\vertexpair{edgeColA}{edgeColA!20!white}{name prefix = h, shift={(0.1,0.25)}}
		
		\draw[edgeColE!80!black, rounded corners, draw, line width = 0.4mm, line cap=round]
			($(dv0)+(220:0.04)$) -- ($(hv0)+(-0.02,0)$) -- ($(gv0)+(-0.05,0)$)
			($(dv0)+(40:0.04)$) -- ($(hv1)+(-0.01,0)$) -- ($(gv1)+(0.05,0)$)
			($(dv1)+(40:0.04)$) -- ($(hv1)+(0.09,0)$) -- ($(gv0)+(+0.05,0)$)
			($(dv1)+(220:0.04)$) -- ($(hv0)+(0.04,0)$) -- ($(gv1)+(-0.05,0)$);
			
		\draw[edgeColF, rounded corners, draw, line width = 0.4mm, line cap=round]
			($(gv0)+(-0.05,0)$) -- ($(fv0)+(-0.09,0)$) -- ($(av0)+(230:0.04)$)
			($(gv0)+(0.05,0)$) -- ($(fv1)+(-0.09,0)$) -- ($(av1)+(50:0.04)$)
			($(gv1)+(0.05,0)$) -- ($(fv1)+(0,0)$) -- ($(av0)+(50:0.04)$)
			($(gv1)+(-0.05,0)$) -- ($(fv0)+(-0.03,0)$) -- ($(av1)+(230:0.04)$);
			
		\draw[edgeColA, line width = 0.4mm, dashed,  thick]
			(av0) -- ($(av0)+(0.65,0.65)$)
			(av1) -- ($(av1)+(0.65,0.65)$)
			(bv0) -- ($(bv0)+(0.65,-0.65)$)
			(bv1) -- ($(bv1)+(0.65,-0.65)$)
			(cv0) -- ($(cv0)+(-0.65,0.65)$)
			(cv1) -- ($(cv1)+(-0.65,0.65)$)
			(cv0) -- ($(cv0)+(0.65,0.65)$)
			(cv1) -- ($(cv1)+(0.65,0.65)$)
			(ev0) -- ($(ev0)+(-0.65,-0.65)$)
			(ev1) -- ($(ev1)+(-0.65,-0.65)$)
			(ev0) -- ($(ev0)+(-0.65,+0.65)$)
			(ev1) -- ($(ev1)+(-0.65,+0.65)$)
			(dv0) -- ($(dv0)+(-0.65,-0.65)$)
			(dv1) -- ($(dv1)+(-0.65,-0.65)$)
			(hv0) -- ($(hv0)+(0.65,0.65)$)
			(hv1) -- ($(hv1)+(0.65,0.65)$)
			(fv0) -- ($(fv0)+(-0.65,-0.65)$)
			(fv1) -- ($(fv1)+(-0.65,-0.65)$)
			;

		\draw[edgeColC, fill = edgeColC!2!white, thick] (-0.4,4) ellipse (4cm and 1cm);
		\draw [decorate,
		decoration = {brace}]
		(5,5) --  (5,3)
		node[pos=0.5,black, right=0.2cm , align=left]
		{gadget vertices\\ of $\CFI{H,f}$};
		\gadget{edgeColC}{edgeColC!20!white}{name prefix = p, shift = {(-2,4)}}
		\gadget{edgeColC}{edgeColC!20!white}{name prefix = q, shift = {(1.5,4)}}
		
		\draw[edgeColC!80!black, thick]
			(pu000) -- (bv0)
			(pu000) -- (cv0)
			(pu000) -- (dv0)
			(pu011) -- (bv0)
			(pu011) -- (cv1)
			(pu011) -- (dv1)
			(pu101) -- (bv1)
			(pu101) -- (cv0)
			(pu101) -- (dv1)
			(pu110) -- (bv1)
			(pu110) -- (cv1)
			(pu110) -- (dv0);
			
		\draw[edgeColC!80!black, thick]
			(qu000) -- (dv0)
			(qu000) -- (av0)
			(qu000) -- (ev0)
			(qu011) -- (dv0)
			(qu011) -- (av1)
			(qu011) -- (ev1)
			(qu101) -- (dv1)
			(qu101) -- (av0)
			(qu101) -- (ev1)
			(qu110) -- (dv1)
			(qu110) -- (av1)
			(qu110) -- (ev0);
			
		\draw[edgeColC!80!black, thick, dashed]
			(bv0) -- ($(bv0) + (-0.5, 0.8)$)
			(bv1) -- ($(bv1) + (-0.5, 0.8)$)
			(ev0) -- ($(ev0) + (0.5, 0.8)$)
			(ev1) -- ($(ev1) + (0.5, 0.8)$);
		
	\end{tikzpicture}
	\caption[Gluing Multipedes to CFI Graphs]{Gluing multipedes to CFI graphs:
		The figure shows the gluing $\gluing{\multipede{G}}{\footbset}{\CFI{H, f}}$
		of the multipede $\multipede{G}$
		and the CFI graph $\CFI{H, f}$ at the set of segments~$\footbset$.
		The mutipede is drawn in blue,
		the segments in~$\footbset$ are drawn in red,
		and the gadget vertices of the CFI graph are drawn in green.
		Only some segments and CFI gadgets are shown.
		Two relational CFI gadgets of the multipede
		are shown, where different colors are used for each gadget.
		The edge-vertex-pairs of the CFI graphs are identified
		with the vertex-pairs of the segments in $\footbset$.
		Formally, every vertex pair of each segment has a unique color
		and there is only a single ternary relation. 
		\label{fig:gluing-multipede-cfi-graph}
	}
\end{figure}

\begin{lem}
	\label{lem:glueing-asymmetric}
	If $\multipede{G}$ is asymmetric, then $\gluing{\multipede{G}}{\footbset}{\CFI{H, f}}$
	is asymmetric.
\end{lem}
\begin{proof}
	Every automorphism of 
	$\gluing{\multipede{G}}{\footbset}{\CFI{H, f}}$
	 is in particular an automorphism of 
	$\CFI{H, f}$.
	Every non-trivial automorphism of $\CFI{H, f}$
	exchanges the two vertices of some edge-vertex-pairs because $H$ is totally ordered.
	Because every edge-vertex-pair is identified with a segment of $\multipede{G}$
	and $\multipede{G}$ is asymmetric,
	$\gluing{\multipede{G}}{\footbset}{\CFI{H, f}}$ is asymmetric, too.
\end{proof}

We now show that
$\gluing{\multipede{G}}{\footbset}{\CFI{H, 0}} \kequiv{k} \gluing{\multipede{G}}{\footbset}{\CFI{H, 1}}$
if $G$ and $H$ satisfy certain conditions.
Let $\tup{\vertA}$ be a vertex-tuple of $\gluing{\multipede{G}}{\footbset}{\CFI{H, f}}$,
i.e., $\tup{\vertA}$ contains either gadget vertices of the gadgets in $\CFI{H,f}$
or feet of $\multipede{G}$. We also write $\segmentOf{\tup{\vertA}}$
for the set of segments of all feet in $\tup{\vertA}$.
\begin{enumerate}
	\item The segments $\segmentOf{\tup{\vertA}}$ are \defining{directly-fixed by $\tup{\vertA}$}.
	\item The segments $\closure{\segmentOf{\tup{\vertA}}} \setminus \segmentOf{\tup{\vertA}}$ are \defining{closure-fixed by $\tup{\vertA}$}.
	\item A segment $\bVertA \in \footbsetA$ is \defining {gadget-fixed by $\tup{\vertA}$}
	if the feet of $\bVertA$ are identified with some edge-vertex-pair with origin $\set{\bVertB,\bVertC}$ in $\CFI{H,f}$
	such that there is a gadget vertex with origin~$\bVertB$ or~$\bVertC$ in $\tup{\vertA}$.
	\item A segment is \defining{fixed by $\tup{\vertA}$} if it is directly fixed, closure-fixed, or gadget-fixed by $\tup{\vertA}$.
\end{enumerate}

\begin{lem}
	\label{lem:number-fixed-segments}
	Let $r \geq k \geq 2$
	and $\tup{\vertA}$ be a vertex-tuple of $\gluing{\multipede{G}}{\footbset}{\CFI{H, f}}$ of length at most $k$.
	If $H$ is $r$-regular, 
	$G$ is $2k$-meager,
	and $\footbset$ is $6k$-scattered,
	then at most $r\cdot \tlength{\tup{\vertA}}$ segments are fixed by $\tup{\vertA}$.
	If $\tup{\vertA}$ contains $i$ gadget vertices and $\ell$ segments in $\footbset$ are directly-fixed by $\tup{\vertA}$,
	then at most $\tlength{\tup{\vertA}}-i-\ell$ segments in $\footbset$ are closure-fixed by $\tup{\vertA}$.
\end{lem}
\begin{proof}
	Assume that~$\tup{\vertA}$ contains~$j$ feet and~$i$ gadget vertices.
	So, $i + j = \tlength{\tup{\vertA}} \leq k$.
	Then at most~$r i$ segments are gadget-fixed by~$\tup{\vertA}$ because~$G$ is $r$-regular.
	From Lemma~\ref{lem:closure-meager-size}
	it follows that $|\closure{\segmentOf{\tup{\vertA}}}| \leq 2j$
	because $G$ is $k$-meager.
	So at most $2j$ segments are directly-fixed or closure-fixed.
	Then $r i + 2j \leq r i + r j = r(i + j) = r \tlength{\tup{\vertA}}$
	because $i+j = \tlength{\tup{\vertA}}$ and $r\geq k\geq 2$.
	
	By Lemma~\ref{lem:closure-new-vertices},
	it holds that $|\closure{\segmentOf{\tup{\vertA}}} \cap (\footbset \setminus \segmentOf{\tup{\vertA}})| \leq |\segmentOf{\tup{\vertA}} \setminus \footbset|$ because~$G$ is $k$\nobreakdash-meager,~$\footbset$ is $6k$\nobreakdash-scattered,
	and $|\segmentOf{\tup{\vertA}}|\leq k$.
	That is, the number of closure-fixed segments in~$\footbset$
	is bounded by the number of directly-fixed segments not in~$\footbset$.
	The number of directly-fixed segments not in~$\footbset$ is $\tlength{\tup{\vertA}}-i-\ell$.
\end{proof}

We now combine winning strategies of Duplicator in the bijective pebble game on multipedes and CFI graphs:
\begin{lem}
	\label{lem:gluing-k-equiv}
	Let $r \geq k \geq 3$, $G$ be $2rk$-meager,
	$H$ be $r$-regular and at least $(k+2)$-connected,
	and $\footbset$ be $6rk$-scattered.
	Then $\gluing{\multipede{G}}{\footbset}{\CFI{H,0}} \kequiv{k} \gluing{\multipede{G}}{\footbset}{\CFI{H,1}}$.
\end{lem}
\begin{proof}
	Let $\StructA = \multipede{G}$,
	$\StructB = \CFI{H,0}$, and
	$\StructB' = \CFI{H,1}$.
	We show that Duplicator has a winning strategy in the bijective $k$-pebble game played on $\gluing{\StructA}{\footbset}{\StructB}$ and $\gluing{\StructA}{\footbset}{\StructB'}$.
	
	For a set of segments~$\footbsetB$ and a tuple~$\tup{\vertA}$,
	we denote by~$\feetrestrict{\tup{\vertA}}{\footbsetB}$ the restriction of~$\tup{\vertA}$ to all feet whose segment is contained in~$\footbsetB$,
	by~$\gadgetrestrict{\tup{\vertA}}$ the restriction of~$\tup{\vertA}$ to all gadget vertices,
	and by~$\nongadgetrestrict{\tup{\vertA}}$
	the restriction of~$\tup{\vertA}$ to all feet.
	
	Duplicator maintains the following invariant.
	At every position $(\gluing{\StructA}{\footbset}{\StructB},\tup{\vertA} ; \gluing{\StructA}{\footbset}{\StructB'}, \tup{\vertA}')$ in the game
	there exist tuples of feet $\gadgetfixed{\tup{\vertB}}, \closurefixed{\tup{\vertB}}$
	of $\gluing{\StructA}{\footbset}{\StructB}$
	and $\gadgetfixed{\tup{\vertB}'},\closurefixed{\tup{\vertB}'}$
	of $\gluing{\StructA}{\footbset}{\StructB'}$ satisfying the following:
	\begin{enumerate}
		\item \label{itm:gadget-fixed-vertices}
		For every segment  gadget-fixed by $\tup{\vertA}$
		(and so by $\tup{\vertA}'$),
		there is exactly one foot of this segment contained in $\gadgetfixed{\tup{\vertB}}$ and in $\gadgetfixed{\tup{\vertB}'}$, respectively.
		No foot of another segment is contained in~$\gadgetfixed{\tup{\vertB}}$ or in~$\gadgetfixed{\tup{\vertB}'}$.
		\item \label{itm:closure-fixed-vertices}
		For every segment contained in $\footbsetA$ and closure-fixed by $\tup{\vertA}$ (and so by $\tup{\vertA}'$),
		there is exactly one foot of this segment contained in $\closurefixed{\tup{\vertB}}$ and in $\closurefixed{\tup{\vertB}'}$, respectively.
		No foot of another segment is contained in~$\closurefixed{\tup{\vertB}}$ or in~$\closurefixed{\tup{\vertB}'}$.
		\item \label{itm:inv-multipede} There is an automorphism $\autoA \in \autGroup{\footinduced{\StructA}{\segmentOf{\nongadgetrestrict{\tup{\vertA}}\gadgetfixed{\tup{\vertB}}\closurefixed{\tup{\vertB}}}}}$ satisfying
		$\autoA(\nongadgetrestrict{\tup{\vertA}}\gadgetfixed{\tup{\vertB}}\closurefixed{\tup{\vertB}}) = \nongadgetrestrict{\tup{\vertA}'}\gadgetfixed{\tup{\vertB}'}\closurefixed{\tup{\vertB}'}$.
		\item \label{itm:inv-cfi} $(\StructB, \feetrestrict{\tup{\vertA}}{\footbsetA}\gadgetrestrict{\tup{\vertA}}\closurefixed{\tup{\vertB}}) \kequiv{k} (\StructB', \feetrestrict{\tup{\vertA}'}{\footbsetA}\gadgetrestrict{\tup{\vertA}'}\closurefixed{\tup{\vertB}'})$.
		\item \label{itm:inv-cfi-gadget-fixed} For every base vertex $\bVertA \in V^H$,
		it holds that
		$(\StructB', \gadgetrestrict{\tup{\vertA}}\gadgetfixed{\tup{\vertB}})[V_\bVertA] \iso (\StructB, \gadgetrestrict{\tup{\vertA}'}\gadgetfixed{\tup{\vertB}'})[V_\bVertA]$,
		where~$V_\bVertA$ is the set of all gadget vertices with origin $\bVertA$ and all edge vertices with origin $\set{\bVertA,\bVertB}$ for some $\bVertB \in \neighbors{G}{\bVertA}$.
	\end{enumerate}
	Regarding Property~\ref{itm:inv-multipede}, we have
	${\segmentOf{\nongadgetrestrict{\tup{\vertA}}\gadgetfixed{\tup{\vertB}}\closurefixed{\tup{\vertB}}} = \segmentOf{\nongadgetrestrict{\tup{\vertA}'}\gadgetfixed{\tup{\vertB}'}\closurefixed{\tup{\vertB}'}}}$ and
	${\tlength{\nongadgetrestrict{\tup{\vertA}}\gadgetfixed{\tup{\vertB}}\closurefixed{\tup{\vertB}}}=\tlength{\nongadgetrestrict{\tup{\vertA}'}\gadgetfixed{\tup{\vertB}'}\closurefixed{\tup{\vertB}'}} \leq rk}$~by Lemma~\ref{lem:number-fixed-segments},
	so the local automorphism $\auto$ extends to the closure of $\segmentOf{\nongadgetrestrict{\tup{\vertA}}\gadgetfixed{\tup{\vertB}}}$
	by Lemma~\ref{lem:closure-extension}
	because $G$ is $rk$-meager.
	Regarding Property~\ref{itm:inv-cfi}, note that ${\tlength{\feetrestrict{\tup{\vertA}}{\footbsetA}\gadgetrestrict{\tup{\vertA}}\closurefixed{\tup{\vertB}}} = \tlength{\feetrestrict{\tup{\vertA}'}{\footbsetA}\gadgetrestrict{\tup{\vertA}'}\closurefixed{\tup{\vertB}'}} \leq k}$:
	By Lemma~\ref{lem:number-fixed-segments},
	the number of closure-fixed segments in $\footbset$ is at most
	${\tlength{\closurefixed{\tup{\vertB}}} \leq k - \tlength{\gadgetrestrict{\tup{\vertA}}} - \tlength{\feetrestrict{\tup{\vertA}}{\footbset}}}$.
	Hence, ${\tlength{\feetrestrict{\tup{\vertA}}{\footbsetA}\gadgetrestrict{\tup{\vertA}}\closurefixed{\tup{\vertB}}}
	= \tlength{\feetrestrict{\tup{\vertA}}{\footbsetA}} +\tlength{\gadgetrestrict{\tup{\vertA}}}
	+\tlength{\closurefixed{\tup{\vertB}}} \leq k}$.
	(Note that ${\tlength{\feetrestrict{\tup{\vertA}}{\footbsetA}}=\tlength{\feetrestrict{\tup{\vertA}'
		}{\footbsetA}}}$, ${\tlength{\gadgetrestrict{\tup{\vertA}}} = \tlength{\gadgetrestrict{\tup{\vertA}'}}}$, etc.,~because otherwise Duplicator has already lost the game).
	Property~\ref{itm:inv-cfi-gadget-fixed} is needed because 
	$\tlength{\gadgetfixed{\tup{\vertB}}} \leq r k$ (possibly with equality), which exceeds $k$.
	Thus, Property~\ref{itm:inv-cfi-gadget-fixed} cannot be implied by Property~\ref{itm:inv-cfi}.
	Property~\ref{itm:inv-cfi-gadget-fixed} guarantees that we pick
	the vertices~$\gadgetfixed{\tup{\vertB}}$ and~$\gadgetfixed{\tup{\vertB}'}$ consistently.
	
	Intuitively, we want to play two games.
	Game~I is played with $rk$ pebbles on the multipede at position
	$(\StructA, \nongadgetrestrict{\tup{\vertA}}\gadgetfixed{\tup{\vertB}}\closurefixed{\tup{\vertB}}; \StructA, \nongadgetrestrict{\tup{\vertA}'}\gadgetfixed{\tup{\vertB}'}\closurefixed{\tup{\vertB}'})$.
	Game~II is played with $k$ pebbles on the CFI graphs at position
	$(\StructB,\feetrestrict{\tup{\vertA}}{\footbsetA}\gadgetrestrict{\tup{\vertA}}\closurefixed{\tup{\vertB}}; \StructB', \feetrestrict{\tup{\vertA}'}{\footbsetA}\gadgetrestrict{\tup{\vertA}'}\closurefixed{\tup{\vertB}'})$.
	We use the winning strategy of Duplicator in both games
	(Lemmas~\ref{lem:k-connected-treewidth-CFI-equiv} and~\ref{lem:multipede-winning-strategy}) to construct
	a winning strategy in the bijective $k$-pebble game played at position 
	$(\gluing{\StructA}{\footbset}{\StructB}, \tup{\vertA}; \gluing{\StructA}{\footbset}{\StructB'}, \tup{\vertA}')$.
	This is possible because in Game~I we artificially fixed all gadget-fixed segments
	and in Game~II we artificially fixed all closure-fixed segments in $\footbsetA$ (and only the segments in $\footbsetA$ are identified with edge-vertex-pairs of the CFI graphs).

	Now assume that it is Spoiler's turn.
	When Spoiler picks up a pair of pebbles $(p_i,q_i)$ from the structures,
	we first update the tuples $\gadgetfixed{\tup{\vertB}}$, $\gadgetfixed{\tup{\vertB}'}$, $\closurefixed{\tup{\vertB}}$, and $\closurefixed{\tup{\vertB}'}$:
	If a segment is no longer gadget-fixed or closure-fixed,
	then we remove the corresponding entries in the corresponding tuples.
	Clearly the invariant is maintained.
	
	We describe how Duplicator defines a bijection $\gamebij$ between
	$\gluing{\StructA}{\footbsetA}{\StructB}$ and $\gluing{\StructA}{\footbsetA}{\StructB'}$
	by defining $\gamebij(\vertC)$ using the following case distinction:
	\begin{enumerate}[label=(\alph*)]
		\item Assume the vertex $\vertC$ is a foot whose segment is contained in $\closure{\segmentOf{\nongadgetrestrict{\tup{\vertA}}\gadgetfixed{\tup{\vertB}}\closurefixed{\tup{\vertB}}}}$.
		The automorphism $\autoA$ from Property~\ref{itm:inv-multipede} extends to an automorphism of $\footinduced{\StructA}{\closure{\segmentOf{\nongadgetrestrict{\tup{\vertA}}\gadgetfixed{\tup{\vertB}}\closurefixed{\tup{\vertB}}}}}$
		by Lemma~\ref{lem:closure-extension}.
		We set $\gamebij(\vertC) := \autoA(\vertC)$.
		(This is actually Duplicator's winning strategy in Game~I~\cite{GurevichShelah96},
		but the exact strategy is needed later).
		 
		\item Assume $\vertC$ is a foot not covered by the previous case. 
		The bijection given by Duplicator's winning strategy in Game~I
		defines $\gamebij(\vertC)$
		(actually,~$\gamebij(\vertC)$ is an arbitrary foot of the same segment as~$\vertC$).
		
		\item Finally, assume $\vertC$ is a gadget vertex.
		We use the bijection given by Duplicator's winning strategy in Game~II
		to define $\gamebij(\vertC)$.
	\end{enumerate}
	Now Spoiler places the pebble pair $(p_i,q_i)$ on the vertices $\vertC$ and  $\vertC' := \gamebij(\vertC)$.
	We update the tuples~$\gadgetfixed{\tup{\vertB}}$,~$\gadgetfixed{\tup{\vertB}'}$,~$\closurefixed{\tup{\vertB}}$, and~$\closurefixed{\tup{\vertB}'}$ as follows:		
	\begin{enumerate}[label=(\alph*)]	
		\item Assume~$\vertC$ (and thus also~$\vertC'$) is a gadget vertex.
		Property~\ref{itm:inv-cfi} clearly holds
		because we followed Duplicator's winning strategy in Game II.
		No new segments in~$\footbset$ get closure-fixed by $\tup{\vertA}\vertC$ or $\tup{\vertA}'\vertC'$,
		so we just do not change~$\closurefixed{\tup{\vertB}}$ and $\closurefixed{\tup{\vertB}}'$
		and Property~\ref{itm:closure-fixed-vertices} is satisfied.
		We satisfy Properties~\ref{itm:gadget-fixed-vertices},~\ref{itm:inv-multipede},
		and~\ref{itm:inv-cfi-gadget-fixed}
		by picking feet from the new gadget-fixed segments as follows:
		\begin{itemize}

		\item Assume that a segment becomes gadget-fixed by $\tup{\vertA}\vertC$ and $\tup{\vertA}'\vertC'$
		(and so the segment is in~$\footbset$)
		that is already closure-fixed by~$\tup{\vertA}$ and~$\tup{\vertA}'$, respectively.
		We pick the same feet as in~$\closurefixed{\tup{\vertB}}$ and~$\closurefixed{\tup{\vertB}'}$
		and append them to~$\gadgetfixed{\tup{\vertB}'}$ and~$\gadgetfixed{\tup{\vertB}'}$, respectively.
		Thus, Property~\ref{itm:gadget-fixed-vertices} holds.
		Because the local automorphism from Property~\ref{itm:inv-multipede}
		already maps~$\closurefixed{\tup{\vertB}}$ to~$\closurefixed{\tup{\vertB}'}$,
		appending the corresponding entries to~$\gadgetfixed{\tup{\vertB}'}$ and~$\gadgetfixed{\tup{\vertB}'}$ satisfies Property~\ref{itm:inv-multipede}.
		Because the closure-fixed segments are part of
		the pebbled vertices in Game~II (Property~\ref{itm:inv-cfi}),
		appending these feet to $\gadgetfixed{\tup{\vertB}'}$ and $\gadgetfixed{\tup{\vertB}'}$
		satisfies Property~\ref{itm:inv-cfi-gadget-fixed}.
		
		\item Assume that a segment~$\bVertA$ becomes gadget-fixed by $\tup{\vertA}\vertC$ and $\tup{\vertA}'\vertC'$
		that is not closure-fixed by~$\tup{\vertA}$ and~$\tup{\vertA}'$, respectively.
		For both~$\StructB$ and~$\StructB'$,
		we pick the unique foot~$\vertB$ respectively~$\vertB'$
		of the segment~$\bVertA$
		adjacent to the newly pebbled gadget vertex
		and append them
		to~$\gadgetfixed{\tup{\vertB}}$ and~$\gadgetfixed{\tup{\vertB}'}$, respectively.
		Hence, Property~\ref{itm:gadget-fixed-vertices} is satisfied.
		Property~\ref{itm:inv-multipede} is satisfied by Lemma~\ref{lem:distance-set-local-aut}:
		We can pick for non-closure-fixed segments arbitrary feet
		and still find a local automorphism mapping them onto each other.
		For the sake of contradiction,
		assume that Property~\ref{itm:inv-cfi-gadget-fixed} is not satisfied
		by appending~$\vertB$ and~$\vertB'$.
		Then there is a base vertex $\bVertB \in V^H$
		such that $(\StructB', \gadgetrestrict{\tup{\vertA}}\gadgetfixed{\tup{\vertB}}\vertB)[V_\bVertB] \not\iso (\StructB, \gadgetrestrict{\tup{\vertA}'}\gadgetfixed{\tup{\vertB}'}\vertB')[V_\bVertB]$.
		There must be a pebble placed on a gadget vertex with origin~$\bVertB$,
		i.e., both~$\gadgetrestrict{\tup{\vertA}}$ and~$\gadgetrestrict{\tup{\vertA}'}$
		are nonempty,
		because otherwise $|\gadgetfixed{\tup{\vertB}}|\leq k$ and $|\gadgetfixed{\tup{\vertA}'}| \leq k$
		and thus the two edge vertices of the segment~$\bVertA$ form an orbit by Lemma~\ref{lem:orbits-connectivity} since $H$ is $(k+2)$-connected
		(the lemma also holds in the setting of a single edge-vertex-pair per base edge).
		In particular,
		$(\StructB', \gadgetrestrict{\tup{\vertA}}\vertB)[V_\bVertB] \not\iso (\StructB, \gadgetrestrict{\tup{\vertA}'}\vertB')[V_\bVertB]$.
		Note that~$\vertB$ and~$\vertB'$ are $\Ck{3}$-definable
		because they are the unique vertices in the segment~$\bVertA$
		adjacent to~$\vertC$ and~$\vertC'$, respectively.
		So Spoiler can win Game~II
		by picking up a pebble pair (which is neither placed on~$\vertC$ nor the gadget of~$\bVertB$)
		and placing it on~$\vertB$.
		Because $k \geq 3$, such a pebble pair actually exists.
		But that contradicts Property~\ref{itm:inv-cfi}
		and hence Property~\ref{itm:inv-cfi-gadget-fixed} is satisfied.
		\end{itemize}

		\item Assume~$\vertC$ (and thus also~$\vertC'$) is a foot.
		Thus, no segments get gadget-fixed 
		by~$\tup{\vertA}\vertC$ and~$\tup{\vertA}'\vertC'$
		that were not already gadget-fixed by~$\tup{\vertA}$ and~$\tup{\vertA}'$, respectively.
		So without picking further feet,~$\gadgetfixed{\tup{\vertB}}$
		and~$\gadgetfixed{\tup{\vertB}}'$ satisfy
		Properties~\ref{itm:gadget-fixed-vertices} and~\ref{itm:inv-cfi-gadget-fixed}.
		Possibly a new segment $\bVertA \in \footbsetA$ becomes closure-fixed.
		By Lemma~\ref{lem:number-fixed-segments},
		there can be at most one such segment.
		We satisfy Properties~\ref{itm:closure-fixed-vertices} to~\ref{itm:inv-cfi}
		as follows:
		\begin{itemize}
		\item Assume that~$\bVertA$ is already gadget-fixed by~$\tup{\vertA}$ and~$\tup{\vertA}'$.
		We append the vertices~$\vertB$ and~$\vertB'$
		whose segment is~$\bVertA$ in~$\gadgetfixed{\tup{\vertB}}$
		and in~$\gadgetfixed{\tup{\vertB}'}$, respectively,
		to~$\closurefixed{\tup{\vertB}}$ and~$\closurefixed{\tup{\vertB}'}$,
		respectively.
		So Property~\ref{itm:closure-fixed-vertices} holds.
		Property~\ref{itm:inv-multipede} is satisfied because
		the local automorphism~$\autoA$ already maps~$\gadgetfixed{\tup{\vertB}}$ to~$\gadgetfixed{\tup{\vertB}'}$
		and hence appending~$\vertB$ respectively~$\vertB'$ satisfies Property~\ref{itm:closure-fixed-vertices}.
		Property~\ref{itm:inv-cfi} is satisfied because of 
		Property~\ref{itm:inv-cfi-gadget-fixed}:
		Fixing a single gadget vertex fixes all edge vertices
		adjacent to the gadget (which is $\Ck{3}$-definable)
		and by Property~\ref{itm:inv-cfi-gadget-fixed}
		we have chosen~$\gadgetfixed{\tup{\vertB}}$
		respectively~$\gadgetfixed{\tup{\vertB}'}$
		consistently with the pebbles on the gadgets.
		So we can actually place a pebble pair~$\vertB$ and~$\vertB'$
		and Property~\ref{itm:inv-cfi} holds.
		
		\item Otherwise,~$\bVertA$ is not gadget-fixed by~$\tup{\vertA}$ and~$\tup{\vertA}'$.
		So the two feet of~$\bVertA$ form an orbit in
		$(\StructB, \feetrestrict{\tup{\vertA}}{\footbsetA}\gadgetrestrict{\tup{\vertA}}\closurefixed{\tup{\vertB}})$
		and likewise in $(\StructB, \feetrestrict{\tup{\vertA}'}{\footbsetA}\gadgetrestrict{\tup{\vertA}'}\closurefixed{\tup{\vertB}'})$
		by Lemma~\ref{lem:orbits-connectivity}
		because~$H$ is $(k+2)$-connected.
		So for every choice of feet in~$\bVertA$,
		Property~\ref{itm:inv-cfi} is satisfied.
		We make an arbitrary choice of a foot~$\vertB$ of the segment~$\bVertA$ in~$\StructB$,
		which we append to~$\closurefixed{\tup{\vertB}}$.
		Because~$\bVertA$ is closure-fixed by $\tup{\vertA}\vertC$,
		we can extend~$\autoA$
		also to the segment~$\bVertA$.
		We pick $\vertB' = \autoA(\vertA)$
		and append it to~$\closurefixed{\tup{\vertB}'}$.
		So Property~\ref{itm:inv-multipede} is satisfied.
		\end{itemize}
	\end{enumerate}
	Hence, Duplicator is able to maintain the invariant.
	We update~$\tup{\vertA}$ and~$\tup{\vertA}'$
	to include~$\vertC$ respectively~$\vertC'$.
	We show that the pebbles induce a local isomorphism:
	By Property~\ref{itm:inv-cfi},
	the pebbles induce a local isomorphism
	$\feetrestrict{\tup{\vertA}}{\footbsetA}\gadgetrestrict{\tup{\vertA}}\closurefixed{\tup{\vertB}} \mapsto \feetrestrict{\tup{\vertA}'}{\footbsetA}\gadgetrestrict{\tup{\vertA}'}\closurefixed{\tup{\vertB}'}$  on the CFI graphs.
	This local isomorphism extends to all gadget-fixed segments
	by Property~\ref{itm:inv-cfi-gadget-fixed},
	that is, the map $\feetrestrict{\tup{\vertA}}{\footbsetA}\gadgetrestrict{\tup{\vertA}}\closurefixed{\tup{\vertB}}\gadgetfixed{\tup{\vertB}} \mapsto \feetrestrict{\tup{\vertA}'}{\footbsetA}\gadgetrestrict{\tup{\vertA}'}\closurefixed{\tup{\vertB}'}\gadgetfixed{\tup{\vertB}'}$
	is a local isomorphism of the CFI graphs.
	By Property~\ref{itm:inv-multipede},
	the map
	$\nongadgetrestrict{\tup{\vertA}}\gadgetfixed{\tup{\vertB}}\closurefixed{\tup{\vertB}} \mapsto \nongadgetrestrict{\tup{\vertA}'}\gadgetfixed{\tup{\vertB}'}\closurefixed{\tup{\vertB}'}$
	is a local automorphism on the multipede.
	Because the maps agree on $\gadgetfixed{\tup{\vertB}}\closurefixed{\tup{\vertB}}$,
	the combined map 
	$\feetrestrict{\tup{\vertA}}{\footbsetA}\gadgetrestrict{\tup{\vertA}}\closurefixed{\tup{\vertB}}\gadgetfixed{\tup{\vertB}}\nongadgetrestrict{\tup{\vertA}} \mapsto \feetrestrict{\tup{\vertA}'}{\footbsetA}\gadgetrestrict{\tup{\vertA}'}\closurefixed{\tup{\vertB}'}\gadgetfixed{\tup{\vertB}'}\nongadgetrestrict{\tup{\vertA}}'$
	is a local isomorphism of the gluings.
	In particular, the map 	$\feetrestrict{\tup{\vertA}}{\footbsetA}\gadgetrestrict{\tup{\vertA}}\nongadgetrestrict{\tup{\vertA}} \mapsto \feetrestrict{\tup{\vertA}'}{\footbsetA}\gadgetrestrict{\tup{\vertA}'}\nongadgetrestrict{\tup{\vertA}}'$
	is a local isomorphism of the gluings,
	but that is just the map $\tup{\vertA} \mapsto \tup{\vertA}'$.
	So 	Duplicator does not lose in this round
	and by induction wins the bijective $k$-pebble game.
\end{proof}

\begin{thm}
	\label{thm:asymmetric-structures-with-CFI graphs}
	There is an \FO{}-interpretation $\interpret$
	and, for every $k \in \nat$, a pair of ternary $\set{R,\spleq}$-structures $(\StructA_k, \StructB_k)$
	such that
	\begin{enumerate}[label=(\alph*)]
		\item $\spleq$ is a total preorder on $\StructA$ and $\StructB$,
		\item $\StructA_k$, $\StructB_k$, $\reduct{\StructA_k}{R}$, and $\reduct{\StructB_k}{R}$ are asymmetric,
		\item $\StructA_k \kequiv{k} \StructB_k$,
		\item $\StructA_k \not\iso \StructB_k$, and
		\item $\interpret(\StructA_k)$ and $\interpret(\StructB_k)$
		are non-isomorphic CFI graphs over the same ordered base graph.
	\end{enumerate}
	The interpretation~$\interpret$ is one-dimensional and equivalence-free.
\end{thm}
\begin{proof}
	
	Let $k \geq 2$ be arbitrary and~$H$ be a clique of size $k+4$.
	Thus, $H$ is $r:=(k+3)$-regular, $(k+2)$-connected, and has $m:=\frac{1}{2}(k+4)(k+3)\leq r(k+2)$ edges.
	There exists an odd and $(6r(k+2))$-meager bipartite graph $G=(\gadgetsb,\feetb, E)$
	that contains a $6r(k+2)$-scattered set $\footbset \subseteq \feetb$ of size $m \leq (6r(k+2))^2$ by Lemma~\ref{lem:odd-meager-distance-base-graph}.
	Equip the graphs~$G$ and~$H$ with arbitrary total orders.
	
	Set $\StructA_k := \gluing{\multipede{G}}{\footbset}{\CFI{H,0}}$
	and $\StructB_k := \gluing{\multipede{G}}{\footbset}{\CFI{H,1}}$.
	Clearly, $\StructA_k \not\iso \StructB_k$
	and $\StructA_k \kequiv{k} \StructB_k$ by Lemma~\ref{lem:gluing-k-equiv}
	because $H$ is $(k+2)$-connected and $r$-regular, $G$ is
	 $(6r(k+2))$-meager and so in particular $2rk$-meager,
	and $\footbset$ is $6r(k+2)$-scattered and so in particular $6rk$-scattered.
	By Lemma~\ref{lem:odd-multipede-asymmetric},
	the multipede~$\multipede{G}$ is asymmetric because $G$ is odd.
	Thus,~$\StructA_k$ and~$\StructB_k$ are asymmetric by Lemma~\ref{lem:glueing-asymmetric}.
	
	We define the interpretation~$\interpret$,
	which maps~$\StructA_k$ to $\CFI{H,0}$ and~$\StructB_k$ to $\CFI{H,1}$.  
	Recall that the gluing extends edges of the CFI graphs to triples by repeating the last entry and multipedes do not contain such triples.
	So we can easily define the vertices contained in the ``CFI triples''.
	By taking the induced graph on these vertices
	and by shortening the triples back to pairs,
	one defines the CFI graphs again.
	This is done by the following one-dimensional and equivalence-free $\FO[\set{R,\spleq},\set{E,\spleq}]$-interpretation $\interpret = (\formA_{\text{dom}}(x), \formB_E(x,y), \formB_{\spleq}(x,y))$:
	\begin{align*}
		\formA_{\text{dom}}(x) &:= \exists y.\qspace R(x,y,y)\lor R(y,x,x), \\
		\formB_E(x,y) &:= R(x,y,y),\\
		\formB_{\spleq}(x,y) & := x \spleq y.
	\end{align*}
	Clearly
	$\interpret(\StructA_k) = \CFI{H,0}$ and 
	$\interpret(\StructB_k) = \CFI{H,1}$.
	
	So far we have satisfied all claims of the lemma apart from that the $R$-reducts are asymetric as well.
	To achieve this, we modify the structures and additionally encode~$\spleq$ into~$R$.
	Let $\StructA'_k$ and $\StructB'_k$ be the structures
	obtained form $\StructA_k$ and $\StructB_k$,
	where a directed path of length~$i$ and another one of length~$i+1$ is added to each vertex in the $i$-th color class.
	In that way,
	we obtain that also the $R$-reducts $\reduct{\StructA'_k}{R}$ and $\reduct{\StructB'_k}{R}$ are asymmetric.
	
	Clearly, there is a one-dimensional and equivalence-free \FO{}-interpretation~$\interpret'$
	such that $\interpret'(\StructA'_k) = \StructA_k$
	and $\interpret'(\StructB'_k) = \StructB_k$ for all $k$:
	remove all vertices of out-degree at most~$1$
	because we attached to all original vertices two directed paths.
	Then the $\StructA'_k$ and $\StructB'_k$ together with the interpretation
	$\interpret \circ \interpret'$ satisfy the claim of the lemma.
	Note that $\interpret \circ \interpret'$ is again a one-dimensional
	equivalence-free interpretation.
\end{proof}

\noindent We now prove Theorem~\ref{thm:wsc-le-wsci}
and separate \IFPCWSC{} from \IFPCWSCI{}.

\begin{proof}[Proof of Theorem~\ref{thm:wsc-le-wsci}]
	The first assertion $\IFPC{} < \IFPCWSC{}$ is Corollary~\ref{cor:ifpc-le-wsc}.
	We now show $\IFPCWSC < \IFPCWSCI{}$.
	Consider the class of $\set{R,\spleq}$-structures $\GraphClass$ given by Theorem~\ref{thm:asymmetric-structures-with-CFI graphs}.
	
	We argue that $\IFPCWSC{}=\IFPC$ on $\GraphClass'$.
	Because the structures are asymmetric, there are only singleton orbits.
	Hence, choice formulas have to define singleton choice-sets. Otherwise, it is not possible to witness them as orbits.
	Hence, choosing becomes useless and can be simulated by (non-WSC)-fixed-point operators: If the choice-set is indeed a singleton,
	its unique member is definable.
	Otherwise evaluate to false (because the WSC-fixed point operator evaluates to false if the choices were not witnessed successfully).

	Let $\interpret$ be the \FO-interpretation
	extracting the CFI graphs from $\GraphClass$-structures provided by Theorem~\ref{thm:asymmetric-structures-with-CFI graphs}.
	For every $k\in \nat$, there are two non-isomorphic structures
	$\StructA_k \kequiv{k} \StructB_k$ in~$\GraphClass$
	such that $\interpret$ evaluates on $\StructA_k$ and $\StructB_k$ to an even and an odd CFI graph, respectively.
	Consequently, we call $\Struct_k$ even and $\StructB_k$ odd.
	Since $\StructA_k \kequiv{k} \StructB_k$ for all $k\in\nat$,
	\IFPC{} does not distinguish even from odd $\GraphClass$\nobreakdash-structures.
	We show  that \IFPCWSCI{} distinguishes odd and even $\GraphClass$-structures.
	The CFI-query on ordered base graphs is definable in \IFPCWSCI{} by Theorem~\ref{thm:canonize-CFI if-base}.
	By Corollary~\ref{cor:canonize-CFI if-base-wsci},
	there is actually a $\WSCIof{\IFPC{}}=\WSCof{\IFPC}$-formula~$\formA_{\text{CFI}}$
	defining the CFI query on ordered base graphs.
	Then the $\Iof{\WSCof{\IFPC}}$-formula
	$\iop{\interpret}{\formA_{\text{CFI}}}$
	distinguishes $\StructA_k$ and $\StructB_k$ for all $k$: the interpretation $\interpret$ reduces by Theorem~\ref{thm:asymmetric-structures-with-CFI graphs}
	the problem of distinguishing even and odd $\GraphClass$~\nobreakdash structures to the  CFI query over ordered base graphs,
	which is defined by~$\formA_{\text{CFI}}$.
\end{proof}

\begin{cor}
	$\IFPCWSC{} < \PTime{}$.
\end{cor}

\begin{cor}
	$\WSCof{\IFPC} < \Iof{\WSCof{\IFPC{}}}$.
\end{cor}
\noindent Note that the prior corollary refines Corollary~\ref{cor:wsci-le-wsci2},
which states that \[\WSCof{\IFPC} = \WSCIof{\IFPC{}} < \WSCIkof{2}{\IFPC{}} = \WSCof{\Iof{\WSCof{\IFPC{}}}}.\]
We actually expect that $\WSCof{\IFPC} < \Iof{\WSCof{\IFPC{}}} < \WSCof{\Iof{\WSCof{\IFPC{}}}}$
because it seems unlikely that $\Iof{\WSCof{\IFPC{}}}$
defines the CFI query of the base graphs of Theorem~\ref{thm:cfi-wsci-wsci}.

Again, the inability of $\IFPCWSC{}$ to distinguish these structures
is not related to Ebbinghaus' reduct property.
$\IFPCWSC{}$ with global reduct semantics fails to distinguish them, either:
\begin{lem}
 $\IFPCWSCglob < \IFPCWSCIglob$ (recall that these are the global reduct semantic variants of $\IFPCWSC$ and $\IFPCWSCI$).
\end{lem}
\begin{proof}
	We consider the same class of structures $\GraphClass$
	as in the proof of~Theorem~\ref{thm:wsc-le-wsci}.
	Since their signature is fixed,
	$\IFPCWSCI \leq \IFPCWSCIglob{}$ on $\GraphClass$ by Lemma~\ref{lem:non-reduct-in-global-reduct-semantics-for-fixed-signature}.
	Thus, \IFPCWSCIglob{} distinguishes even and odd structures.
	To show that $\IFPCWSCglob$ does not distinguish them,
	we show that $\IFPCWSCglob = \IFPC$ on $\GraphClass$.
	Then~Theorem~\ref{thm:wsc-le-wsci} implies the claim.
	
	If an \IFPCWSC{}-formula uses~$R$,
	then, under the global reduct semantics, the formula is still evaluated on an asymmetric structure because the $R$-reducts are asymmetric.
	Hence, choosing is useless can be simulated by (non-WSC)-fixed-point operators, again.
	If otherwise the formula does not use~$R$,
	then the $\spleq$-reduct 
	is fully determined by the number and sizes of the color classes.
	They must be equal for~$\StructA_k$ and~$\StructB_k$ for every $k \geq 3$ because otherwise~$\Ck{3}$ distinguishes~$\StructA_k$ and~$\StructB_k$.
	The number of color classes and their sizes are clearly \IFPC{}-definable.
\end{proof}

\begin{cor}
	The logics \IFPCWSC{} and \IFPCWSCglob{} are not closed under \IFPC{}-interpretations
	and not even under one-dimensional and equivalence-free \FO{}-interpretations.
\end{cor}
\begin{proof}
	The FO-interpretation $\interpret$ in the proof of Theorem~\ref{thm:wsc-le-wsci} is one-dimensional and equivalence-free because it is obtained from Theorem~\ref{thm:asymmetric-structures-with-CFI graphs}.
	The interpretation actually only removes vertices.
	The claim follows.
\end{proof}

Dawar and Richerby~\cite{DawarRicherby03} asked whether \IFPSC{} is closed under interpretations.
For the global reduct semantics, we can give a negative answer and show that \IFPSC{} is not closed under one-dimensional \FO{}-interpretations.
\begin{cor}
	\IFPSC{} with global reduct semantics is not closed under one-dimensional and equivalence-free \FO{}-interpretations.
\end{cor}
\begin{proof}
	We consider the same structures:
	$\IFP = \IFPSC$ for the constructed structures 
	because they only have trivial orbits and
	it does not make a difference whether we need to witness choices.
	Surely also \IFP{} does not distinguish the structures.
	Because $\IFPWSC{}$ defines the CFI query for ordered base graphs~\cite{GireHoang98}, so does \IFPSC{}.
	We conclude that \IFPSC{} with global reduct semantics is not closed under \FO{}-interpretations, either.
\end{proof}
For the subformula reduct semantics, as used by Dawar and Richerby~\cite{DawarRicherby03}, the question remains open because with subformula reduct semantics one-dimensional and equivalence-free interpretation operators are expressible (cf.~Lemma~\ref{lem:ipfcwsc-subformula-reduct-closed-under-one-dim-interpretations}, the proof also works for non-witnessing symmetric choice).
This shows that, without the interpretation operator,
the subformula reduct semantics is more expressive then the global reduct one:
\begin{cor}
 $\IFPCWSCglob < \IFPCWSCsub$.
\end{cor}
\begin{proof}
	It suffices to show that \IFPCWSCsub{} (the variant of \IFPCWSC{} with subformula reduct semantics) distinguishes the structures from~Theorem~\ref{thm:asymmetric-structures-with-CFI graphs}.
	For this we see that the defining $\Iof{\WSCof{\IFPC}}$~formula in the proof of Theorem~\ref{thm:wsc-le-wsci} is of the form
	$\iop{\interpret}{\formA_{\text{CFI}}}$,
	where $\formA_{\text{CFI}}$ is an $\IFPCWSC${}-formula
	and $\interpret$ a one-dimensional equivalence-free \FO{}-interpretation.
	
	Because \IFPCWSCsub{} is closed under one-dimensional and equivalence-free interpretations by Lemma~\ref{lem:ipfcwsc-subformula-reduct-closed-under-one-dim-interpretations},
	the formula $\iop{\interpret}{\formA_{\text{CFI}}}$ can be expressed in \IFPCWSCsub{},
	which, as seen earlier, is not possible in \IFPCWSCglob{}.
\end{proof}

\section{Discussion}
\label{sec:discussion}

We defined the logics \IFPCWSC{} and \IFPCWSCI{}
to study the combination of witnessed symmetric choice and interpretations
beyond simulating counting.
We provided graph constructions to prove lower bounds for these logics.
\IFPCWSCI{} canonizes CFI graphs if it canonizes the base graphs, but operators have to be nested.
We proved that this increase in nesting depth is unavoidable using double CFI graphs obtained by essentially applying the CFI construction twice.
Does iterating our construction further show an operator nesting hierarchy in \IFPCWSCI{}?
We have seen that also in the presence of counting
the interpretation operator strictly increases the expressiveness.
So indeed both, witnessed symmetric choice and interpretations are needed to possibly capture \PTime{}.
This answers the question of the relation between witnessed symmetric choice and interpretations for \IFPC{}.
But it remains open whether \IFPCWSCI{} captures \PTime{}.
Here, iterating our CFI construction is of interest again:
If one shows an operator nesting hierarchy using this construction,
then one in particular will separate \IFPCWSCI{} from \PTime{}
because our construction does not change the signature of the structures.

Answering this question can also help to compare \IFPCWSCI{} to other logics in the quest for a logic capturing \PTime{}.
One of them is rank logic, which was separated from \PTime{} by a generalization of CFI graphs~\cite{Lichter2023} not only over $\FF_2$ but over arbitrary modulo rings $\ZZ_{2^i} = \ZZ / 2^i \ZZ$.
On the one hand, the separating class is based on ordered base graphs
and hence we expect that \IFPCWSC{} defines this CFI query.
On the other hand, we expect that rank logic defines the CFI query from Theorem~\ref{thm:cfi-wsci-wsci}.
If it is possible to separate \IFPCWSCI{} from \PTime{} by iterating the construction,
it would be interesting to show that rank logic also defines this CFI query,
which would make rank logic and \IFPCWSCI{} incomparable.
Studying all these questions remain for future work.

\bibliographystyle{alphaurl}
\bibliography{ifpc_wsc_i}
\end{document}